\documentclass[11pt,letterpaper]{article}

\usepackage{booktabs}

\newcommand\spacingset[1]{\renewcommand{\baselinestretch}%
  {#1}\small\normalsize}

\usepackage{latexsym}
\usepackage{amssymb,amsmath, bm,pgfplots,tikz,bbm}
\usepackage{graphicx}
\usepackage{marvosym}
\usepackage{multirow,float}
\usepackage{algpseudocode}
\usepackage{caption}
\usepackage{mathtools}
\usepackage{subcaption}
\usepackage{comment}
\usepackage{enumitem}

\usepackage{natbib}
\usepackage{color}
\usepackage[bookmarksopen=true, bookmarksnumbered=true,
pdfstartview=FitH, breaklinks=true, urlbordercolor={0 1 0}, citebordercolor={0 0 1}]{hyperref}

\usepackage{dcolumn}
\newcolumntype{.}{D{.}{.}{-1}}
\newcolumntype{d}[1]{D{.}{.}{#1}}

\usepackage{theorem}
\theoremstyle{plain}
\theoremheaderfont{\scshape}
\newtheorem{theorem}{Theorem}
\newtheorem{proposition}{Proposition}
\newtheorem{assumption}{Assumption}

\newtheorem{lemma}{Lemma}

\newcommand{\qed}{\hfill \ensuremath{\Box}}
\newcommand{\indep}{\mbox{$\perp\!\!\!\perp$}}

\DeclareMathOperator*{\argmin}{argmin}

\newcommand{\norm}[1]{\left\lVert#1\right\rVert}
\newenvironment{proof}{\vspace{1ex}\noindent{\bf Proof}\hspace{0.5em}}
{\hfill\qed\vspace{1ex}}
\usepackage{kantlipsum}
\allowdisplaybreaks

\graphicspath{{Manuscripts_application/figs/}}

\usepackage{rotating}

\usepackage{arydshln}
\usepackage{threeparttable}

\usepackage[compact]{titlesec}

\usepackage[ruled,linesnumbered,vlined]{algorithm2e}
\usepackage{pifont}

\newtheorem{definition}{Definition}

\newcommand\E{\mathbb{E}}
\newcommand\R{\mathbb{R}}
\newcommand\cR{\mathcal{R}}
\renewcommand\P{\mathbb{P}}

\newcommand\bP{\bm{P}}

\newcommand\br{\bm{r}}
\newcommand\bR{\bm{R}}

\newcommand\bh{\bm{h}}

\newcommand\bx{\bm{x}}
\newcommand\bX{\bm{X}}

\newcommand\bW{\bm{W}}

\newcommand\boldf{\bm{f}}
\newcommand\bg{\bm{g}}
\newcommand\bu{\bm{u}}
\newcommand\bU{\bm{U}}
\newcommand\cD{\mathcal{D}}

\newcommand\cS{\mathcal{S}}

\newcommand\cU{\mathcal{U}}
\newcommand\cF{\mathcal{F}}
\newcommand\cX{\mathcal{X}}
\newcommand\cY{\mathcal{Y}}

\newcommand\bgamma{\boldsymbol{\gamma}}
\newcommand\btheta{\boldsymbol{\theta}}

\newcommand\blambda{\boldsymbol{\lambda}}
\usetikzlibrary{decorations.markings}
\usetikzlibrary{decorations.pathmorphing}
\usetikzlibrary{shapes.geometric, arrows}
\usetikzlibrary{arrows,decorations.pathmorphing,backgrounds,positioning,fit,matrix}
\usetikzlibrary{shapes,decorations,arrows,calc,arrows.meta,fit,positioning}
\tikzset{auto,node distance =1 cm and 1 cm,semithick,
	state/.style ={circle, draw, minimum width = 0.7 cm},
	point/.style = {circle, draw, inner sep=0.04cm,fill,node contents={}},
	bidirected/.style={Latex-Latex,dashed},
	el/.style = {inner sep=2pt, align=left, sloped}
}

\graphicspath{{figs/}}

\begin{document}

\title{\bf GenAI Powered Dynamic Causal Inference with\\ Unstructured Data}
\author{Kentaro Nakamura\thanks{Ph.D. student, John F. Kennedy School
    of Government, Harvard University. Email:
    \href{mailto:knakamura@g.harvard.edu}{knakamura@g.harvard.edu}} \and Kosuke Imai\thanks{Professor, Department of Government and
    Department of Statistics, Harvard University, Cambridge, MA
    02138. Phone: 617--384--6778, Email:
    \href{mailto:Imai@Harvard.Edu}{Imai@Harvard.Edu}, URL:
    \href{https://imai.fas.harvard.edu}{https://imai.fas.harvard.edu}} }
\date{\today}
\maketitle\thispagestyle{empty}
%\tableofcontents\thispagestyle{empty}

\begin{abstract}
  A growing number of scholars seek to estimate causal effects of
  unstructured data such as text, images, and video. However, existing
  methods typically treat each object as a single, static
  observation. We develop a statistical framework for dynamic causal
  inference with unstructured data by leveraging generative artificial
  intelligence (GenAI) models. Our approach enables researchers to
  estimate the causal effects of sequences of treatment features,
  including their positions within text and video. We first extract
  internal representations of unstructured objects from a GenAI model
  and then estimate a marginal structural model using a neural network
  architecture that jointly learns a deconfounder for each treatment
  feature in the sequence. Our semiparametric inference framework
  yields valid asymptotic confidence intervals. Simulation studies
  demonstrate that the proposed estimator recovers the target causal
  effects and that the confidence intervals achieve nominal coverage
  in finite samples. We further apply our method to a randomized
  experiment on the Hong Kong protests, showing that the effect of a
  treatment feature depends critically on its position within the
  text.
\end{abstract}


\newpage
\section{Introduction}

Estimating the causal effects of specific features of unstructured
data, such as text, images, and video, is becoming increasingly
important across a range of fields, including political science
\citep[][]{blumenau_variable_2022}, medicine
\citep[]{levine2021randomized}, and education
\citep[e.g.,][]{snell2022text}. In many applications, the central
question concerns how particular features are introduced and how they
shape the responses of those exposed to them.

These questions are often dynamic, as most unstructured data are
inherently sequential: texts unfold sentence by sentence, and videos
frame by frame.  While this structure enables researchers to study
fine-grained, sequence-level interventions, it also poses a
fundamental methodological challenge, since the appearance of later
features may depend on the presence or absence of earlier ones. Much
of the existing causal inference literature treats unstructured data
as single, static observations, thereby overlooking both the
opportunities and the challenges created by their sequential nature
\citep[e.g.,][]{fong_discovery_2016, pryzant_causal_2021,
  fong_causal_2023, gui_causal_2023, guo2025estimating,
  imai2026causal, imai2026leveraging, lin2024isolated, feldman2026causal}.

In this paper, we propose a methodological framework for causal
inference with dynamic interventions in unstructured data that
leverage the power of generative artificial intelligence (GenAI). We
conceptualize each unit as an ordered sequence of components (e.g.,
sentences, frames) and consider interventions that manipulate the
presence of treatment features at each point in the sequence. This
framework goes beyond the standard setup with coarse object-level
treatments and can characterize how causal effects change as
information unfolds over the sequence. The goal is to identify and
estimate the causal effects of different sequences of treatment
features, while controlling the remaining features of unstructured
features.

Under this dynamic setting, we first establish the nonparametric
identification for the causal quantity of interest. The key difficulty
in this setting is that the treatment feature of interest appears
together with many other features over the course of sequence that may
also affect the outcome. These confounding features are
high-dimensional, unknown, and may evolve over time, making standard
adjustment infeasible.

To address this challenge, we extend the GenAI-Powered Inference (GPI)
framework of \cite{imai2026causal, imai2026leveraging} to dynamic settings
under stochastic interventions \citep{kennedy_nonparametric_2019}. We
show that internal representations extracted from an open-source
generative AI model can be leveraged to nonparametrically estimate a
\emph{deconfounder}, which is a low-dimensional summary of confounding
information at each point in the sequence. This approach enables
recovery of the causal quantity of interest by adjusting for evolving
confounding structure, without requiring researchers to directly
observe or manually encode all substantively relevant features of the
unstructured input.

We develop an estimation procedure based on a neural network
architecture tailored to the dynamic GPI framework. The method first
estimates the deconfounder at each position in the sequence, and then
uses flexible machine learning models to estimate both the conditional
expectation of the outcome and the propensity score. Building on
results from the longitudinal causal inference literature
\citep{robi:hern:brum:00, rotnitzky2017multiply,
  kennedy_estimating_2019}, we propose a semiparametric estimation
strategy and establish conditions under which the resulting estimator
is consistent and asymptotically normal.

We validate the proposed dynamic GPI methodology through simulation
studies and an empirical application. The simulations show that the
estimator accurately recovers the target causal effect and that the
associated confidence intervals achieve nominal coverage in finite
samples. In the empirical application, we analyze a randomized
experiment by \cite{fong_causal_2023}, in which respondents are
assigned text vignettes composed of multiple components, with the
treatment feature appearing at different positions within the
sequence. Applying our method, we examine how the treatment effect
varies with the timing of the feature while holding other aspects of
the text fixed. The results indicate that the effect depends on where
the feature appears in the text, highlighting the importance of
accounting for sequential structure when studying causal effects in
unstructured data.

The remainder of this paper is organized as
follows. Section~\ref{sec::example} presents a motivating
application. Section~\ref{sec:method} formalizes the proposed dynamic
GPI methodology, describing the setup, assumptions, identification,
and estimation strategies. Section~\ref{sec:simulation} shows the
results of numerical simulations while Section~\ref{sec:application}
applies the proposed methodology to the real-world example used in the
motivating example. Finally, Section~\ref{sec:conclusion} concludes by
summarizing the advantages and limitations of the proposed dynamic GPI
methodology and outlining directions for future research.

\begin{comment}
This finer level of intervention introduces several key
challenges. First, components that contain the same treatment feature
may differ along many other dimensions, such as wording, visual
context, or surrounding content. For example, in
Table~\ref{tab:hongkong_examples}, different sentences referencing
U.S. commitment vary substantially in their framing and content,
making it necessary to account for such variation in order to isolate
the causal effect of the treatment feature itself. Second, the
sequential structure induces dynamic dependence across components,
possibly violating standard overlap assumptions. In particular, as the
sequence length increases and we condition on the treatment status of
earlier components, the support of later treatments can become highly
constrained—sometimes nearly deterministic, leading to violations of
overlap. We now develop a unified framework that addresses these
challenges.
\end{comment}

\paragraph{Related Literature.}
There is a growing methodological literature on estimating causal
effects with unstructured data (e.g., \citealt{fong_discovery_2016,
  fong_causal_2023, gui_causal_2023, imai2026causal, imai2026leveraging,
  pryzant_causal_2021}). However, existing approaches primarily focus
on static interventions applied to entire unstructured objects,
thereby overlooking the sequential structure inherent in many
unstructured data such as text and video. Building on the GPI
framework \citep{imai2026causal, imai2026leveraging}, we show that the
causal effects of more fine-grained, sequential interventions can be
estimated by leveraging the internal representations of GenAI models.

Our approach also builds on the semiparametric literature on
longitudinal causal inference \citep[e.g.,][]{robi:hern:brum:00,
  rotnitzky2017multiply, kennedy_estimating_2019}. In particular, we
use marginal structural models \citep{robi:hern:brum:00} to formalize
interventions at each stage of the sequence. Unlike standard
applications, we observe the full sequence of unstructured data, so
the assumption of sequential ignorability is satisfied by design,
provided that the selection of unstructured data is unconfounded.

Leveraging recent advances in semiparametric statistics, we develop a
fully nonparametric estimation framework that is well suited to
settings with unstructured data, where the underlying functional form
is unknown. Instead, the primary challenge in our setting is the
failure of positivity, as prior sequences may nearly determine
subsequent treatment assignments. We address this issue by adopting a
stochastic intervention framework, which enables more realistic
interventions under strictly weaker overlap assumptions
\citep{muno:vand:12, kennedy_nonparametric_2019,
  papadogeorgou_causal_2022}.

Finally, our paper contributes to the literature on causal
representation learning. Because unstructured data are
high-dimensional, it is essential to reduce dimensionality while
preserving features relevant for downstream inference
\citep[e.g.,][]{shi_adapting_2019, veitch_adapting_2020,
  wang_desiderata_2022}. We propose a neural network architecture
tailored to sequential settings, producing low-dimensional
representations suitable for causal effect estimation. We formally
show that these representations support identification and consistent
estimation of the causal quantity of interest.

\section{Motivating Example: Hong Kong Experiment}\label{sec::example}

The 2019 Hong Kong protests generated intense international debate
about whether the United States should support the
protesters. Political elites and media coverage framed this issue
using diverse arguments, such as prior U.S. legal commitments, human
rights concerns, and geopolitical threats.  To explore which arguments
are effective, \cite{fong_causal_2023} conducted a randomized
experiment in which they generated a large number of textual arguments
and randomly assigned different combinations of these arguments to
respondents.

Specifically, they fielded two survey waves, one in December 2019
($N = 1{,}983$) and another in October 2020 ($N = 2{,}072$). In each
survey, every respondent was randomly assigned a single vignette
constructed from multiple argument components. These components
included references to U.S. commitments under the Hong Kong Policy
Act, descriptions of protesters' bravery, China’s mistreatment of
citizens, security threats posed by China, and other contextual
arguments. The authors generated many distinct versions of each
component by varying wording, timing, and framing of each component,
yielding over 500,000 possible text combinations, and randomly
composed vignettes by combining two or three such components. After
reading the assigned vignette, each respondent was asked to indicate
their level of agreement with the statement that the United States
should support Hong Kong protesters on a scale ranging from 0 to 100.

\begin{table}[t]
\centering
\small
\renewcommand{\arraystretch}{1.25}

\begin{tabular}{|p{0.9\textwidth}|}
\hline

\textbf{Example Text 1} \\[-0.2em]
%\emph{First sentence:} Protesters waving American flags \\
%\emph{Second sentence:} Bravery in risking physical harm \\[0.3em]
Many protesters carry American flags to demonstrate their love of liberty.
Protesters have been targeted by brutal police tactics and intimidation,
which has led to bloodshed. \\
\hline

\textbf{Example Text 2} \\[-0.2em]
%\emph{First sentence:} U.S. commitment to Hong Kong\\
%\emph{Second sentence:} China’s mistreatment of its citizens \\[0.3em]
\textcolor{red}{During George H.W. Bush's administration, Congress passed a bill which
enshrined America's commitment to preserving Hong Kong's special status
and its freedoms.} China is a communist country with a long history of
mistreating its citizens, including the massacre of protestors at
Tiananmen Square in 1989. \\
\hline

\textbf{Example Text 3} \\[-0.2em]
%\emph{First sentence:} Hong Kong’s political system and economy \\
%\emph{Second sentence:} U.S. commitment to Hong Kong \\[0.3em]
China is interfering with a long tradition of independence for Hong Kong.
Hong Kong has been governed separately from China for the past 150 years.
It is admired around the world for its substantial economic freedom. \textcolor{red}{In
1992, Congress passed the Hong Kong Policy Act, which commits the United
States to preserving Hong Kong's autonomy.} \\
\hline
\end{tabular}
\caption{Three example vignettes from the Hong Kong experiment
  \citep{fong_causal_2023}. The treatment feature (U.S. commitment to
  Hong Kong) is absent in Example~1 and appears in different sentence
  positions in Examples~2~and~3 (highlighted in red).}
\label{tab:hongkong_examples}
\end{table}

While the original study conceptualizes treatment at the level of the
entire text, this aggregation masks potential heterogeneity in how
different parts of the text contribute to persuasion.
Table~\ref{tab:hongkong_examples} presents example texts, illustrating
that the treatment of interest --- U.S. commitment to Hong Kong
through 1992 Hong Kong Policy Act --- can appear at different
positions within a vignette. When the treatment is embedded in a
sequence of sentences, its causal impact may depend not only on its
presence, but also on the timing and context of its
appearance.  Our goal is to develop a methodological framework, under
which we formally define, identify, and estimate these sentence-level
causal effects.

\section{The Proposed Methodology}\label{sec:method}

In this section, we describe the proposed methodology, which extends
the GPI framework of \citet{imai2026causal, imai2026leveraging} to dynamic settings.
Throughout this section, we focus on the case in which text serves as
the treatment, motivated by the example in Section~\ref{sec::example}.
However, the proposed approach applies more broadly to other dynamic
treatment settings, such as audio and video. We begin by introducing the setup
and key assumptions. We then present the nonparametric identification
results, describe the estimation strategy, and derive the asymptotic
properties of the proposed causal effect estimator.

\subsection{Setup}

Consider a simple random sample of $N$ respondents drawn from a
population of interest. For each individual indexed by
$i = 1, \ldots, N$, we randomly assign unstructured data consisting of
$S_i$ segments where $S_i \in \cS := \{1,\ldots,s_{\max}\}$ for some
positive integer $s_{\max} > 1$. Here, $S_i$ denotes the total number
of segments in the data assigned to individual $i$, where researchers
may adopt different segment definitions across applications. For
example, when text data are segmented at the sentence level, $S_i$
corresponds to the total number of sentences in the document assigned
to respondent $i$. Alternatively, one may choose paragraphs as
segments. Let $\bX_{is} \in \cX_s$ represent the $s$th segment of
unstructured data assigned to individual $i$, where $\cX_s$ denotes
the support of the $s$th segment $\bX_{is}$.

We consider the following experimental design. Each respondent $i$ is
presented with unstructured data $\bX_i \in \cX$ where $\cX$ is the
support of $\bX_i$. After observing the entire content of $\bX_i$, the
respondent provides an answer to a survey question. Let $Y_i \in \cY$ denote the observed outcome for respondent $i$, where the support $\cY \subset \mathbb{R}$ is bounded.  We use
$Y_i(\bx_1, \ldots, \bx_{s_{\max}})$ to represent the potential
outcome that would be realized if the respondent were exposed to the
unstructured data
$(\bX_{i1}, \ldots, \bX_{i s_{\max}}) = (\bx_1, \ldots,
\bx_{s_{\max}})$.  Note that if an assigned document has fewer
segments than $s_{\max}$, then we set $\bX_{is} = \emptyset$ for all $s$
with $S_i < s \le s_{\max}$.  Therefore, we often write
$Y_i(\bX_{i1},\ldots,\bX_{iS_i}, \ldots, \bX_{i
  s_{\max}})=Y_i(\bX_{i1}, \ldots, \bX_{iS_i})$, ignoring the empty
segments of the assigned document. We assume that the observed outcome
is a function of the assigned unstructured data, i.e.,
$Y_i = Y_i(\bX_{i1}, \ldots, \bX_{iS_i})$. This assumption rules out
spillover effects across respondents. Formally, our setup implies the
following two assumptions.

\begin{assumption}[Consistency]\label{consistency}
  For each segment $s=1,\ldots,S_i$, the observed outcome $Y_{i}$
  equals the potential outcome under the realized texts
  $\bX_{i1}, \ldots, \bX_{iS_i}$, i.e.,
  $$
  Y_{i} =  Y_{i}(\bX_{i1}, \ldots, \bX_{iS_i}).
  $$
\end{assumption}

\begin{assumption}[Random Assignment of Unstructured 
  Data]\label{randomization} The unstructured data $\bX_i$ is
  randomly assigned to each individual respondent so that
$$
   Y_{i}(\bx_1, \ldots, \bx_{s_{\max}}) \ \indep \ \bX_{i} 
$$
for all $(\bx_1 , \ldots, \bx_{s_{\max}})$ given any $i=1,\ldots,N$.
\end{assumption}
Assumption~\ref{randomization} implies the following conditional
independence,
$$
Y_{i}(\bx_1, \ldots, \bx_{s_{\max}}) \ \indep \ \bX_{is} \mid \bX_{i1},
\ldots, \bX_{i,s-1}
$$
for each $s=1,\ldots,S_i$ where we set $\bX_{i0}=\emptyset$.
    
We are interested in the causal effect of a sequence of features on an
outcome.  For simplicity, we consider a binary treatment feature
denoted by $W_{is}$ for each unstructured data $i$ and segment $s$.
We assume that the treatment feature is a deterministic function of
each segment of the unstructured data.
\begin{assumption}[Treatment feature]\label{treatment_feature}
  There exists a deterministic function $g_{W}: \cX_s \mapsto \mathcal{W} := \{0,1\}$
  that maps each non-empty segment of the unstructured data
  $\bX_{is} \in \cX_s$ to a binary treatment feature of interest,
  i.e.,
$$
W_{is} = g_{W}(\bX_{is}).
$$
\end{assumption}

Next, we define the confounding features, which represent all features
of $\bX_{is}$ (other than the treatment feature of interest $W_{is}$) that influence the outcome and future confounding features.  These confounding features, denoted by
$\bU_{is}$, are assumed to be represented by a lower-dimensional
vector-valued deterministic function of $\bX_{is}$ and are denoted by
$\bU_{is} \in \cU_s$, where $\cU_s$ is their support at segment
$s$. Unlike the treatment feature of interest, these confounding
features are not observed and must be learned from the data.

\begin{assumption}[Confounding features]\label{confounding_feature}
  There exists an unknown vector-valued deterministic function
  $\bg_{\bU_s}: \cX_s \mapsto \cU_s$ that maps each non-empty segment
  $s$ of unstructured data $\bX_{is} \in \cX_s$ to the confounding
  features $\bU_{is} \in \cU_s$, i.e.,
\begin{align*}
    \bU_{is} = \bg_{\bU_s}(\bX_{is})
\end{align*}
where $\mathrm{dim}(\bU_{is}) \ll \mathrm{dim}(\bX_{is})$. 
\end{assumption}

We now propose our key identification assumption.  We assume that it
is possible to intervene in the treatment feature of each segment
without altering the confounding features at the same segment. Specifically, confounding features should not be any function of
treatment features to avoid post-treatment bias. We refer to this
assumption as {\it sequential separability}.
\begin{assumption}[Sequential Separability]\label{separability} The
  following equalities hold:
\begin{align*}
    Y_{i}(\bx_1, \ldots, \bx_{S_i})
    %&= Y_{i}(W_{i1} = w_1, \bU_{i1} = \bu_1, \ldots, W_{i S_i} = w_{S_i},
    %  \bU_{i S_i} = \bu_{S_i})\\
    & = Y_{i}(g_W(\bx_1), \bg_{\bU_1}(\bx_1), \ldots, g_W(\bx_{S_i}), \bg_{\bU_{S_i}}(\bx_{S_i}) ),\\
    \bU_{is}(\bx_1, \bx_2 \ldots, \bx_{s-1}) & = \bU_{is}(g_W(\bx_1), \bg_{\bU_1}(\bx_1), \ldots, g_W(\bx_{s-1}), \bg_{\bU_{s-1}}(\bx_{s-1})).
\end{align*}
for any $s \in \mathcal{S}$. In addition, 
there exist no deterministic functions $\bg^\prime: \cX_s \to \cX_s^\prime$ and $\tilde{\bg}_{\bU_s}: \{0,1\} \times \cX_s^\prime \to \cU_s$, which satisfy $\bg_{\bU_s}(\bx_s) = \tilde{\bg}_{\bU_s} ( g_W(\bx_s), \bg^\prime(\bx_s))$ for all $\bx_s \in \cX_s$ and $\tilde{\bg}_{\bU_s} ( 1, \bg^\prime(\bx_s^\prime))\ne \tilde{\bg}_{\bU_s} (0, \bg^\prime(\bx_s^\prime))$ for some $\bx_s^\prime \in \cX_s$.
\end{assumption}

Assumption~\ref{separability} has two parts. First, we assume that the sequence of unstructured data affects the outcome and the evolution of
subsequent confounding features only through the corresponding treatment histories and confounding histories. Second, the assumption states that within each segment, one can intervene on the treatment feature $W_{is}$ without changing the {\it contemporaneous} confounding features $\bU_{is}$. Importantly, we do not rule out the possibility that such an intervention affects future treatment and confounding features.

Sequential separability clarifies how researchers should define the treatment and confounding features: defining a treatment requires specifying which aspects of the unstructured object the intervention changes and which are held fixed. In our Hong Kong application, the treatment is whether a sentence contains an argument about the U.S. commitment to Hong Kong; features unaffected by adding or removing this argument, such as surrounding context or unrelated wording, may act as confounders. By contrast, if the treatment were defined as a passage's ideological orientation, changing it would necessarily alter rhetoric, framing, and substantive content. These induced changes cannot be held fixed as confounders; they are part of the intervention and contribute to its causal effect. Sequential separability thus requires defining the treatment broadly enough to include features that necessarily change under the intervention, reserving the confounder role for outcome-relevant features that can conceptually remain fixed.

An important implication of sequential separability is that the treatment and confounding features must have independent support \citep{wang_desiderata_2022}. As discussed above, researchers can often define the treatment broadly enough so that the remaining confounding features are conceptually separable from it. In practice, however, these confounding features are unobserved and must ultimately be learned from the data. A learned representation of the confounding features may therefore fail to preserve this separability. The following proposition provides a necessary implication of sequential separability that will later allow us to diagnose such violations.

\begin{proposition}[Sequential separability implies independent support]
\label{prop:seq_independent_support}
Suppose Assumption~\ref{separability} holds. Fix any
$1 \le s \le s' \le s_{\max}$. Define the joint support induced by the $s$th text segment as
\[
  \mathcal S_{WU,s}
  :=
  \left\{
  \left(g_W(\bx), \bg_{\bU_s}(\bx)\right):
  \bx \in \cX_s
  \right\}.
\]
Then, we have,
\[
  \mathcal S_{WU,s}
  =
  \mathcal W_s \times \mathcal U_s.
\]
\end{proposition}
The proof is in Appendix~\ref{proof:seq_independent_support}, which is almost identical to Proposition 3 of \cite{imai2026leveraging}. Proposition~\ref{prop:seq_independent_support} establishes that, under
sequential separability, any feasible treatment value and any feasible
confounder value can jointly occur within a given segment. This is
precisely the support condition required for positivity conditional on
the confounding features, as it ensures that both treatment levels are
attainable for every confounder value.

\subsection{Deep Generative Models}

Following the general strategy of GPI methodology, we use a deep
generative model to generate unstructured data $\bX_{i}$. If we are
interested in analyzing existing unstructured data rather than
creating new ones, we reproduce them by appropriately prompting a
deep generative model.  We adopt the general definition of deep
generative models \citep{imai2026causal}.
\begin{definition}[Deep Generative Model]\label{deep_use} \spacingset{1} A deep generative model is the following probabilistic model that takes prompt $\bP_{i}$ as an input and generates the treatment object $\bX_{i}$ as an output:
$$
\begin{aligned}
&\P(\bX_{i}  \mid \bh_{\bgamma}(\bR_{i}))\\
&\P(\bR_{i} \mid \bP_{i})
\end{aligned}
$$
where $\bR_{i} \in \cR$ denotes an observable
internal representation of $\bX_{i}$ contained in the model and
$\bh_{\bgamma}(\bR_{i})$ is a deterministic function parameterized by
$\bgamma$ that completely characterizes the conditional distribution
of $\bX_{i}$ given $\bR_{i}$.
\end{definition}
Under this definition of deep generative model, $\bR_{i}$ represents a
lower-dimensional representation of $\bX_{i}$ and is a hidden
representation of the neural network. This definition encompasses many
models for texts, images, and videos. Importantly, even when the
unstructured data of interest are not originally generated by a deep
generative model, researchers can still obtain such internal
representations by prompting a deep generative model to reproduce the
same object (or a close semantic equivalent) and extracting the
associated hidden representation $\bR_{i}$.

We assume that each segment of the unstructured data $\bX_{is}$ is
solely generated from the corresponding hidden states, which we denote
by $\bR_{is}$. We also assume that the last layer of the deep
generative model is a deterministic function of $\bR_i$. Under these
assumptions, the relevant low-dimensional features of the unstructured
object $(\bW_{is}, \bU_{is})$ can be regarded as a deterministic
function of the low-dimensional internal representation $\bR_{is}$ of
the deep generative model. We formalize these assumptions as follows.

\begin{assumption}[Factorized Deterministic Decoding]\label{det_dec} The output
  layer of a deep generative model is deterministic for a given text
  segment. That is, given the number of segments $S_i$,
\begin{align*}
   \P(\bX_{i} \mid \bh_{\bgamma}(\bR_i)) = \prod_{s = 1}^{S_i} \P(\bX_{is} \mid \bh_{\bgamma}(\bR_{is})) 
\end{align*}
and $\P(\bX_{is} \mid \bR_{is})$ is degenerate for any respondent $i$ and segment $s$.
\end{assumption}
Many deep generative models for texts satisfy this
assumption. Specifically, many decoder-only models use an
autoregressive architecture with a linear decoding layer (e.g.,
\citealt{vaswani_attention_2017, radford2018improving}). Consequently,
the vocabulary produced at segment $s$, $\bX_{is}$, is solely a
function of the hidden state for that segment, $\bR_{is}$. These
models also typically allow deterministic decoding by, for example,
setting the temperature parameter to zero.

\begin{figure}[t!]
\centering

\begin{tikzpicture}[
  every node/.style={font=\small},
  >=Latex,
  var/.style={
    circle,
    draw,
    fill=white,
    inner sep=0pt,
    minimum size=7.2mm
  },
  det/.style={
    -{Latex[length=1.8mm,width=1.3mm]},
    draw=red,
    double=white,
    double distance=1.0pt,
    line width=0.52pt,
    shorten <=1pt,
    shorten >=1pt
  },
  stoch/.style={
    -{Latex[length=1.8mm,width=1.3mm]},
    draw=black,
    line width=0.58pt,
    shorten <=1.1pt,
    shorten >=1.1pt
  },
  stage/.style={
    draw=red,
    dashed,
    rounded corners=2pt,
    line width=0.75pt,
    inner sep=8pt
  }
]

% ========================================================
% Segment 1
% ========================================================

\node[var] (R1) at (0,4.55) {$\bR_1$};
\node[var] (X1) at (0,3.30) {$\bX_1$};

\node[var] (W1) at (0,1.75) {$W_1$};
\node[var] (U1) at (0,0.45) {$\bU_1$};
\node[var] (V1) at (0,-0.85) {$\bm V_1$};

% ========================================================
% Segment 2
% ========================================================

\node[var] (R2) at (4.60,4.55) {$\bR_2$};
\node[var] (X2) at (4.60,3.30) {$\bX_2$};

\node[var] (W2) at (4.60,1.75) {$W_2$};
\node[var] (U2) at (4.60,0.45) {$\bU_2$};
\node[var] (V2) at (4.60,-0.85) {$\bm V_2$};

% ========================================================
% Outcome
% ========================================================

\node[var] (Y) at (8.25,0.45) {$Y$};

% ========================================================
% Dashed boxes for the decomposed segment features
% ========================================================

\begin{scope}[on background layer]

\node[
  stage,
  fit=(W1)(U1)(V1)
] (stage1) {};

\node[
  stage,
  fit=(W2)(U2)(V2)
] (stage2) {};

\end{scope}

% ========================================================
% Deterministic relations
% ========================================================

\draw[det] (R1) -- (X1);
\draw[det] (X1.south) -- (stage1.north);

\draw[det] (R2) -- (X2);
\draw[det] (X2.south) -- (stage2.north);

% ========================================================
% Segment 1 -> Segment 2
%
% W_1, U_1, V_1 may affect W_2 and V_2.
% Only W_1 and U_1 may directly affect U_2.
% ========================================================

% Same-level arrows
\draw[stoch] (W1.east) -- (W2.west);
\draw[stoch] (U1.east) -- (U2.west);
\draw[stoch] (V1.east) -- (V2.west);

% W_1 -> U_2 and V_2
\draw[stoch] (W1.345) -- (U2.165);
\draw[stoch] (W1.330) -- (V2.150);

% U_1 -> W_2 and V_2
\draw[stoch] (U1.15) -- (W2.195);
\draw[stoch] (U1.345) -- (V2.165);
\draw[stoch] (W1) to (Y);
\draw[stoch] (U1) to[bend right=15] (Y);

% V_1 -> W_2
% No V_1 -> U_2 under sequential separability
\draw[stoch] (V1.30) -- (W2.210);

% ========================================================
% Segment 2 -> outcome
% ========================================================

\draw[stoch]
  (W2.east)
  to
  (Y.150);

\draw[stoch]
  (U2.east)
  -- (Y.west);

\end{tikzpicture}

\caption{
Directed acyclic graph of the assumed data-generating process when
$S_i=2$. Red double-line arrows denote deterministic relations, while
black single-line arrows denote possibly stochastic relations. The
red dashed rectangles represent the decomposition of each segment
$\bX_s$ into the treatment feature $W_s$, confounding features
$\bU_s$, and remaining features $\bm V_s$.
}
\label{dag}

\end{figure}

Figure~\ref{dag} presents a directed acyclic graph (DAG) that
summarizes the data generating process and all the assumptions
described above. In this DAG, an arrow with red double lines
represents a deterministic causal relation while an arrow with a
single black line represents a possibly stochastic causal
relation. 
For clarity, DAG splits $\bX_s$ into three parts; treatment feature $W_s$, confounding features $\bU_s$, and the remaining features $\bm{V}_s$.
Under Assumption~\ref{separability}, the treatment and confounding features capture all the relevant features of $\bX_{is}$ that influence the outcome and future confounding features.
Therefore, there is no direct arrow from $\bm{V}_s$ to $\bU_{s+1}$ or to $Y$ although $\bm{V}_s$ may affect $W_{s+1}$ and $\bm{V}_{s+1}$.
This implies the following conditional independence relations,
\begin{align*}
    Y_i \ &\indep \ \{\bR_{i1}, \cdots, \bR_{is}\} \mid  W_{i1}, \cdots, W_{is}, \bU_{i1}, \cdots, \bU_{is}, S_i = s \quad \\
    \bU_{i,s'+1} \ &\indep \ \{\bR_{i1}, \cdots, \bR_{is'}\} \mid W_{i1}, \cdots, W_{is'}, \bU_{i1}, \cdots, \bU_{is'}, S_i = s
\end{align*}
for any $s \in \mathcal{S}$ and $s' < s$. 

\subsection{Causal estimand under dynamic stochastic
  interventions} \label{sec::estimand}

Under this setup, we are interested in the average potential outcome
under a dynamic stochastic intervention. We adopt a stochastic
intervention because a deterministic intervention is likely to suffer
from severe violation of the positivity assumption in our
high-dimensional dynamic settings.  Stochastic interventions mitigate
this overlap problem by generating counterfactual treatment paths from
a distribution that remains close to that of the observed treatment
paths \citep[e.g.,][]{munoz2012population,
  kennedy_nonparametric_2019,papadogeorgou_causal_2022}.
Specifically, we consider the following dynamic stochastic
intervention.

\begin{definition}[Dynamic Stochastic
  intervention] \label{def:stochastic} For each segment
  $s=1,\ldots,S_i$, we consider the following
  dynamic stochastic intervention
$$
q_s(\delta_s; \overline{\bm{w}}_{s-1}) = 
\frac{\delta_s p_s(\overline{\bm{w}}_{s-1}) }{\delta_s
  p_s(\overline{\bm{w}}_{s-1}) + 1 -
  p_s(\overline{\bm{w}}_{s-1})}, 
$$
where
$p_s(\overline{\bm{w}}_{s-1}) = \P(W_{is} = 1 \mid
\overline{\bm{W}}_{i,s-1} = \overline{\bm{w}}_{s-1}, S_i \geq s)$ denotes the
observed treatment assignment probability. Note that
$\overline{\bm{W}}_{i,s-1} = (W_{i1}, \ldots, W_{i,s-1})$ represents
the past treatment history, whose value is given by the vector
$\overline{\bm{w}}_{s-1} = (\bm{w}_1, \ldots, \bm{w}_{s-1})$ for
$s > 1$ with $\overline{\bm{W}}_{i0} = \emptyset$.
\end{definition}
The user-specified incremental parameter $\delta_s \in (0, \infty)$
determines the extent to which the counterfactual treatment assignment
probability diverges from the observed one at segment $s$.
Specifically, this parameter equals the following odds ratio,
\begin{align*}
    \delta_s \ = \ \frac{q_s(\delta_s; \overline{\bm{w}}_{s-1}) }{1 - q_s(\delta_s; \overline{\bm{w}}_{s-1}) } \cdot \frac {1 - p_s(\overline{\bm{w}}_{s-1}) }{p_s(\overline{\bm{w}}_{s-1}) }.
\end{align*}
Thus, a greater value of $\delta_s$ corresponds to a policy that assigns
treatment with a greater probability. For generality, we allow $\delta_s$ to vary across time $s$. Our stochastic intervention
builds on the incremental propensity score intervention of
\cite{kennedy_nonparametric_2019}. Unlike the original formulation,
however, our proposed intervention depends only on the observed
treatment history, since the confounders $\bU_{is}$ must be estimated
in our setting. We therefore do not condition on $\bU_{is}$, as doing
so would not yield a uniquely identified estimand.

Consequently, in contrast to \cite{kennedy_nonparametric_2019}, our
approach requires the weak overlap assumption, which states that all
possible treatment assignments under stochastic intervention can be
realized. Specifically, the ratio of the propensity score over the
corresponding treatment assignment probability under the stochastic
intervention must be bounded away from zero. Formally, let
$\overline{\bm{U}}_{is} = (\bU_{i1}, \ldots, \bU_{is})$ denote the
history of confounding features up to segment $s \ge 1$.  Then, we can
define the propensity score at segment $s$ as
$\pi_s(\overline{\bm{W}}_{i,s-1}, \overline{\bm{U}}_{is}, S_i) :=
\P(W_{is} = 1 \mid \overline{\bm{W}}_{i,s-1},
\overline{\bm{U}}_{is}, S_i)$, which represents the conditional
probability of treatment assignment given the treatment and confounder
histories up to segment $s$. We formalize this assumption as
follows.

\begin{assumption}[Bounded relative overlap]\label{overlap} There exists a constant $c > 0$ such that 
\begin{align*}
\pi_s(\overline{\bm{w}}_{s-1}, \overline{\bm{u}}_{s}, \ell) \ge c \cdot p_s(\overline{\bm{w}}_{s-1}) \quad \text{and} \quad  1 - \pi_s(\overline{\bm{w}}_{s-1}, \overline{\bm{u}}_{s}, \ell) \ge c \cdot (1 -  p_s(\overline{\bm{w}}_{s-1}))
\end{align*}
for all $s=1,\ldots,S_i$, $s \leq \ell$, $i=1,\ldots,N$,
$\overline{\bm{w}}_{s-1} = (\bm{w}_1, \ldots, \bm{w}_{s-1}) \in
\mathcal{W}^{s-1}$, and
$\overline{\bm{u}}_{s} = (\bm{u}_1, \ldots, \bm{u}_{s}) \in
\mathcal{U}^{(s)}$ where
$\mathcal{U}^{(s)} = \prod_{s^\prime = 1}^{s}
\cU_{s^\prime}$. 
\end{assumption}

Importantly, Assumption~\ref{overlap} is weaker than the standard
positivity assumption that the propensity score itself is bounded away
from zero since the overlap is defined relative to the stochastic
intervention.  In particular, whenever
$p_s(\overline{\bm{W}}_{i,s-1}) = 0$, this intervention allows
the true propensity score to be zero
($p_s(\overline{\bm{W}}_{i,s-1}) = 0$ implies
$\pi_s(\overline{\bm{W}}_{i,s-1}, \overline{\bm{U}}_{is}, S_i) = 0$ by
the law of iterated expectation).  This relaxation allows us to
consider realistic interventions even when the total number of
segments $S_i$ is large.

As the causal quantity of interest, we consider the average outcome
under a specified stochastic intervention, which is formally defined
as,
\begin{align}
    \Psi(\bm{\delta}) := \E\biggl[ \int_{\mathcal{W}^{S_i}} Y_i(W_{i1} = w_1,
  \ldots, W_{iS_i} = w_{S_i}) 
  \prod_{s= 1}^{S_i} d Q_s(w_s; \overline{\bm{w}}_{s-1}, \delta_s)\biggr], \label{target_quantity}
\end{align}
where $\bm\delta = (\delta_1, \cdots, \delta_{s_{\max}})$ and
$$
dQ_s(w_s; \overline{\bm{w}}_{s-1}, \delta_s) = \frac{w_s \delta_s
  p_s(\overline{\bm{w}}_{s-1}) + (1 - w_s) (1 -
  p_s(\overline{\bm{w}}_{s-1}) ) }{\delta_s p_s(\overline{\bm{w}}_{s-1}
  )+ 1 - p_s(\overline{\bm{w}}_{s-1})},
$$
where we omit confounding features, which are possibly affected by past treatment features, from the potential outcome for notational simplicity. Throughout the paper, we use the following notation
\begin{align*}
    \bm\delta_{\underline{s}} \ := \ \bm\delta_{(s+1):s_{\max}}, \qquad
    \bm\delta_{\overline{s}} \ := \ \bm\delta_{1:s}, \qquad
    \bm\delta \ = \ (\bm\delta_{\overline{s}}, \bm\delta_{\underline{s}}), \qquad
    \bm\delta_{-s} \ := \ (\delta_1, \ldots, \delta_{s-1}, \delta_{s+1}, \ldots, \delta_{s_{\max}}).
\end{align*}

The estimand $\Psi(\bm\delta)$ given in Equation~\eqref{target_quantity}
represents the expected outcome when the sequence of treatment
features is generated according to the proposed stochastic
intervention indexed by the incremental parameter $\bm \delta$, while all
confounding features evolve according to their counterfactual distributions under the stochastic intervention. The inner integral averages the potential
outcome over all possible treatment paths, weighted by the stochastic
intervention that shifts the treatment assignment probability toward
one as $\delta_s$ increases. The outer expectation marginalizes over the
distribution of confounding features, the potential outcomes, and the
number of text segments $S_i$.

\subsection{Nonparametric identification}

Given the above setup, we establish the nonparametric identification
of the causal estimand defined in
Equation~\eqref{target_quantity}. Recall that, under
Assumptions~\ref{treatment_feature},
\ref{confounding_feature}, and~\ref{det_dec}, both the
treatment and confounding features for each text segment are
deterministic functions of the corresponding hidden state. As a
result, conditioning on the full history of latent confounding
features up to text segment $s$ suffices to block all backdoor paths
between treatment and outcome (see Figure~\ref{dag}).

Although these confounding features are unknown and must be inferred
from the data, identification does not require recovering the true
confounding features exactly. 
Instead, we adjust for a deconfounder
function $\boldf_s: \mathcal{R} \mapsto \mathcal{F} \subset
\R^{D_f}$ that retain the information
needed for identification. To formalize the idea, we first define the outcome model given the treatment and deconfounder history as
\begin{align*}
\mu
\left(
\overline{\bm w}_{S_i},
\boldf_1(\bR_{i1}),\ldots,\boldf_{S_i}(\bR_{iS_i}), S_i
\right)
&:=
\E\left[
Y_i
\mid
\overline{\bW}_{iS_i}=\overline{\bm w}_{S_i},
\boldf_1(\bR_{i1}),\ldots,\boldf_{S_i}(\bR_{iS_i}), S_i
\right].
\end{align*}
We also define the conditional mean function under the stochastic intervention,
\begin{equation}
\begin{aligned}
&m_{s}
\left(
\overline{\bm w}_{s},
\boldf_1(\bR_{i1}),\ldots,\boldf_{s}(\bR_{is}), S_i;
\bm\delta_{\underline{s}}
\right) \\
:= \ & \int_{\mathcal R^{S_i-s}} \int_{\mathcal W^{S_i-s}} \ \mu
\left(
\overline{\bm w}_{S_i},
\boldf_1(\bR_{i1}),\ldots,\boldf_{S_i}(\bR_{iS_i}), S_i
\right)\\
&\qquad \qquad \times  \prod_{t = s+1}^{S_i} d Q_{t}
(w_{t};\overline{\bm w}_{t-1},\delta_{t}) dF(\bR_{it} \mid \overline{\bm{W}}_{i,t-1} = \overline{\bm{w}}_{t-1},  \boldf_1(\bR_{i1}), \cdots, \boldf_{t-1}(\bR_{i,t-1}), S_i)
%\E\left[
%\int_{\mathcal{W}}
%m_{s'+1}
%\left(
%\overline{\bm w}_{s'+1},
%\tilde \boldf_1(\bR_{i1}),\ldots,
%\tilde \boldf_{s'+1}(\bR_{i,s'+1});
%\bm\delta_{(s'+2):s}
%\right)
%dQ_{s'+1}
%(w_{s'+1};\overline{\bm W}_{is'},\delta_{s'+1})
%\right.\\
%&\hspace{3in}
%\Bigr | \ 
%\overline{\bW}_{is'}=\overline{\bm w}_{s'},
%\boldf_1(\bR_{i1}),\ldots,
%\boldf_{s'}(\bR_{i,s'}),
%S_i=s \biggr ].
\end{aligned} \label{eq:recursive_m}
\end{equation}
for $s=0,\ldots,S_i-1$ and $m_{S_i} = \mu$. Specifically, $m_{s}$ represents the expected outcome when all future treatments from the ($s+1$)th segment onward are generated according to the stochastic intervention. Thus, if the functions $\boldf_s(\bR_{is})$ capture the same information as the true confounding features $\bU_{is}$, averaging the resulting $m_0$ over the distribution of $S_i$
yields the target estimand in Equation~\eqref{target_quantity}.

Following \citet{imai2026causal, imai2026leveraging}, we characterize the deconfounder through mean independence, which ensures that it retains the confounding information relevant for the outcome. In the dynamic setting, this condition must hold not only for the terminal outcome regression but also for each intermediate outcome model arising in the backward recursion above. The following definition extends the notion of a deconfounder to this dynamic setting.

\begin{definition}[Deconfounder]\label{def:deconfounder}
A collection of functions
$\boldf_s:\mathcal{R}\mapsto\mathcal{F}\subset\R^{D_f}$,
$s=1,\ldots,s_{\max}$, is a deconfounder if it satisfies the following
mean-independence conditions:
\begin{align}
&\E\left[
Y_i
\mid
\boldf_1(\bR_{i1}),\ldots,\boldf_s(\bR_{is}),
\overline{\bW}_{is},
S_i=s
\right]
=
\E\left[
Y_i
\mid
\boldf_1(\bR_{i1}),\ldots,\boldf_s(\bR_{is}),
\overline{\bW}_{is},
\overline{\bR}_{is},
S_i=s
\right],
\label{deconfounder_terminal_independence}
\end{align}
and, for any pair of deconfounder functions $\boldf_s$ and
$\tilde{\boldf}_s$ (which may coincide so that
$\boldf_s=\tilde{\boldf}_s$),
\begin{align}
&\E\left[m_{s'+1}
\left(
\overline{\bm w}_{s'+1},
\tilde \boldf_1(\bR_{i1}),\ldots,
\tilde \boldf_{s'+1}(\bR_{i,s'+1}), s;
\bm\delta_{\underline{s'+1}}
\right) \ \middle | \
\boldf_1(\bR_{i1}),\ldots,
\boldf_{s'}(\bR_{is'}),
\overline{\bW}_{is'},
S_i=s
\right] \nonumber
\\
= \ &
\E\left[
m_{s'+1}
\left(
\overline{\bm w}_{s'+1},
\tilde \boldf_1(\bR_{i1}),\ldots,
\tilde \boldf_{s'+1}(\bR_{i,s'+1}), s;
\bm\delta_{\underline{s'+1}}
\right) \ \middle| \
\boldf_1(\bR_{i1}),\ldots,
\boldf_{s'}(\bR_{is'}),
\overline{\bR}_{is'},
\overline{\bW}_{is'},
S_i=s
\right],
\label{deconfounder_independence}
\end{align}
for any $s\in\cS$ with $s>1$, $s'=1,\ldots,s-1$,
$\overline{\bm w}_{s'}\in\mathcal W^{s'}$, and any value of the
incremental parameter $\bm\delta\in\R_+^{s_{\max}}$, where
$\overline{\bR}_{is'}=(\bR_{i1},\ldots,\bR_{is'})$ and
$\overline{\bR}_{i0}=\emptyset$. 
%Here, $m_{s+1}$ denotes the outcome models, which are recursively defined as
%\begin{align*}
%m_s
%\left(
%\overline{\bm w}_s,
%\boldf_1(\bR_{i1}),\ldots,\boldf_s(\bR_{is}); \bm\delta_{\underline{s}}
%\right)
%&:=
%\E\left[
%Y_i
%\mid
%\overline{\bW}_{is}=\overline{\bm w}_s,
%\boldf_1(\bR_{i1}),\ldots,\boldf_s(\bR_{is}),
%S_i=s
%\right],
%\end{align*}
%and, for $s'=1,\ldots,s-1$, 
%\begin{equation*}
%\begin{aligned}
%&m_{s'}
%\left(
%\overline{\bm w}_{s'},
%\boldf_1(\bR_{i1}),\ldots,\boldf_{s'}(\bR_{is'});
%\bm\delta_{\underline{s'+1}}
%\right)\\
%&:=
%\E\left[
%\int_{\mathcal{W}}
%m_{s'+1}
%\left(
%\overline{\bm w}_{s'+1},
%\tilde \boldf_1(\bR_{i1}),\ldots,
%\tilde \boldf_{s'+1}(\bR_{i,s'+1});
%\bm\delta_{(s'+2):s}
%\right)
%dQ_{s'+1}
%(w_{s'+1};\overline{\bm W}_{is'},\delta_{s'+1})
%\right.\\
%&\hspace{3in}
%\Bigr | \ 
%\overline{\bW}_{is'}=\overline{\bm w}_{s'},
%\boldf_1(\bR_{i1}),\ldots,
%\boldf_{s'}(\bR_{i,s'}),
%S_i=s \biggr ].
%\end{aligned}
%\end{equation*}
\end{definition}
Importantly, the deconfounder for a given text segment may differ across
outcome regressions in the backward recursion, as indicated by the
distinction between $\boldf_s$ and $\tilde{\boldf}_s$ above.
This flexibility is useful for estimation because a representation
learned to satisfy the mean-independence condition for one regression
need not satisfy the corresponding condition for another.

The existence of deconfounder function is guaranteed under Assumption~\ref{confounding_feature} since confounding features themselves satisfy the same mean independence condition.  In addition, although there may exist multiple deconfounders, we show that any low-dimensional
representation that satisfies the conditions given in Definition~\ref{def:deconfounder} yields the same identification
formula. Theorem~\ref{identification} formalizes this result. 

\begin{theorem}[Nonparametric Identification]\label{identification}
Under Assumptions~\ref{consistency}--\ref{overlap}, let
$\boldf_s:\cR\mapsto\cF\subset\R^{D_f}$,
$s=1,\ldots,s_{\max}$, be deconfounder functions satisfying
Definition~\ref{def:deconfounder} and bounded relative overlap. Then,
by adjusting for such deconfounders, we can nonparametrically identify
the average outcome under the stochastic intervention defined in
Equation~\eqref{target_quantity} as
\begin{align*}
\Psi(\bm\delta)
& \ = \  \E\left[  \int_{\cR^{S_i}} \int_{\mathcal{W}^{S_i}}
                \E\left[Y_{i} \mid \overline{\bm{w}}_{S_i}, \
                \{\boldf_{s}^{[S_i]}(\bR_{is})\}_{s=1}^{S_i}, \ S_i\right]
                \right. \\
& \left. \hspace{1in} \times
\prod_{s = 1}^{S_i}
\biggl\{
d Q_s(w_{s}; \overline{\bm{w}}_{s-1}, \delta_s)
dF(\bR_{is} \mid
\overline{\bm{W}}_{i,s-1} = \overline{\bm{w}}_{s-1},
\{\boldf_{s'}^{[s-1]}(\bR_{is'})\}_{s'=1}^{s-1},
S_i)
\biggr\} \right],
\end{align*}
where the superscript $[t]$ in $\boldf_{s}^{[t]}$ indicates that we may have
$\boldf_{s}^{[t]}(\bR_{is}) \neq
\boldf_{s}^{[t^\prime]}(\bR_{is})$
with $t \neq t^\prime$. 
\end{theorem}
The proof is in Appendix~\ref{proof_identification}.

\subsection{Estimation and Inference}

Given the identification result, we next consider estimation and
inference. In the cross-section settings, the deconfounder can be
learned by using the neural network architecture that encodes the
required independence relationship \citep{imai2026causal}. However,
the challenge in the dynamic settings is that we need to learn the
deconfounder $\boldf_{s}^{[t]}$ for the $t$-th outcome regression of each segment $s \in \cS$, while satisfying
Equation~\eqref{deconfounder_independence}.

To avoid parameter explosion, we consider the architecture shared
across segments by assuming the mapping
$\boldf^{[s]}: \cR \times \cS \mapsto \cF$, which takes the internal
representation $\bR_{is'}$ and the segment indicator variable $s'$ as
inputs. Specifically, we first rewrite the conditional outcome model under the stochastic intervention defined in Equation~\eqref{eq:recursive_m} as the following pseudo-outcome model given the deconfounder for each $s = S_i, S_i-1, \cdots, 1$,
\begin{align}
    m_{s}(\overline{\bm{W}}_{is}, \{\boldf^{[s]}(\bR_{is'},s')\}_{s'=1}^{s}, S_i; \bm\delta_{\underline{s}})
  \ = \
    \E[\tilde{Y}_{is} \mid \overline{\bm{W}}_{is}, \{\boldf^{[s]}(\bR_{is'}, s')\}_{s'=1}^{s}, S_i].
\end{align}
where $\tilde{Y}_{is}$ denotes the pseudo outcome.  When
$S_i = s$, the pseudo outcome is defined as $\tilde{Y}_{is} := Y_i$, but for $S_i > s$, we marginalize $W_{i,s+1}$ under the stochastic intervention to obtain,
\begin{equation}
    \begin{aligned}
    \tilde{Y}_{is} &:=
    \frac{\delta_{s+1} p_{s+1}(\overline{\bm
      W}_{is})m_{s+1}(\overline{\bm{W}}_{is}, W_{i,s+1} = 1, \ \overline{\bm F}_{i,s+1}; \bm\delta_{\underline{s+1}}) }{\delta_{s+1}  p_{s+1}(\overline{\bm W}_{is}) + 1 - p_{s+1}(\overline{\bm
      W}_{is})}\\
      &\quad \hspace{.5in} + \frac{\{1 -
      p_{s+1}(\overline{\bm W}_{is})\} m_{s+1}(\overline{\bm{W}}_{is}, W_{i,s+1} = 0, \ \overline{\bm F}_{i,s+1}; \bm\delta_{\underline{s+1}})
      }{\delta_{s+1}  p_{s+1}(\overline{\bm W}_{is}) + 1 - p_{s+1}(\overline{\bm
      W}_{is})}, \label{pseudo_outcome}
    \end{aligned}
\end{equation}
where $\overline{\bm F}_{is} := \{\boldf^{[s]}(\bR_{i1},1) \cdots, \boldf^{[s]}(\bR_{is},s), S_i\}$ is the collection of deconfounders estimated for the pseudo-outcome for segment $s$.

Figure~\ref{architecture} summarizes our architecture, which
simultaneously estimates the deconfounder and the pseudo-outcome model. We
learn this neural network architecture by minimizing the squared loss,
\begin{align}
    \{\hat{\blambda}_s, \hat{\btheta}_s\} = \argmin_{\blambda_s, \btheta_s} \; \frac{1}{N_s} \ \sum_{i = 1}^N \mathbf{1}\{S_i \geq s\} \left\{ \tilde Y_{is} - m_{s}(\overline{\bm{W}}_{is}, \{\boldf^{[s]}(\bR_{is'},s';
  \blambda_s)\}_{s'=1}^{s}, S_i; \bm\delta_{\underline{s}}, \btheta_s) \right\}^2. \label{eq:minimization}
\end{align}
where $N_s = \sum_{i = 1}^N \mathbf{1}\{S_i \geq s\}$.
We make the parameters of the neural network explicit by letting $\bm\lambda_s$ represent
the parameters of deconfounder $\boldf^{[s]}$ to be estimated and using $\btheta_s$ to denote the parameters of the pseudo outcome model. 

\begin{figure}
\centering
\begin{tikzpicture}[
  node distance=0.7cm and 1.1cm,
  box/.style={draw, rectangle, minimum height=0.6cm, minimum width=0.9cm, align=center},
  func/.style={draw, rectangle, minimum height=0.6cm, minimum width=1.6cm, align=center},
  plate/.style={draw, rectangle, rounded corners=1pt, inner sep=0.35cm},
]

% ------------------ Plate block (repeated over t=1,...,T) ------------------
\node[box] (R) {$\bR_{is'}$};
\node[box, right=1.1 of R, fill=lightgray] (f) {$\boldf^{[s]}(\bR_{is'},\,s';\blambda_s)$};

\draw[->, line width=1] (R) -- (f);

% Plate around R and f
\node[plate, fit=(R)(f), label={[anchor=south east]{$\times s$}}] (P) {};

% ------------------ Outcome model + output ------------------
\node[box, right=1.5 of f, fill=lightgray, minimum width=2.7cm] (mu)
{$m_{s} \bigl(\overline{\bm{W}}_{is}, \{\boldf^{[s]}(\bR_{is'},\,s';\blambda_s )\}_{s' = 1}^{s}, S_i;\,\bm\delta_{\underline{s}}, \btheta_s\bigr)$};

\node[box, right=1.0 of mu] (Y) {$\tilde Y_{is}$};

\draw[->, line width=1] (f) -- (mu);
\draw[->, line width=1] (mu) -- (Y);

% ------------------ W-sequence input ------------------
\node[box, above left=1.0 and 0.8 of mu, fill=lightgray] (Wbar) {$\overline{\bm{W}}_{is}$};

\node[box, below left=1.0 and 0.8 of mu, fill=lightgray] (T) {$S_i$};

\draw[->, line width=1] (Wbar) -- (mu);
\draw[->, line width=1] (T) -- (mu);

\end{tikzpicture}
\caption{Diagram Illustrating the Proposed Model Architecture when predicting $\tilde Y_{is}$. The proposed model takes an internal
representation at each text segment $\bR_{is}$ as an input and finds the deconfounder $\boldf(\bR_{is})$, which is a lower-dimensional representation of $\bR_{is}$. The set of deconfounders at the all text segment and the treatment history $\overline{\bm{W}}_{iS_i}$ are then used to predict the outcome $\tilde Y_{is}$.}
\label{architecture}
\end{figure}

Given the above architecture, we estimate the causal quantity of
interest using the semiparametric estimation method for the
longitudinal data. To characterize the behavior of the estimator at
the limit, we first derive an influence function, using the results
from \cite{kennedy_nonparametric_2019} and
\cite{rotnitzky2017multiply}.
\begin{theorem}[Influence Function]\label{theorem_if}
  Denote the observed data by $\cD := \{\cD_i\}_{i = 1}^N$, where
  $\cD_i = \{Y_i, S_i, \overline{\bW}_{iS_i},
  \overline{\bR}_{iS_i}\}$. The influence function for the target
  parameter $\Psi(\bm\delta)$ is given by
{\small \begin{equation}
    \begin{aligned}
        &\psi(\cD_i; \bm\delta, \{ \pi_s, m_{s}, \tilde m_{s},  p_s, \boldf^{[s]} \}_{s= 1}^{S_i}, \Psi)\\
    = \ &  \sum_{s = 1}^{S_i}\left\{
\prod_{s^\prime=1}^{s-1}
\omega_{s^\prime}(\overline{\bm{H}}_{is^\prime}, W_{is^\prime};\delta_{s^\prime}, p_{s'}, \pi_{s'})
\right\}\\
&\quad \times
 \frac{\delta_s  p_s(\overline{\bm W}_{i,s-1}) m_s( \overline{\bW}_{i,s-1}, 1, \overline{\bm F}_{is}; \bm\delta_{\underline{s}}) \left\{1 - \frac{W_{is}}{\pi_s(\overline{\bm H}_{is})}\right\} +
    \{1 - p_s(\overline{\bm W}_{i,s-1})\} m_s( \overline{\bW}_{i,s-1}, 0, \overline{\bm F}_{is}; \bm\delta_{\underline{s}}) \left\{1 - \frac{1 - W_{is}}{1 - \pi_s(\overline{\bm H}_{is})}\right\}
    }{\delta_s  p_s(\overline{\bm W}_{i,s-1}) + 1 - p_s(\overline{\bm
        W}_{i,s-1})}\\
    &\quad + \sum_{s=1}^{S_i} \frac{\delta_s \{W_{is} - p_s(\overline{\bm W}_{i,s-1})\} \{
      \tilde m_s(\overline{\bm{W}}_{i,s-1}, 1; \bm\delta_{-s}) - \tilde
      m_s(\overline{\bm{W}}_{i,s-1}, 0; \bm\delta_{-s}) \} }{ \{\delta_s
      p_s(\overline{\bm W}_{i,s-1}) + 1 - p_s(\overline{\bm
        W}_{i,s-1})\}^2 } + \left\{\prod_{s= 1}^{S_i} \omega_s(\overline{\bm{H}}_{is}, W_{is};\delta_s, p_s, \pi_s )\right\} Y_i -  \Psi(\bm\delta), 
    \end{aligned} \label{if_formula}
\end{equation}}
where for $s = 1, \ldots, S_i$, and,
{\small \begin{align*}
\overline{\bm{H}}_{is} & \ := \ \{S_i, \overline{\bW}_{i,s-1},
\{\boldf^{[t]}(\bR_{ir}, r) : 1 \leq r\leq  t \leq s \}\}, \\ 
  \omega_s(\overline{\bm{H}}_{is}, W_{is}; \delta_s, p_s, \pi_s )
  & \ := \ \frac{\delta_s W_{is}  \frac{p_s(\overline{\bm W}_{i,s-1})}{\pi_s(\overline{\bm H}_{is})} + (1 - W_{is})
    \frac{1 - p_s(\overline{\bm W}_{i,s-1})}{1 - \pi_s(\overline{\bm
    H}_{is})}}{ \delta_s p_s(\overline{\bm W}_{i,s-1}) + 1 -
    p_s(\overline{\bm W}_{i,s-1}) }, \\
    \tilde m_{s}(\overline{\bm{W}}_{i,s-1}, w_s; \bm\delta_{-s})
    & \ := \ \E \biggl[\biggl(\prod_{s' = 1}^{s-1}\omega_{s'}(\overline{\bm{H}}_{is'}, W_{is'}
    ; \delta_{s'}, p_{s'}, \pi_{s'}
    ) \biggr)   m_s(\overline{\bW}_{i,s-1}, w_s, \overline{\bm F}_{is}; \bm\delta_{\underline{s}}) \mid
      \overline{\bm W}_{i,s-1}, S_i \geq s \biggr],
\end{align*}}
with $\omega_0 = 1$.

\begin{comment}
We also recursively define,
\begin{equation}
\begin{aligned}
   & m_{S_i}(\overline{\bm{H}}_{iS_i}, W_{i S_i}; \tilde{\bm\delta}_{S_i+ 1}) 
  \ := \ 
\mu(\overline{\bm{W}}_{i S_i}, \{\boldf(\bR_{is},s)\}_{s=1}^{S_i}) \\
    & m_{s}(\overline{\bm{H}}_{is}, W_{is};\tilde{\bm\delta}_{s+1}) 
  \ := \  \\
  & \qquad \E\left[\frac{\delta_{s+1} p_{s+1}(\overline{\bm
      W}_{is})m_{s+1}(\overline{\bm{H}}_{i,s+1}, 1;\tilde{\bm\delta}_{s+2}) + \{1 -
      p_{s+1}(\overline{\bm W}_{is})\} m_{s+1}(\overline{\bm{H}}_{i,s+1}, 0;\tilde{\bm\delta}_{s+2})
      }{\delta_{s+1}  p_{s+1}(\overline{\bm W}_{is}) + 1 - p_{s+1}(\overline{\bm
      W}_{is})} \ \biggl | \  \overline{\bm{H}}_{is}, W_{is}\right], \label{outcome_model}
\end{aligned}
\end{equation}
for $s=1,2,\ldots,S_i-1$.
\end{comment}
\end{theorem}
The proof is given in Appendix \ref{proof_if}. 
This influence function
is similar to the one obtained for marginal structural models in the
literature \citep{rotnitzky2017multiply}, except for the additional
term that accounts for the estimation error of the treatment
assignment probability under the stochastic intervention. This term
differs from the one that appears in
\cite{kennedy_nonparametric_2019}, as our stochastic intervention is
slightly different.

Using the influence function above, we construct the asymptotically
valid estimator and confidence interval.  Specifically, we use the
following $K$-fold cross-fitting procedure, assuming that $N$ is
divisible by $K$:

\begin{enumerate}
\item Set the incremental parameter $\bm{\delta} \in \R_{+}^{s_{\max}}$
\item Randomly partition the data into $K$ folds of equal size where
  the size of each fold is $n = N / K$. The observation index is
  denoted by $I(i) \in \{1, \ldots, K\}$, where $I(i) = k$ implies
  that the $i$th observation belongs to the $k$th fold

\item For each fold $k =1, \ldots, K$, use the observations with
  $I(i) \neq k$ as training data:

\begin{enumerate}
\item Estimate the conditional treatment probability given the past
  treatment sequence, which is denoted by
  $\hat p_s^{(-k)}(\overline{\bW}_{i,s-1}) := \widehat{\P}^{(-k)}(W_{is} = 1
  \mid \overline{\bW}_{i,s-1}, S_i \geq s)$ for each
  $s = 1, \ldots, s_{\max}$

\item For each segment $s = s_{\max}, s_{\max}-1, \ldots, 1$,
  proceeding backwards:
\begin{enumerate}
\item Simultaneously obtain the estimated deconfounder and pseudo outcome
  regression at segment $s$, denoted by 
  $\hat{\boldf}^{[s],(-k)}(\bR_{is'}, s') := {\boldf}^{[s]}(\bR_{is'}, s';
  \hat{\blambda}^{[s],(-k)})$ and $\hat m_s^{(-k)}$, by solving Equation~\ref{eq:minimization}
\item Construct
  $\widehat{\overline{\bm{H}}}_{is}^{(-k)} := \{S_i, \overline{\bW}_{i,s-1},
  \{\hat \boldf^{[t],(-k)}(\bR_{ir},r): 1 \leq r \leq t \leq s\} \}$
\item For each training observation with $S_i \geq s$, compute the
  pseudo outcome to be used at the next segment defined in Equation~\eqref{pseudo_outcome}.
\end{enumerate}

\item Estimate the propensity score given the estimated deconfounder
  history, which is denoted by
  $\hat{\pi}_s^{(-k)}(\widehat{\overline{\bm{H}}}_{is}^{(-k)}) :=
  \widehat{\P}^{(-k)}(W_{is} = 1 \mid
  \widehat{\overline{\bm{H}}}_{is}^{(-k)})$, for each
  $s = 1, \ldots, s_{\max}$, using the training observations with
  $S_i \geq s$

\item For each observation $i$ and $s =1, \ldots, S_i$, compute
\begin{align*}
    \hat\omega_{is}^{(-k)}(\widehat{\overline{\bm{H}}}_{is}^{(-k)}, W_{is};\delta_s,\hat p_s^{(-k)}, \hat \pi_s^{(-k)})
& := \frac{ \delta_s W_{is} \frac{\hat p_s^{(-k)}(\overline{\bm W}_{i,s-1})}{\hat{\pi}_s^{(-k)}(\widehat{\overline{\bm{H}}}_{is}^{(-k)})} + (1 - W_{is}) \frac{1 - \hat p_s^{(-k)}(\overline{\bm W}_{i,s-1})}{1 - \hat{\pi}_s^{(-k)}(\widehat{\overline{\bm{H}}}_{is}^{(-k)})} }{ \delta_s \hat p_s^{(-k)}(\overline{\bm W}_{i,s-1}) + 1 - \hat p_s^{(-k)}(\overline{\bm W}_{i,s-1}) },\\
\tilde\omega_{is}^{(-k)}(\widehat{\overline{\bm{H}}}_{is}^{(-k)}, W_{is}; \bm\delta_{\overline{s}} ) & :=\prod_{s'=1}^s \hat\omega_{is'}^{(-k)}(\widehat{\overline{\bm{H}}}_{is'}^{(-k)}, W_{is'};\delta_{s'}, \hat p_{s'}^{(-k)}, \hat \pi_{s'}^{(-k)} )
\end{align*}

\item Estimate the marginalized outcome regression appearing in the
  influence function, denoted
  $\widehat{\tilde m}_s^{(-k)}(\overline{\bm{W}}_{i,s-1}, w_s; \bm\delta_{-s})$ for $w_s \in \{0,1\}$, by the saturated
  regression of
  $\tilde\omega_{i,s-1}^{(-k)} \cdot
  \hat m_s^{(-k)}(
  \overline{\bW}_{i,s-1},w_s, \widehat{\overline{\bm F}}^{(-k)}_{is}; \bm\delta_{\underline{s}})$ on $\overline{\bW}_{i,s-1}$ among training
  observations with $S_i \geq s$.
\end{enumerate}

\item Compute the estimator $\widehat\Psi(\bm\delta)$ as a solution to
\begin{align}
    \frac{1}{nK}\sum_{k = 1}^K \sum_{I(i) = k} \psi(\cD_i; \bm\delta, \{
  \hat\pi_s^{(-k)}, \hat{m}_{s}^{(-k)}, \widehat{\tilde m}_{s}^{(-k)}, \hat{p}_s^{(-k)}, \hat\boldf^{[s](-k)}  \}_{s = 1}^{S_i}, \widehat{\Psi} ) = 0, \label{est_eq}
\end{align}
and compute 
\begin{align*}
    \hat\sigma^2(\bm\delta) = \frac{1}{nK} \sum_{k = 1}^K \sum_{I(i) = k} \psi(\cD_i; \bm\delta, \{\hat \pi_s^{(-k)}, \hat m_{s}^{(-k)}, \widehat{\tilde m}_{s}^{(-k)}, \hat p_s^{(-k)}, \hat\boldf^{[s](-k)}  \}_{s = 1}^{S_i}, \widehat \Psi )^2.
\end{align*}
\end{enumerate}

To derive the asymptotic properties of the proposed estimator, we
impose the following regularity conditions. 
\begin{assumption}[Regularity Conditions]\label{reg_cond}
Let $0 < c < C < \infty$ be positive constants and $r_n$ be a sequence of positive constants approaching zero as the sample size $N$ increases. Then, the following conditions hold.
\begin{enumerate}
    \item (Bounded incremental parameter) The set of incremental parameter $\Delta = [\delta_{\ell}, \delta_{u}]$ is bounded with $0 < \delta_\ell \leq \delta_u < \infty$.
    \item (Primitive Condition) For all $i \in \{1, \ldots, N\}$ and $s \in \{1, \ldots, s_{\max}\}$,
    \begin{enumerate}
        \item (Boundedness)
        \begin{align*}
            \P(|m_{s}(
            \overline{\bW}_{i,s-1},w_s, \overline{\bm F}_{is}; \bm\delta_{\underline{s}})| \leq C) & \ = \ \P(|\hat m_{s}(\overline{\bW}_{i,s-1},w_s, \widehat{\overline{\bm F}}_{is}; \bm\delta_{\underline{s}})| \leq C)\\
        &\ = \ \P(|\widehat{\tilde m}_s(\overline{\bm W}_{i,s-1}, w_s; \bm\delta_{-s})| \leq C) \\
        & \ = \ \P(|Y_i| \le C) \ = \ 1
        \end{align*}
        for $w \in \{0,1\}$ and $\bm\delta \in \Delta^{s_{\max}}$,
        and $$\P(\hat p_s(\overline{\bm{w}}_{s-1}) \leq 1) = \P(\hat\pi_s(\widehat{\overline{\bm{H}}}_{is}) \leq 1) = 1.$$
        \item (Sample Bounded Overlap)
        $$
        \P(\hat\pi_s(\widehat{\overline{\bm{H}}}_{is}) \ge c \cdot \hat p_s(\overline{\bm{W}}_{i,s-1})) \to 1 \quad \text{and} \quad \P( 1 - \hat \pi_s(\widehat{\overline{\bm{H}}}_{is}) \ge c \cdot (1 - \hat p_s(\overline{\bm{W}}_{i,s-1}))) \to 1
        $$ 
        %and $$
        %\P\biggl(\max_{1 \leq s \leq s_{\max}} \frac{W_{is} \pi_s({\overline{\bm{H}}}_{is}) + (1 - W_{is}) \{ 1 - \pi_s({\overline{\bm{H}}}_{is})\}
        %}{W_{is}\hat p_s(\overline{\bm{W}}_{i,s-1}) + (1 - W_{is}) \{ 1 - \hat p_s(\overline{\bm{W}}_{i,s-1})\}} \leq C\biggr) = 1.$$
        \item (Stable weight) $\P\biggl(\max_{1 \leq s \leq s_{\max}} \sup_{\delta_s \in \Delta} |\hat\omega_{is}(\widehat{\overline{\bm{H}}}_{is}, W_{is};\delta_s, p_s, \pi_s )| \leq C\biggr) \to 1$.
        \item (Variance stability) For
        $\sigma^2(\bm\delta) = \E[\psi(\cD_i; \bm\delta, \{ \pi_s, m_{s}, \tilde m_{s},  p_s, \boldf^{[s]} \}_{s= 1}^{S_i},  \Psi)^2]$, we have $$\inf_{\bm\delta\in\Delta^{\mathrm{s_{\max}}}}\sigma(\bm\delta)>0.$$
        %$$\P\biggl(\sup_{\delta \in \Delta} \biggl| \frac{\hat{\sigma}(\delta)}{\sigma(\delta)} - 1\biggr| < r_n\biggr) \to 1 \quad \text{and} \quad \inf_{\delta\in\Delta}\sigma(\delta)>0. $$
        %\item (Consistency) 
        %$$\P\biggl(\biggl|\biggl| \psi(\mathcal{D}_i; \delta, \{\hat\pi_s, \hat m_s, \hat p_s\}_{s = 1}^{S_i}, \hat\boldf, \widehat\Psi) - \psi(\cD_i; \delta, \{ \pi_s, m_{s}, \tilde m_{s}, p_s \}_{s = 1}^{S_i}, \boldf, \Psi  ) \biggr|\biggr| < r_n\bigg) \to 1$$
    \end{enumerate}
    \item (Nuisance Function Estimation) $\pi_s(\cdot)$ and $m_s(\cdot)$ are Lipschitz continuous at every point of its support, and satisfies:
    \begin{align*}
        & \E\left[\norm{\hat p_s(\overline{\bm{W}}_{i,s-1}) - p_s(\overline{\bm{W}}_{i,s-1})}^2 \right]^{\frac{1}{2}} \leq r_n N^{-\frac{1}{4}}, \hspace{1.0cm}
        \E\left[\norm{ 
        \hat\boldf^{[s']}(\bR_{is}, s) - \boldf^{[s']}(\bR_{is}, s) 
        }^2 \right]^{\frac{1}{2}} \leq r_n N^{-\frac{1}{4}},\\
        &\qquad\qquad\qquad\qquad\quad\quad\quad\quad \E\left[\norm{ \hat\pi_s(\widehat{\overline{\bm{H}}}_{is}) - \pi_s(\widehat{\overline{\bm{H}}}_{is})}^2 \right]^{\frac{1}{2}} \leq r_n N^{-\frac{1}{4}},\\
        &\sup_{\bm\delta \in \Delta^{s_{\max}} } \biggl\{
        \E\left[\norm{\hat{\tilde{m}}_{s}(\overline{\bW}_{i,s-1},W_{is}; \bm\delta_{-s}) - \tilde{m}_{s}(\overline{\bW}_{i,s-1},W_{is}; \bm\delta_{-s})}^2 \right]^{\frac{1}{2}}
        \\
        &\hspace{0.5in} + \E\left[\norm{\hat m_{s}(\overline{\bW}_{i,s-1},W_{is}, \widehat{\overline{\bm F}}_{is}; \bm\delta_{\underline{s}}) - m_{s}(\overline{\bW}_{i,s-1},W_{is}, \widehat{\overline{\bm F}}_{is}; \bm\delta_{\underline{s}})}^2 \right]^{\frac{1}{2}} \biggr\} \leq r_n N^{-\frac{1}{4}},
    \end{align*}
    for any $1 \leq s \leq s' \leq s_{\max}$ and $i \in \{1, \ldots, N\}$.
\end{enumerate}
\end{assumption}
These are standard regularity conditions in the literature
\citep{kennedy_nonparametric_2019}. Most importantly, the last
condition requires that each nuisance function be estimated at a rate
of $N^{-1/4}$, which can be achieved using standard feedforward neural
networks with ReLU activations and sufficient depth and width
\citep{farrell_deep_2021}. Unlike the standard static double machine
learning framework, however, our setting does not enjoy double
robustness.

Given the above assumptions, the asymptotic normality of the proposed estimator can be established as follows.
\begin{theorem}[Asymptotic Normality]\label{asymp_normal}
Under Assumptions \ref{consistency}--\ref{reg_cond}, the estimator under the stochastic intervention $\hat\Psi(\bm\delta)$ obtained by solving the estimating equations satisfies the asymptotic normality:
\begin{align*}
    \frac{\sqrt{N}(\widehat{\Psi}(\bm\delta) - \Psi(\bm\delta) )}{\sigma(\bm\delta) } \xrightarrow{d} \mathbb{G}(\bm\delta)
\end{align*}
where $\mathbb{G}(\cdot)$ is a mean-zero Gaussian process with covariance $$\E[\mathbb{G}(\bm\delta^{(1)})\mathbb{G}(\bm\delta^{(2)})] = \E\biggl[\biggl(\frac{\psi(\cD_i; \bm\delta^{(1)},  \{ \pi_s, m_{s}, p_s,\boldf^{[s]} \}_{s = 1}^{S_i}, \Psi)}{\sigma(\bm\delta^{(1)})}\biggr)\biggl(\frac{\psi(\cD_i; \bm\delta^{(2)}, \{ \pi_s, m_{s}, p_s, \boldf^{[s]} \}_{s = 1}^{S_i}, \Psi)}{\sigma(\bm\delta^{(2)})}\biggr)\biggr].$$
\end{theorem}
The proof is given in Appendix~\ref{proof_asymp_normal}. The proof structure closely parallels Theorem 3 of \cite{kennedy_nonparametric_2019}, with the key differences being the estimation error introduced by the deconfounder and the modified stochastic intervention.

\section{Simulation Studies}\label{sec:simulation}
We conduct simulation studies to evaluate the empirical performance of the proposed estimator. Because generating sequences of actual sentences under known data-generating mechanisms is infeasible, we instead perform numerical simulations grounded in the motivating example described in Section~\ref{sec::example}. We then apply our proposed methodology to real unstructured data in the empirical application presented in Section~\ref{sec:application}.

\subsection{Setup}
For each simulation, we consider sample sizes
\(N \in \{2000,3000,4000,5000\}\) and generate \(N\) observations for
each setting. For each observation \(i\), we independently draw the text
length \(S_i\) as
\begin{align*}
S_i
=
\min\{s_{\max},S_{\min}+\operatorname{Poisson}(\lambda)\},
\end{align*}
where \(S_{\min}=5\), \(\lambda=2.5\), and \(s_{\max}=10\). 

Before generating the observations, for each segment
\(s=1,\ldots,s_{\max}\), we draw the outcome coefficients
\begin{align*}
\bm\beta_{s}
&=
0.6\frac{\bm g_{s}}{\sqrt p}, \quad 
\bm g_{s}
\sim
\mathcal{N}(\bm 0,\bm I_p), \quad 
\lambda_s
\sim
\mathcal{N}(0,0.15^2), \quad
\kappa_s
\sim
\mathcal{N}(0,0.15^2),
\end{align*}
the treatment--confounder feedback coefficients
\begin{align*}
\bm\phi_{s}
=
0.7\frac{\bm h_{s}}{\sqrt p},
\qquad
\bm h_{s}
\sim
\mathcal{N}(\bm 0,\bm I_p),
\qquad
s=1,\ldots,s_{\max}-1,
\end{align*}
and the propensity coefficients
\begin{align*}
\bm a_{s}
=
0.375\,
\frac{\bm\beta_{s}}{\|\bm\beta_{s}\|_1},
\qquad
r_s
=
0.125\,\operatorname{sign}(g_{s}),
\qquad
g_{s}
\sim
\mathcal{N}(0,1),
\end{align*}
so that \(\|\bm a_s\|_1+|r_s|=0.5\) for every segment. All coefficients are held fixed across all
Monte Carlo replications. We also generate, for each segment, a matrix
\(\bm M_s\in\mathbb{R}^{d\times(p+1+q)}\) with orthonormal columns,
so that
\(\bm M_s^\top\bm M_s=\bm I_{p+1+q}\), where $d = 512$ is the dimensions of internal representations, $p =4$ is that of confounding features, and $q = 4$ is that of irrelevant features. Specifically, we obtain
\(\bm M_s\) by applying a QR decomposition to a
\(d\times(p+1+q)\) matrix whose entries are independently drawn from a
standard normal distribution. The matrices \(\bm M_s\) are also held
fixed across Monte Carlo replications.

We first generate the confounding features. For the first segment, we
draw
\begin{align*}
\bU_{i1}
=
\tanh(\bm\eta_{i1}),
\qquad
\bm\eta_{i1}
\sim
\mathcal{N}(\bm 0,\bm I_p).
\end{align*}
For each subsequent segment, we generate
\begin{align}
\bU_{is}
=
\tanh\left(
W_{i,s-1}
\bm\phi_{s-1}
+
\bm\eta_{is}
\right),
\qquad
\bm\eta_{is}
\sim
\mathcal{N}(\bm 0,\bm I_p),
\quad
s=2,\ldots,S_i.
\label{eq:sim_U_transition}
\end{align}

Next, for each segment, we independently draw
\(\xi_{is}\sim\operatorname{Logistic}(0,1)\) and define the binary treatment feature as
\begin{align}
W_{is}
=
\mathbf{1}\{\bm a_{s}^{\top}\bU_{is}
+
r_s W_{i,s-1}
+
\xi_{is}>0\}.
\label{eq:sim_treatment}
\end{align}
This construction ensures that the true propensity score function is Lipschitz-continuous and satisfies the bounded overlap assumption.
\begin{comment}
Consequently, the conditional treatment probability is
\begin{align}
\pi_s(\overline{\bm W}_{i,s-1},\overline{\bU}_{is})
&=
\P\left(
W_{is}=1
\mid
\overline{\bm W}_{i,s-1},
\overline{\bU}_{is},
S_i\geq s
\right)
\nonumber\\
&=
\operatorname{expit}\left(
\bm a_{s}^{\top}\bU_{is}
+
r_s W_{i,s-1}
\right).
\label{eq:sim_propensity}
\end{align}
Because each element of \(\bU_{is}\) lies between \(-1\) and \(1\) and
\(W_{i,s-1}\in\{0,1\}\), the absolute value of the linear predictor in
Equation~\eqref{eq:sim_propensity} is bounded above by
\(\|\bm a_s\|_1+|r_s|=0.5\).
Therefore,
\begin{align*}
\operatorname{expit}(-0.5)
\leq
\pi_s(\overline{\bm W}_{i,s-1},\overline{\bU}_{is})
\leq
\operatorname{expit}(0.5),
\end{align*}
which implies the bounded relative overlap condition. In addition,
the propensity score is a smooth and Lipschitz-continuous function of
the confounding features.
\end{comment}

We then construct the high-dimensional internal representation. Let
\(\bm V_{is}\in\mathbb{R}^{q}\) denote features of the representation
that affect neither treatment nor the outcome, where $\bm V_{is}
\sim \mathcal{N}(\bm 0,\bm I_q)$
independently across observations and segments. We define
\begin{align}
\bR_{is}
=
\bm M_s
\begin{bmatrix}
\bU_{is}\\
c\,W_{is}\\
\bm V_{is}
\end{bmatrix},
\label{eq:sim_representation}
\end{align}
where $c=0.05$.
\begin{comment}
This construction distributes the treatment, confounding, and
remaining features across all coordinates of the
\(d\)-dimensional representation. Since
\(\bm M_s^\top\bm M_s=\bm I_{p+1+q}\), we have
\begin{align*}
\bm M_s^\top\bR_{is}
=
\begin{bmatrix}
\bU_{is}\\
c\,T_{is}\\
\bm V_{is}
\end{bmatrix}.
\end{align*}
Thus, both the confounding feature \(\bU_{is}\) and the treatment
feature \(W_{is}=\mathbf{1}\{(\bm M_s^\top\bR_{is})_{p+1}>0\}\) are
deterministic functions of \(\bR_{is}\) for any \(c>0\). Notice that
treatment remains stochastic conditional on \(\bU_{is}\), because of
the logistic innovation \(\xi_{is}\), even though it is deterministic
conditional on the complete internal representation.
\end{comment}
Finally, we generate the bounded outcome as
\begin{align}
Y_i
=
50
+
25\tanh\left[
\frac{1}{S_i}
\left\{
\sum_{s=1}^{S_i}
\left(
\tau_s W_{is}
+
\bm\beta_{s}^{\top}\bU_{is}
+
\lambda_s W_{is} U_{is,1}
\right)
+
\sum_{s=2}^{S_i}
\kappa_s W_{i,s-1} W_{is}
\right\}
+
\epsilon_i
\right],
\label{eq:sim_outcome}
\end{align}
where \(U_{is,1}\) denotes the first coordinate of \(\bU_{is}\),
\(\epsilon_i\sim\mathcal{N}(0,0.1^2)\), and $\tau_s
=
0.6\times0.95^{s-1}$. Averaging the segment sum by \(S_i\) keeps the
argument of the \(\tanh\) transformation away from its saturation
region.

We evaluate the target estimand \(\Psi(\bm\delta)\) over the grid $\delta_s \in \{0.5,0.75,1.0,1.25,1.5,2.0\}$
for all \(s\). This data generating process ensures that the marginal treatment
assignment probability satisfies
\begin{align*}
p_s(\overline{\bm w}_{s-1})
=
\P(W_{is}=1\mid W_{i,s-1}=w_{s-1},S_i\geq s)
\end{align*}
for \(s\geq2\), while \(p_1=\P(W_{i1}=1)\). We approximate these
probabilities using an independent Monte Carlo sample of size
\(N_{\mathrm{oracle},p}=400{,}000\), and construct the oracle
stochastic intervention as defined in
Definition~\ref{def:stochastic}. Because the treatment feeds back into
the confounding features, we compute the true value of
\(\Psi(\bm\delta)\) by the g-formula: for each \(\bm\delta\), we draw
the treatment path from the intervened distribution and regenerate the
confounding features from the transition in
Equation~\eqref{eq:sim_U_transition} given the intervened treatments,
using an independent Monte Carlo sample of size
\(N_{\mathrm{truth}}=400{,}000\).

We compare the proposed estimator against this oracle benchmark. Specifically,
we implement the proposed method using two-fold sample splitting rather than
full cross-fitting due to computational constraints. Within the training fold,
we first train the representation neural network shown in
Figure~\ref{architecture}. In the simulation implementation, the feature
representation network consists of a hidden layer with width 64 followed by a
two-dimensional deconfounder representation, and the prediction head consists
of two hidden layers with widths 128 and 64. The network is trained for 100
epochs with batch size 512, learning rate $0.001$, patience 10, and
dropout rate 0.1.

We then estimate the remaining nuisance components needed for the
semiparametric estimator. The treatment-history probability
$p_s(\overline{\bm W}_{i,s-1})$ and the marginalized outcome regression
$\tilde m_s(\overline{\bm W}_{i,s-1},w;\bm\delta_{-s}))$ are estimated using
saturated models based on the full observed treatment history.
The propensity score conditional on the estimated deconfounder history is
estimated using a feedforward neural network with a single hidden layer of
width 32, learning rate $3\times10^{-4}$, batch size 256, and dropout rate 0.3. The entire estimation procedure is repeated 200 times for each
sample size.

\subsection{Results}

\begin{figure}[!t]
    \centering
    \includegraphics[width=1.0\linewidth]{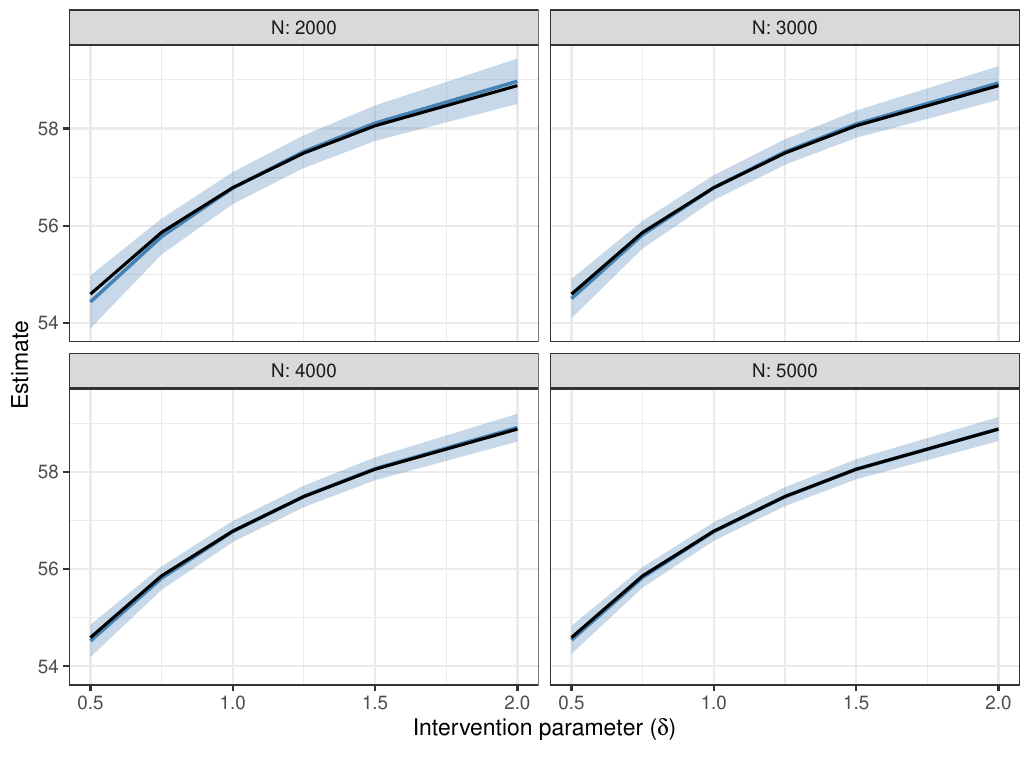}
    \caption{Average estimated potential outcomes (blue) versus their oracle values (black) across different values of the incremental parameter $\bm\delta$, based on 200 Monte Carlo iterations. The blue solid line denotes the mean estimate, and the blue shaded region represents the average confidence interval constructed from the average standard error. The horizontal axis displays the incremental parameter $\bm\delta$, which scales the treatment assignment probability, while the vertical axis reports the corresponding estimated average counterfactual outcome. The value $\delta_s = 1$ denotes the observed treatment mechanism; $\delta_s < 1$ attenuates, and $\delta_s > 1$ amplifies, the probability of treatment for each segment $s$.}
    \label{sim_results}
\end{figure}

Figure~\ref{sim_results} reports the average estimates (blue solid line) and average 95\% confidence intervals (blue shaded region) across different values of the incremental parameters and sample sizes, based on 200 Monte Carlo iterations, together with the ground truth (black). The results demonstrate that, across different sample sizes, the estimated curve closely follows the oracle estimates and becomes increasingly accurate as the sample size increases.

\begin{figure}[!t]
    \centering
    \includegraphics[width=1.0\linewidth]{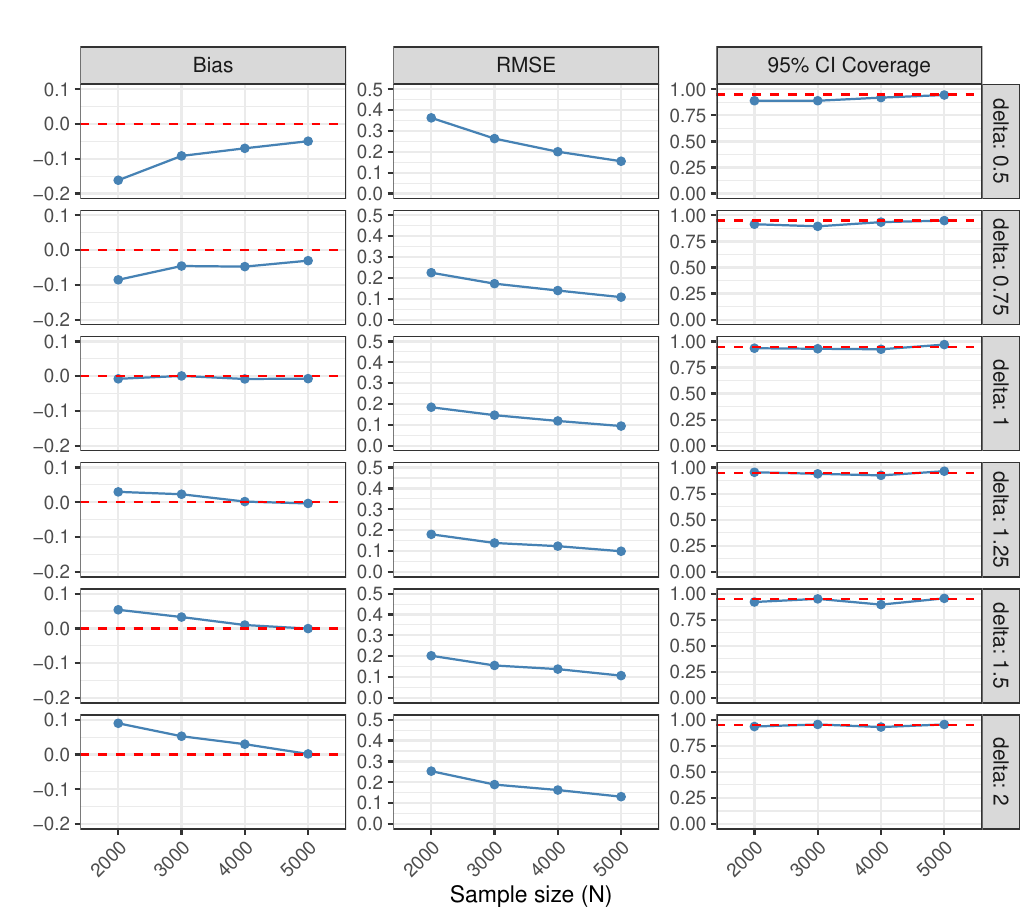}
    \caption{Performance of the estimator across different sample sizes and values of the incremental parameter $\delta$, based on 200 Monte Carlo iterations. The red dotted line in the bias plot indicates zero, while the red dotted line in the 95\% confidence interval coverage plot denotes the nominal coverage level (0.95).  }
    \label{metricN}
\end{figure}

We also compute the bias, RMSE (root mean squared error), and coverage
of 95\% confidence intervals across different values of the
incremental parameter $\delta$ and sample sizes. Figure~\ref{metricN}
illustrates the results, and the corresponding numerical results are
reported in Table~\ref{tab:mc_summary_raw} of
Appendix~\ref{app:emp_result}. We find that performance improves, on
average, as the sample size increases. Moreover, across all sample
sizes and values of $\delta$, the 95\% confidence intervals achieve
coverage close to the nominal level. 
%We note that the coverage does not exactly reach 95\% in some cases, even with large sample sizes, presumably because the treatment assignment mechanism in the data-generating process is not Lipschitz continuous. This non-Lipschitz structure is necessary to satisfy Assumption~\ref{treatment_feature} and to make the simulation design controllable. However, as a result, estimation errors in the deconfounder may propagate into the propensity score estimation, potentially leading to slight undercoverage.

\section{Empirical Application}\label{sec:application}

Finally, we apply the proposed methodology to the data from
\cite{fong_causal_2023} introduced in the
Section~\ref{sec::example}. Following the original study, we use the
description about U.S. commitments to Hong Kong as a treatment and the
public’s preference for the U.S. government to support Hong Kong
protesters as an outcome. However, unlike the original studies, we
explore whether the ordering of treatment features alters the
treatment effect. Specifically, we combine the responses from two
surveys (December 2019 with $N = 1983$ and October 2020 with
$N = 2072$) and then exclude the respondents who are assigned three
sentences because they always receive the treatment, violating the
overlap assumption (Assumption~\ref{overlap}).  This yields the final
sample size of $N = 2085$ with $S_i = 2$ for all observations.

We regenerate each text segment using the large language
models. Specifically, we use LLaMa3.1 with 8 billion parameters and
ask LLMs to exactly reproduce the input texts.  We then extract the
internal representations of the last token, which is the pooling
approach widely used in the literature
(e.g.,\citealt{neelakantan2022text, ma2024fine,
  jiang2023scaling}). The input dimension size for the internal
representation is 4096 for each text segment.

Once we extract the internal representations, we implement the
proposed dynamic GPI methodology. This time, we estimate
\(\widehat{\Psi}(\delta_1, 1)\) and \(\widehat{\Psi}(1, \delta_2)\),
corresponding to interventions applied only at time~1 and time~2,
respectively. In each case, $\delta_s$ varies over a common grid
\(\delta_s \in [10^{-4}, 10^{4}]\). Given the limited sample size, we
employ 10-fold cross-fitting to maximize the amount of data used for
nuisance function estimation. In each fold, we first train the neural
network shown in Figure~\ref{architecture}, which consists of two
hidden layers with width 1024 and a prediction head with a single
hidden layer of width 200. We use the same hyperparameters as in the
simulation studies and estimate the three nuisance functions: the
treatment history model, the propensity score model \(\hat{\pi}_1\),
and the outcome model \(\hat{m}_s\) for each \(s \in \{1,2\}\). For
nuisance estimation, we use saturated models for the observed
treatment history probabilities and the marginalized regression term
under the stochastic intervention. The propensity score and outcome
models are estimated using neural networks with the same
hyperparameters as in the simulation section.

\begin{figure}[!t]
    \centering
    \includegraphics[width=1.0\linewidth]{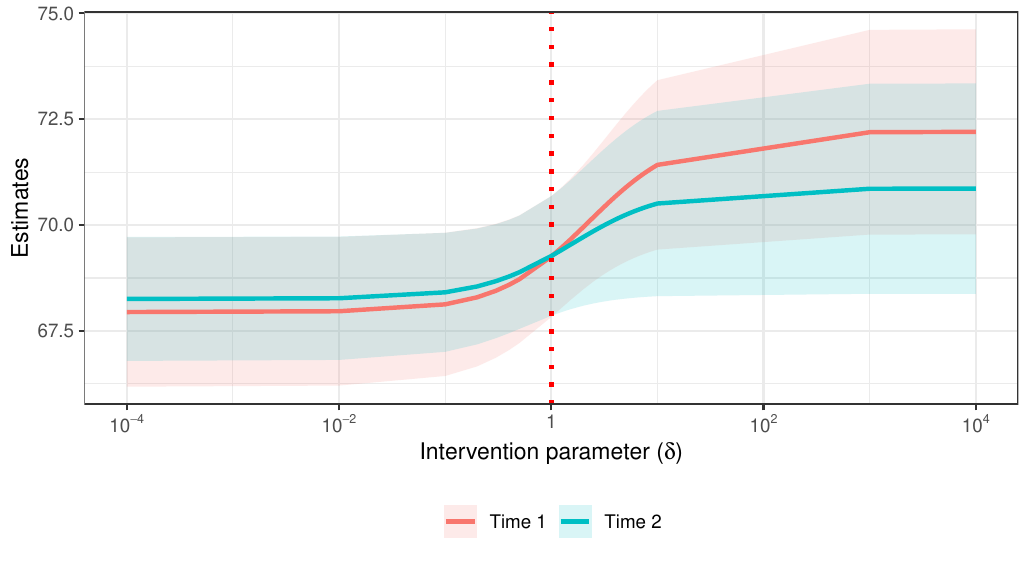}
    \caption{Estimated average potential outcomes under stochastic
      interventions that perturb the treatment assignment probability
      at time 1 (red) and time 2 (blue) separately. The horizontal
      axis displays the incremental parameter $\delta$, which scales
      the treatment assignment probability, while the vertical axis
      reports the corresponding estimated average potential
      outcome. The value $\delta = 1$ denotes the observed treatment
      mechanism; $\delta < 1$ attenuates, and $\delta > 1$ amplifies,
      the probability of treatment.}
    \label{fig:res_emp}
\end{figure}

The result is presented in Figure~\ref{fig:res_emp} where the
horizontal axis represents $\delta$ and the estimated average
potential outcome is the solid line with the shaded area indicating
the 95\% pointwise confidence interval. We find that while
interventions at the first sentence meaningfully shift the average
counterfactual response, whereas the effect is substantially
attenuated when applied to the second sentence. This aligns with the
existing evidence of the survey research that the ordering of the
stimulus matters \citep[e.g.,][]{krosnick1987evaluation,
  auspurg2017first}. Therefore, while this application confirms the
finding from the existing literature, it also highlights the efficacy
of the proposed methodology, enabling researchers to explore the
heterogeneity of treatment effects in dynamic unstructured data.

\section{Conclusion}\label{sec:conclusion}

In this paper, we develop a framework for dynamic causal inference
with unstructured objects by treating text, audio, and video as
ordered sequences and defining causal effects with respect to the
timing and placement of treatment features within those sequences. By
leveraging internal representations from generative AI models, we show
that it is possible to construct low-dimensional deconfounders that
summarize evolving contextual information.  This allows us to identify
and estimate causal effects under dynamic stochastic interventions.

We further propose a neural network–based estimation strategy combined
with semiparametric inference, and demonstrate through simulations
that the resulting estimator performs well in finite samples. Our
empirical application to the Hong Kong experiment illustrates the
substantive value of this perspective, showing that the effect of a
treatment feature depends on when it appears in the text.

While the application in this paper has focused on causal inference
with texts, the proposed methodology can be extended to other
modalities such as video. For video data, it may be necessary to
consider the appropriate encoding architecture, as we must capture
multimodal information, including both visual and audio
signals. Another promising extension is the settings with sequential
outcomes. In the current framework, we consider a single outcome per
respondent; however, for modalities such as audio and video, it is
often possible to record responses continuously over time. This would
substantially increase the number of effective observations and
potentially allow researchers to leverage greater statistical power.

\newpage 
\bibliography{references, my}

\newpage

\appendix
\setcounter{equation}{0}
\setcounter{figure}{0}
\setcounter{table}{0}
\setcounter{section}{0}
\renewcommand {\theequation} {S\arabic{equation}}
\renewcommand {\thefigure} {S\arabic{figure}}
\renewcommand {\thetable} {S\arabic{table}}
\renewcommand {\thesection} {S\arabic{section}}

\begin{center}
  {\LARGE \bf Appendix} 
\end{center}

\section{Theoretical Proofs}

\subsection{Proof of Proposition~\ref{prop:seq_independent_support}}\label{proof:seq_independent_support}

%comment: cite PNAS paper once it is accepted.
\begin{proof}
We prove it by contradiction. Suppose, for contradiction, that
\begin{align*}
\mathcal S_{WU,s}^{(s')}
\neq
\mathcal W_s \times \cU_s^{(s')}.
\end{align*}
where $\mathcal{W}_s = \{0,1\}$ is the support of treatment feature at the segment $s$ and $ \mathcal S_{WU,s}^{(s')}
=
\{(g_W(\bx),{\bg}_{\bU_s}^{(s')}(\bx)):\bx\in\cX_s\}$ is the joint support of treatment feature and confounding feature.
This means that there exists a pair $(w^\dagger,\bu^\dagger)
\in
\mathcal W_s \times \cU_s^{(s')}
$
such that
$
(w^\dagger,\bu^\dagger)
\notin
\mathcal S_{WU,s}^{(s')}
$.
Because $\bu^\dagger \in \cU_s^{(s')}$, there exists some $\bx_0 \in \cX_s$ satisfying $\bg_{\bU_s}^{(s')}(\bx_0)=\bu^\dagger$ and $w_0 := g_W(\bx_0) \neq w^\dagger$. Now, we consider $\bg'$ and $\tilde{\bg}_{\bU_s}^{(s')}$ as $\bg'(\bx)=\bg_{\bU_s}^{(s')}(\bx)$
and
\begin{align*}
     \tilde{\bg}_{\bU_s}^{(s')}(w,\bu) = 
     \begin{cases}
     \bu \quad \text{if } (w,\bu) \neq (w^\dagger, \bu^\dagger)\\
     \bm c \quad \text{if } (w,\bu) = (w^\dagger,\bu^\dagger) 
     \end{cases}
\end{align*}
where $\bm c \neq \bu^\dagger$. This is well-defined because $(w^\dagger,\bu^\dagger)\notin\mathcal S_{WU,s}^{(s')}$ (so the value of $\tilde{\bg}_{\bU_s}^{(s')}$ at this point is unconstrained by the requirement that it reproduce the observed data; the independence of support means that we can change $W_{is}$ and $\bU_{is}^{(s')}$ separately). Now, by definition of $\tilde{\bg}_{\bU_s}^{(s')}$,
\[
\tilde{\bg}_{\bU_s}^{(s')}(w_0, \bg'(\bx_0))
=
\tilde{\bg}_{\bU_s}^{(s')}(w_0,\bu^\dagger)
=
\bu^\dagger,
\]
while
\[
\tilde{\bg}_{\bU_s}^{(s')}(w^\dagger,\bg'(\bx_0))
=
\tilde{\bg}_{\bU_s}^{(s')}(w^\dagger,\bu^\dagger) = \bm c
\neq
\bu^\dagger.
\]
Thus, $\tilde{\bg}_{\bU_s}^{(s')}(w,\bg'(\bx_0))$ changes with $w$ while holding $\bg'(\bx_0)$ fixed. Therefore, $\bU_{is}^{(s')}$ can be represented as a nontrivial deterministic function of
$W_{is}$ and another feature $\bg'(\bx)$, contradicting the separability assumption. Hence no such unattainable pair $(w^\dagger,\bu^\dagger)$ can exist under separability.
Therefore, $\mathcal S_{WU,s}^{(s')}
=
\mathcal W_s \times \cU_s^{(s')}
$.
\end{proof}

\subsection{Proof of Theorem~\ref{identification} }\label{proof_identification}

First, the estimand defined in Equation~\eqref{target_quantity} can be
written as $\Psi(\bm\delta) = \E[\Psi_{S_i}(\bm\delta)]$ where,
\begin{align*}
&\Psi_{s}(\bm\delta) := 
    \E\left[\int_{\mathcal{W}^{s} }
                 Y_i(\overline{\bW}_{is} = \bar{\bm w}_{s}, \overline{\bU}_{is}) \prod_{s^\prime= 1}^{s} d Q_{s^\prime}(w_{s^\prime};
                 \overline{\bm{w}}_{s^\prime-1}, \delta_{s^\prime}) \ \biggl | \ S_i =
                 s \right], %\\
    %&= \int_{\cR^{S_i}} \int_{\mathcal{W}^{S_i} } \E[Y_i(W_{i1} = w_1, \bU_{i1}, \ldots, W_{iS_i} = w_{S_i}, \bU_{iS_i}) \mid S_i] \prod_{s= 1}^{S_i} d Q_s(w_s; \overline{\bm{W}}_{i,s-1}, \overline{\bR}_{i,s}, \delta)dF(\bR_{is} \mid \overline{\bR}_{i,s-1}).
\end{align*}
Thus, it suffices to identify $\Psi_{s}(\bm\delta)$.
Assumptions~\ref{randomization}--\ref{separability} imply the
following sequential ignorability given the confounding features,
\begin{align}
    Y_{i}(\bar{\bm w}_{s}, \overline{\bU}_{is}) \ \indep \ {W}_{is} \mid \overline{\bW}_{i,s-1}, \overline{\bU}_{is}, S_i = s,\label{ignorability}
\end{align}
for any respondent $i$ and segment $s$. Moreover, by Assumptions~\ref{treatment_feature}, \ref{confounding_feature} and \ref{det_dec},
\begin{align}
    \bU_{i,s'+1} \ \indep \ \overline{\bR}_{is^\prime} \mid \overline{\bW}_{is^\prime}, \overline{\bU}_{is^\prime}, S_i = s  \label{sufficiency_U_for_R}
\end{align}
for any $s \in \mathcal{S}$ and $s' < s$.

As the support $\cY$ is bounded, by Fubini theorem, 
\begin{align*}
    \Psi_{s}(\bm\delta) = 
    \int_{\mathcal{W}^{s} } \E\left[
                 Y_i(\overline{\bW}_{is} = \bar{\bm w}_{s}, \overline{\bU}_{is}) \ \biggl | \ S_i =
                 s \right] \prod_{s^\prime= 1}^{s} d Q_{s^\prime}(w_{s^\prime};
                 \overline{\bm{w}}_{s^\prime-1}, \delta_{s'}).
\end{align*}
Therefore, by appealing to the
standard g-formula, we obtain,
\begin{align*}
  & \Psi_{s}(\bm\delta) \\
  % & := \E\left[ \int_{\mathcal{W}^{S_i} }
                       %Y_i(W_{i1} = w_1, \bU_{i1}, \ldots,
                       %W_{iS_i} = w_{S_i}, \bU_{iS_i})  \prod_{s^\prime=
                       %1}^{S_i} d Q_{s^\prime}(w_{s^\prime}; \overline{\bm{W}}_{i,s^\prime-1},
                       %\delta) \ \biggl | \ S_i = s \right]\\
   = \ & \int_{\mathcal{W}^{s}} \E\left[\E\left\{   Y_i(\bar{\bm w}_{s}, \overline{\bU}_{is})   \ \biggl
      | \ \bU_{i1}, S_i = s\right\} \ \biggl | \  S_i = s \right] \prod_{s^\prime=
      1}^{s} d Q_{s^\prime}(w_{s^\prime}; \overline{\bm{w}}_{s^\prime-1}, \delta_{s'}) \\
   = \ & \int_{\mathcal{W}^{s}} \E\left[\E\left\{   Y_i(\bar{\bm w}_{s}, \overline{\bU}_{is})
          \ \biggl
      | \ W_{i1} = w_1, \bU_{i1}, S_i = s\right\}\ \biggl | \  S_i = s\right] \prod_{s^\prime = 1}^{s} d Q_{s^\prime}(w_{s^\prime}; \overline{\bm{w}}_{s^\prime-1}, \delta_{s'})   \\
  = \ & \int_{\mathcal{W}^{s}} \E\biggl[\E\biggl\{ \E\biggl(   Y_i(\bar{\bm w}_{s}, \overline{\bU}_{is})  \ \biggl | \  W_{i1} = w_1,
      \overline{\bU}_{i2}, S_i  = s\biggr)  \ \biggl | \ W_{i1} = w_1, \bU_{i1},
        S_i = s \biggr\} \ \biggl | \  S_i = s\biggl] \\
  & \qquad \times \prod_{s^\prime= 1}^{s} d Q_{s^\prime}(w_{s^\prime};
      \overline{\bm{w}}_{s^\prime-1}, \delta_{s'}) \\
   = \  & \int_{\mathcal{W}^{s}} \E\biggl[\E\biggl\{ \E\biggl(    Y_i(\bar{\bm w}_{s}, \overline{\bU}_{is})  \ \biggl | \
      \overline{\bW}_{i2} = \bar{\bm w}_{2}, \overline{\bU}_{i2}, S_i = s
      \biggr) \  \biggr | \ W_{i1} = w_1, \bU_{i1}, S_i = s\biggr\}
          \ \biggr | \  S_i = s\biggr] \\
  & \qquad \times \prod_{s^\prime= 1}^{s} d Q_{s^\prime}(w_{s^\prime};
      \overline{\bm{w}}_{s^\prime-1}, \delta_{s'}),
\end{align*}
where the first and third equalities are due to the law of iterated
expectation and the second and fourth equalities follow from
Equation~\eqref{ignorability}. Under
Assumptions~\ref{confounding_feature}~and~\ref{det_dec}, we can write
$\bU_{is} = \boldf^*_{s}(\bR_{is})$ with some deterministic function
$\boldf^*_{s}$. By repeating this operation up to segment $S_i = s$,
we obtain
\begin{align*}
  & \Psi_{s}(\bm\delta) \\
   = \ &  \int_{\mathcal{W}^{s}}  \int_{\cU^{s}} \E[
         Y_i(\bar{\bm w}_{s}, \overline{\bU}_{is})  \mid
         \overline{\bm{W}}_{is} = \bar{\bm w}_{s}, \overline{\bU}_{is}, S_i = s]\\
    &\qquad\qquad \times \prod_{s^\prime= 1}^{s}  dF(\bU_{is^\prime} \mid
      \overline{\bm{W}}_{i,s^\prime - 1} = \overline{\bm{w}}_{s^\prime-1}, \overline{\bU}_{i,s^\prime-1}, S_i = s) d Q_{s^\prime}(w_{s^\prime};
      \overline{\bm{w}}_{s^\prime-1}, \delta_{s'})  \\
  = \  &  \int_{\mathcal{W}^{s}} \int_{\cU^{s}} \E[  Y_i  \mid \overline{\bm{W}}_{is} = \bar{\bm w}_{s}, \overline{\bU}_{is}, S_i=s]\\
  &\qquad\qquad \times \prod_{s^\prime= 1}^{s}  dF(\bU_{is^\prime} \mid
      \overline{\bm{W}}_{i,s^\prime - 1} = \overline{\bm{w}}_{s^\prime-1}, \overline{\bU}_{i,s^\prime-1}, S_i = s) d Q_{s^\prime}(w_{s^\prime};
      \overline{\bm{w}}_{s^\prime-1}, \delta_{s'})\\
   = \ &   \int_{\mathcal{W}^{s}} \int_{\cR^{s}} \E[  Y_i
         \mid \overline{\bm{W}}_{is} = \bar{\bm w}_{s},
         \boldf^*_{1}(\bR_{i1}),\ldots,\boldf^*_{s}(\bR_{is}),
         S_i = s] \\
    &\qquad \qquad\times \prod_{s^\prime = 1}^{s}  d Q_{s^\prime}(w_{s^\prime}; \overline{\bm{w}}_{s^\prime-1}, \delta_{s'})dF(\bR_{is^\prime} \mid \overline{\bm{W}}_{i,s^\prime-1} = \bar{\bm w}_{s^\prime-1}, \boldf_{1}^*(\bR_{i1}),\ldots,\boldf_{s'-1}^*(\bR_{i,s'-1}), S_i=s),
\end{align*}
where the second equality is due to Assumption~\ref{consistency}.

Finally, we show that any deconfounder $\boldf_s(\bR_{is})$
satisfying the mean independence relation in Equations~\ref{deconfounder_terminal_independence} and \ref{deconfounder_independence} lead to the same identification formula. Suppose there is another
function $\boldf_{s}: \mathcal{R} \rightarrow \mathcal{Q}$ which
satisfies the above mean independence relation, is separable from the
treatment feature, and satisfies the bounded relative overlap. Then, the inner expectation is written as
\begin{align*}
&\E[Y_i \mid \overline{\bW}_{is}, \boldf_1(\bR_{i1}), \ldots,
\boldf_{s}(\bR_{is}), S_i = s]\\
= \ & \E[Y_i \mid \overline{\bW}_{is},
\boldf_{1}(\bR_{i1}),\ldots,\boldf_{s}(\bR_{is}),
\overline{\bR}_{is}, S_i = s]\\
= \ & \E[Y_i \mid \overline{\bW}_{is},
\boldf_{1}^*(\bR_{i1}),\ldots,\boldf_{s}^*(\bR_{is}),
\overline{\bR}_{is}, S_i = s]\\
= \ & \E[Y_i \mid \overline{\bW}_{is},
\boldf_{1}^*(\bR_{i1}),\ldots,\boldf_{s}^*(\bR_{is}), S_i = s]
\end{align*}
where the last equality follows from $Y_i \ \indep \ \overline{\bR}_{is} \mid  \overline{\bW}_{is}, \overline{\bU}_{is}, S_i = s$.  This implies that the inner expectation remains identical. 

Similarly,
\begin{align*}
&\int_{\mathcal{R}} \int_{\mathcal{W}} \E[Y_i \mid \overline{\bW}_{is}= \bar{\bm w}_{s}, \tilde\boldf_1(\bR_{i1}), \ldots,
\tilde\boldf_{s}(\bR_{is}), S_i = s] d Q_{s}(w_{s}; \overline{\bm{w}}_{s-1}, \delta_{s}) \\
&\qquad \times dF(\bR_{is} \mid \overline{\bm{W}}_{i,s-1} = \bar{\bm w}_{s-1}, \boldf_1(\bR_{i1}), \ldots,
\boldf_{s-1}(\bR_{i,s-1}), S_i=s)\\
= \ &  \E\biggl[ \int_{\mathcal{W}}  \E[Y_i \mid \overline{\bW}_{is}= \bar{\bm w}_{s}, \tilde\boldf_1(\bR_{i1}), \ldots,
\tilde\boldf_{s}(\bR_{is}), S_i = s] d Q_{s}(w_{s}; \overline{\bm{w}}_{s-1}, \delta_{s}) \\
& \qquad \qquad \biggr | \ \overline{\bm{W}}_{i,s-1} = \bar{\bm w}_{s-1}, \boldf_1(\bR_{i1}), \ldots,
\boldf_{s-1}(\bR_{i,s-1}), S_i=s  \biggr]\\
= \ & \E\biggl[ \int_{\mathcal{W}}  \E[Y_i \mid \overline{\bW}_{is}= \bar{\bm w}_{s}, \boldf_1^\ast(\bR_{i1}), \ldots,
\boldf_{s}^\ast(\bR_{is}), S_i = s] d Q_{s}(w_{s}; \overline{\bm{w}}_{s-1}, \delta_{s}) \\
&\qquad \qquad \biggr | \ \overline{\bm{W}}_{i,s-1} = \bar{\bm w}_{s-1}, \boldf_1(\bR_{i1}), \ldots,
\boldf_{s-1}(\bR_{i,s-1}), S_i=s  \biggr]\\
= \ & \E\biggl[ \int_{\mathcal{W}}  \E[Y_i \mid \overline{\bW}_{is}= \bar{\bm w}_{s}, \boldf_1^\ast(\bR_{i1}), \ldots,
\boldf_{s}^\ast(\bR_{is}), S_i = s] d Q_{s}(w_{s}; \overline{\bm{w}}_{s-1}, \delta_{s}) \\
&\qquad \qquad \biggr | \ \overline{\bm{W}}_{i,s-1} = \bar{\bm w}_{s-1}, \boldf_1^*(\bR_{i1}), \ldots,
\boldf_{s-1}^*(\bR_{i,s-1}), S_i=s  \biggr]
\end{align*}
where the first equality is due to rewriting the integral with the expectation operator, the second equality is proved above, and the last equality is due to the following derivation, which is similar to the one used above,
\begin{align*}
&\E\biggl[ \int_{\mathcal{W}}  \E[Y_i \mid \overline{\bW}_{is}= \bar{\bm w}_{s}, \boldf_1^\ast(\bR_{i1}), \ldots,
\boldf_{s}^\ast(\bR_{is}), S_i = s] d Q_{s}(w_{s}; \overline{\bW}_{i,s-1}, \delta_{s})\\
& \qquad \qquad \biggr | \ \overline{\bm{W}}_{i,s-1}, \boldf_1(\bR_{i1}), \ldots,
\boldf_{s-1}(\bR_{i,s-1}), S_i=s  \biggr]\\
= \ & \E\biggl[ \int_{\mathcal{W}}  \E[Y_i \mid \overline{\bW}_{is}= \bar{\bm w}_{s}, \boldf_1^\ast(\bR_{i1}), \ldots,
\boldf_{s}^\ast(\bR_{is}), S_i = s] d Q_{s}(w_{s}; \overline{\bW}_{i,s-1}, \delta_{s}) \\
& \qquad \qquad \biggr | \ \overline{\bm{W}}_{i,s-1}, \boldf_1(\bR_{i1}), \ldots,
\boldf_{s-1}(\bR_{i,s-1}), \overline{\bR}_{i,s-1}, S_i=s  \biggr]\\
= \ & \E\biggl[ \int_{\mathcal{W}}  \E[Y_i \mid \overline{\bW}_{is}= \bar{\bm w}_{s}, \boldf_1^\ast(\bR_{i1}), \ldots,
\boldf_{s}^\ast(\bR_{is}), S_i = s] d Q_{s}(w_{s}; \overline{\bW}_{i,s-1}, \delta_{s}) \\
&\qquad \qquad \biggr | \ \overline{\bm{W}}_{i,s-1}, \boldf_1^*(\bR_{i1}), \ldots,
\boldf_{s-1}^*(\bR_{i,s-1}), \overline{\bR}_{i,s-1}, S_i=s  \biggr]\\
= \ &  \int_{\mathcal{W}} \E\biggl[ \E[Y_i \mid \overline{\bW}_{is}= \bar{\bm w}_{s}, \boldf_1^\ast(\bR_{i1}), \ldots,
\boldf_{s}^\ast(\bR_{is}), S_i = s] \\
&\qquad \qquad \biggr | \ \overline{\bm{W}}_{i,s-1}, \boldf_1^*(\bR_{i1}), \ldots,
\boldf_{s-1}^*(\bR_{i,s-1}), \overline{\bR}_{i,s-1}, S_i=s  \biggr] d Q_{s}(w_{s}; \overline{\bW}_{i,s-1}, \delta_{s})\\
= \ & \int_{\mathcal{W}} \E\biggl[ \E[Y_i \mid \overline{\bW}_{is}= \bar{\bm w}_{s}, \boldf_1^\ast(\bR_{i1}), \ldots,
\boldf_{s}^\ast(\bR_{is}), S_i = s] \\
&\qquad \qquad \biggr | \ \overline{\bm{W}}_{i,s-1}, \boldf_1^*(\bR_{i1}), \ldots,
\boldf_{s-1}^*(\bR_{i,s-1}), S_i=s  \biggr] d Q_{s}(w_{s}; \overline{\bW}_{i,s-1}, \delta_{s}).
\end{align*}
where we use Equation~\eqref{deconfounder_independence} for the first equality, Fubini's theorem (because $Y_i$ is bounded) for the third equality, and Equation~\eqref{sufficiency_U_for_R} for the last equality.

By recursively applying the mean independence, we can finally show that
\begin{align*}
&\Psi_{s}(\bm\delta)\\
 = \ & \int_{\mathcal{W}^{s}} \int_{\cR^{s}} \E[  Y_i
         \mid \overline{\bm{W}}_{is} = \bar{\bm w}_{s},
         \boldf^*_{1}(\bR_{i1}),\ldots,\boldf^*_{s}(\bR_{is}),
         S_i = s] \\
    &\qquad \qquad \times \prod_{s^\prime = 1}^{s}  d Q_{s^\prime}(w_{s^\prime}; \overline{\bm{w}}_{s^\prime-1}, \delta_{s'}) dF(\bR_{is^\prime} \mid \overline{\bm{W}}_{i,s^\prime-1} = \bar{\bm w}_{s^\prime-1}, \boldf_{1}^*(\bR_{i1}),\ldots,\boldf_{s'-1}^*(\bR_{i,s'-1}), S_i=s)\\
= \ &  \int_{\mathcal{W}^{s}} \int_{\cR^{s}} \E[  Y_i
         \mid \overline{\bm{W}}_{is} = \bar{\bm w}_{s},
         \boldf_{1}^{[s]}(\bR_{i1}),\ldots,\boldf_{s}^{[s]}(\bR_{is}),
         S_i = s] \\
    &\qquad \qquad\times \prod_{s^\prime = 1}^{s}  d Q_{s^\prime}(w_{s^\prime}; \overline{\bm{w}}_{s^\prime-1}, \delta_{s'}) dF(\bR_{is^\prime} \mid \overline{\bm{W}}_{i,s^\prime-1} = \bar{\bm w}_{s^\prime-1}, \boldf_{1}^{[s'-1]}(\bR_{i1}),\ldots,\boldf_{s'-1}^{[s'-1]}(\bR_{i,s'-1}), S_i=s)
\end{align*}
for any set of deconfounder function $\boldf_s^{[s']}$ that satisfies the mean independence, is separable from the
treatment feature, and satisfies the bounded relative overlap.

\qed

\subsection{Proof of Theorem \ref{theorem_if} }\label{proof_if}

We first derive the influence function for $\Psi_{S_i}(\bm\delta)$ while
assuming that $S_i$ is fixed.  By Lemma~3 of
\citet{kennedy_nonparametric_2019}, the influence function for the
estimated stochastic intervention under the longitudinal setting is
given by
\begin{equation}
    \begin{aligned}
      &   \psi_{S_i}(\cD_i; \bm\delta, \{ \pi_s, m_{s}, \tilde m_{s},  p_s, \boldf^{[s]} \}_{s = 1}^{S_i}, \Psi) \\
= \  & 
\tilde{\psi}^*_{S_i}(\cD_i; \bm\delta, \{ \pi_s, m_{s}, p_s, \boldf^{[s]} \}_{s = 1}^{S_i}, \Psi_{S_i} )\\
&\ +
\sum_{s=1}^{S_i}
\sum_{w\in\{0,1\}}
\
\phi(\overline{\bm W}_{i,s-1}, W_{is}; w, \delta_s) \\
&\hspace{1in} \times
\E\Biggl[
\Biggl(
\prod_{s^\prime=1}^{s-1}
\omega_{s^\prime}(\overline{\bm{H}}_{is^\prime}, W_{is^\prime}; \delta_{s^\prime}, p_{s^\prime}, \pi_{s^\prime}  )
\Biggr)
m_s(\overline{\bW}_{i,s-1}, w, \overline{\bm F}_{is}; \bm\delta_{\underline{s}}) \mid \overline{\bm W}_{i,s-1}, S_i \geq s\Biggr], \label{kennedy_lemma3}
    \end{aligned}
\end{equation}
where $\tilde{\psi}^*_{S_i}$ is the standard longitudinal influence
function under a known stochastic intervention that does not depend on
the data distribution (the exact expression will be given below).
Note that $\phi$ (the exact expression will again be given below) is
part of the influence function of the stochastic intervention, which
is given by
$\frac{\mathbbm{1}\{\overline{\bm W}_{i,s-1} = \overline{\bm
    w}_{s-1}, S_i \geq s\}}{\P(\overline{\bm W}_{i,s-1} = \overline{\bm
    w}_{s-1}, S_i \geq s)} \phi(\overline{\bm W}_{i,s-1}, W_{is}; w, \delta_s)$.

\begin{comment}
\begin{align*}
    \omega_s
    & := \frac{
W_{is} q_s(\overline{\bm W}_{i,s-1}) +
(1 - W_{is}) (1- q_s(\overline{\bm W}_{i,s-1}))
}{
W_{is} \pi_s(\overline{\bm H}_{is}) +
(1 - W_{is}) (1- \pi_s(\overline{\bm H}_{is}))
}\\
&= \frac{ \delta W_{is} \frac{p_s(\overline{\bm W}_{i,s-1})}{\pi_s(\overline{\bm H}_{is})} + (1 - W_{is}) \frac{1 - p_s(\overline{\bm W}_{i,s-1})}{1 - \pi_s(\overline{\bm H}_{is})} }{ \delta p_s(\overline{\bm W}_{i,s-1}) + (1 - p_s(\overline{\bm W}_{i,s-1})) } 
, 
\end{align*}
with
$\pi_s(\overline{\bm H}_{is}) := \P(W_{is} = 1 \mid \overline{\bm
  H}_{is})$ denoting the propensity score.
\end{comment}

From \cite{rotnitzky2017multiply}, the closed-form expression of
$\tilde{\psi}^*_{S_i}$ is given by
\begin{align*}
    &\tilde{\psi}^*_{S_i}(\cD_i; \bm\delta, \{ \pi_s, m_{s}, p_s, \boldf^{[s]} \}_{s = 1}^{S_i}, \Psi_{S_i} )\\
   = \ & \sum_{s= 1}^{S_i} \Biggl[   \frac{\delta_s  p_s(\overline{\bm W}_{i,s-1}) m_s(\overline{\bW}_{i,s-1}, 1, \overline{\bm F}_{is}; \bm\delta_{\underline{s}}) + 
    \{1 - p_s(\overline{\bm W}_{i,s-1})\} m_s(\overline{\bW}_{i,s-1}, 0, \overline{\bm F}_{is}; \bm\delta_{\underline{s}})
    }{\delta_s  p_s(\overline{\bm W}_{i,s-1}) + \{1 - p_s(\overline{\bm W}_{i,s-1})\}}\\
&  \hspace{.25in} - m_s(\overline{\bW}_{i,s-1}, W_{is}, \overline{\bm F}_{is}; \bm\delta_{\underline{s}}) \times \frac{
\delta_s W_{is}  \frac{p_s(\overline{\bm W}_{i,s-1})}{\pi_s(\overline{\bm H}_{is})} + (1 - W_{is}) \frac{1 - p_s(\overline{\bm W}_{i,s-1})}{1 - \pi_s(\overline{\bm H}_{is})} }{ \delta_s p_s(\overline{\bm W}_{i,s-1}) + \{1 - p_s(\overline{\bm W}_{i,s-1})\} }  \Biggr]  \prod_{s^\prime = 1}^{s-1} \omega_{s^\prime}(\overline{\bm{H}}_{is^\prime}, W_{is^\prime};\delta_{s^\prime}, p_{s^\prime}, \pi_{s^\prime} )\\
&\hspace{0.25in} + \prod_{s= 1}^{S_i} \omega_s(\overline{\bm{H}}_{is}, W_{is};\delta_s, p_{s}, \pi_{s} )  Y_i - \Psi_{S_i}(\bm\delta)\\
= \ & \sum_{s= 1}^{S_i} \Biggl[   \frac{\delta_s  p_s(\overline{\bm W}_{i,s-1}) m_s(\overline{\bW}_{i,s-1}, 1, \overline{\bm F}_{is}; \bm\delta_{\underline{s}}) + 
    \{1 - p_s(\overline{\bm W}_{i,s-1})\} m_s(\overline{\bW}_{i,s-1}, 0, \overline{\bm F}_{is}; \bm\delta_{\underline{s}})
    }{\delta_s  p_s(\overline{\bm W}_{i,s-1}) + \{1 - p_s(\overline{\bm W}_{i,s-1})\}}\\
& \hspace{.10in}  -  \frac{\begin{array}{@{}l@{}} \delta_s  m_s(\overline{\bW}_{i,s-1}, 1, \overline{\bm F}_{is}; \bm\delta_{\underline{s}})  W_{is}  \frac{p_s(\overline{\bm W}_{i,s-1})}{\pi_s(\overline{\bm H}_{is})} \\[2pt] \quad + m_s(\overline{\bW}_{i,s-1}, 0, \overline{\bm F}_{is}; \bm\delta_{\underline{s}})  (1 - W_{is}) \frac{1 - p_s(\overline{\bm W}_{i,s-1})}{1 - \pi_s(\overline{\bm H}_{is})} \end{array} }{ \delta_s p_s(\overline{\bm W}_{i,s-1}) + \{1 - p_s(\overline{\bm W}_{i,s-1})\} }  \Biggr]\\
&\hspace{0.25in} \times \prod_{s^\prime = 1}^{s-1} \omega_{s^\prime}(\overline{\bm{H}}_{is^\prime}, W_{is^\prime};\delta_{s^\prime}, p_{s^\prime}, \pi_{s^\prime} ) + \prod_{s= 1}^{S_i} \omega_s(\overline{\bm{H}}_{is}, W_{is};\delta_s, p_{s}, \pi_{s} )  Y_i - \Psi_{S_i}(\bm\delta)\\
= \ & \sum_{s= 1}^{S_i} \left[   \frac{\begin{array}{@{}l@{}} \delta_s  p_s(\overline{\bm W}_{i,s-1}) m_s(\overline{\bW}_{i,s-1}, 1, \overline{\bm F}_{is}; \bm\delta_{\underline{s}}) \left\{1 - \frac{W_{is}}{\pi_s(\overline{\bm H}_{is})}\right\} \\[2pt]
    \quad + \{1 - p_s(\overline{\bm W}_{i,s-1})\} m_s(\overline{\bW}_{i,s-1}, 0, \overline{\bm F}_{is}; \bm\delta_{\underline{s}})  \left\{1 - \frac{1 - W_{is}}{1 - \pi_s(\overline{\bm H}_{is})}\right\} \end{array}
    }{\delta_s  p_s(\overline{\bm W}_{i,s-1}) + \{1 - p_s(\overline{\bm W}_{i,s-1})\}} \right]\\
&\qquad \times \prod_{s^\prime = 1}^{s-1} \omega_{s^\prime}(\overline{\bm{H}}_{is^\prime}, W_{is^\prime};\delta_{s^\prime}, p_{s^\prime}, \pi_{s^\prime} ) + \prod_{s= 1}^{S_i} \omega_s(\overline{\bm{H}}_{is}, W_{is};\delta_{s}, p_{s}, \pi_{s} )  Y_i - \Psi_{S_i}(\bm\delta).
\end{align*}
\begin{comment}
%omitted since we defined in the main text
where we recursively define $m_{S_i}(\overline{\bm{H}}_{iS_i}, \overline{\bm{W}}_{S_i}; \overline{\bm{\delta}}_{S_i + 1}) =
\mu_{\overline{\bm{W}}_{S_i}}(\{\boldf(s,\bR_{is};
\blambda)\}_{s=1}^{S_i})$ and 
\begin{align*}
    &m_{s}(\overline{\bm{H}}_{is}, W_{is};\delta)
    = \E\left[\frac{\delta p_{s+1}(\overline{\bm
      W}_{is})m_{s+1}(\overline{\bm{H}}_{i,s+1}, 1;\delta) + (1 - p_{s+1}(\overline{\bm
      W}_{is}))m_{s+1}(\overline{\bm{H}}_{i,s+1}, 0;\delta)  }{\delta  p_{s+1}(\overline{\bm
      W}_{is}) + (1 - p_{s+1}(\overline{\bm W}_{is}))} \ \biggl | \  \overline{\bm{H}}_{is}, W_{is} \right]
\end{align*}
for $s=1,\ldots,S_i-1$.
\end{comment}

For binary $w_s\in\{0,1\}$, the stochastic intervention distribution
satisfies,
\[
dQ_s(w_s; \overline{\bm w}_{s-1}, \delta_s)
= \frac{w_s \cdot \delta_s  p_s(\overline{\bm w}_{s-1}) + (1 - w_s)\{1 - p_s(\overline{\bm w}_{s-1})\} }{\delta_s  p_s(\overline{\bm w}_{s-1}) + (1 - p_s(\overline{\bm w}_{s-1}))},
\]
where
$p_s(\overline{\bm w}_{s-1})=\P(W_{is}=1\mid \overline{\bm
  W}_{i,s-1} = \overline{\bm w}_{s-1}, S_i \geq s)$.  Since $Q_s$
depends on the data only through $p_s$, via the chain rule, the
influence function of $Q_s$ is the product of two terms;
(i) the derivative of $Q_s$ with respect to $p_s$ and (ii) the
influence function of $p_s$. Now, the derivative of $Q_s$ with respect to $p_s$ is given by
\begin{align*}
    \frac{d Q_s(w_s; \overline{\bm w}_{s-1}, \delta_s)
  }{d p_s(\overline{\bm w}_{s-1})} =
  \frac{(2w_s-1)\delta_s}{\{\delta_s p_s(\overline{\bm w}_{s-1}) + (1 -
  p_s(\overline{\bm w}_{s-1}))\}^2 }, 
\end{align*}
whereas Example~6 of \cite{hines_demystifying_2022} gives the
influence function of
$p_s(\overline{\bm w}_{s-1})=\P(W_{is}=1\mid \overline{\bm
  W}_{i,s-1} = \overline{\bm w}_{s-1}, S_i \geq s)$ as,
\begin{align*}
    \frac{\mathbbm{1}\{\overline{\bm W}_{i,s-1} = \overline{\bm w}_{s-1}, S_i \geq s\}}{\P(\overline{\bm W}_{i,s-1} = \overline{\bm w}_{s-1}, S_i \geq s)} \biggl( W_{is} - p_{s}(\overline{\bm w}_{s-1})\biggr).
\end{align*}
Therefore, we obtain,
\[
\phi(\overline{\bm W}_{i,s-1},W_{is}; w_s, \delta_s)
=  \frac{ (2w_s -1)\delta_s (W_{is} - p_s(\overline{\bm W}_{i,s-1})) }{[\delta_s p_s(\overline{\bm W}_{i,s-1}) + \{1 - p_s(\overline{\bm W}_{i,s-1})\}]^2  }.
\]
Substituting this expression into Equation \eqref{kennedy_lemma3} yields
\begin{align*}
&\psi_{S_i}(\cD_i; \bm\delta, \{ \pi_s, m_{s}, \tilde m_{s},  p_s, \boldf^{[s]} \}_{s = 1}^{S_i}, \Psi)\\
= \ & \sum_{s= 1}^{S_i} \left[ \frac{\begin{array}{@{}l@{}} \delta_s  p_s(\overline{\bm W}_{i,s-1}) m_s(\overline{\bW}_{i,s-1}, 1, \overline{\bm F}_{is}; \bm\delta_{\underline{s}}) \left\{1 - \frac{W_{is}}{\pi_s(\overline{\bm H}_{is})}\right\} \\[2pt]
    \quad + \{1 - p_s(\overline{\bm W}_{i,s-1})\} m_s(\overline{\bW}_{i,s-1}, 0, \overline{\bm F}_{is}; \bm\delta_{\underline{s}}) \left\{1 - \frac{1 - W_{is}}{1 - \pi_s(\overline{\bm H}_{is})}\right\} \end{array}
    }{\delta_s  p_s(\overline{\bm W}_{i,s-1}) + \{1 - p_s(\overline{\bm W}_{i,s-1})\}}  \right] \\
    &\   \times \prod_{s^\prime = 1}^{s-1} \omega_{s^\prime}(\overline{\bm{H}}_{is^\prime}, W_{is^\prime};\delta_{s^\prime} ) + \prod_{s= 1}^{S_i} \omega_s(\overline{\bm{H}}_{is}, W_{is};\delta_{s}) Y_i\\
    &\   + \sum_{s=1}^{S_i}
\Biggl\{
\sum_{w\in\{0,1\}}
\frac{(2w-1)\delta_s \{W_{is} - p_s(\overline{\bm W}_{i,s-1})\} }{[\delta_s p_s(\overline{\bm W}_{i,s-1}) + \{1 - p_s(\overline{\bm W}_{i,s-1})\}]^2  }\\
&\hspace{1.5in} \times \E\biggl[\Biggl(
\prod_{s^\prime=1}^{s-1}
\omega_{s^\prime}(\overline{\bm{H}}_{is^\prime}, W_{is^\prime};\delta_{s^\prime} )
\Biggr) m_s(\overline{\bW}_{i,s-1}, w, \overline{\bm F}_{is}; \bm\delta_{\underline{s}}) \mid \overline{\bm W}_{i,s-1}, S_i \geq s\biggr]
      \Biggr\} \\
  & \ \ 
     - \Psi_{S_i}(\bm\delta)\\
= \ & \sum_{s = 1}^{S_i}\Biggl(
\prod_{s^\prime=1}^{s-1}
\omega_{s^\prime}(\overline{\bm{H}}_{is^\prime}, W_{is^\prime};\delta_{s^\prime},
p_{s^\prime}, \pi_{s^\prime}
)
\Biggr)\\
&\times \left[
 \frac{\begin{array}{@{}l@{}} \delta_s  p_s(\overline{\bm W}_{i,s-1}) m_s(\overline{\bW}_{i,s-1}, 1, \overline{\bm F}_{is}; \bm\delta_{\underline{s}}) \left\{1 - \frac{W_{is}}{\pi_s(\overline{\bm H}_{is})}\right\} \\[2pt]
    \quad + (1 - p_s(\overline{\bm W}_{i,s-1})) m_s(\overline{\bW}_{i,s-1}, 0, \overline{\bm F}_{is}; \bm\delta_{\underline{s}}) \left\{1 - \frac{1 - W_{is}}{1 - \pi_s(\overline{\bm H}_{is})}\right\} \end{array}
    }{\delta_s  p_s(\overline{\bm W}_{i,s-1}) + \{1 - p_s(\overline{\bm
      W}_{i,s-1})\}}\right] \\
    &  + \frac{\delta_s (W_{is} - p_s(\overline{\bm W}_{i,s-1})) \{ \tilde m_s(\overline{\bm W}_{i,s-1},  1; \bm\delta_{-s}) - \tilde m_s(\overline{\bm W}_{i,s-1},  0; \bm\delta_{-s}) \} }{ [\delta_s  p_s(\overline{\bm W}_{i,s-1}) + \{1 - p_s(\overline{\bm W}_{i,s-1})\}]^2 } + \prod_{s= 1}^{S_i} \omega_s(\overline{\bm{H}}_{is}, W_{is};\delta_{s}, p_{s}, \pi_{s}) Y_i -  \Psi_{S_i}(\bm\delta).
\end{align*}
Finally, we derive the desired expression for $\psi(\cD_i; \bm\delta, \{ \pi_s, m_{s}, \tilde m_{s}, p_s, \boldf^{[s]} \}_{s= 1}^{S_i}, \Psi)$. Recall that
\begin{align*}
    \Psi(\bm\delta) = \int \Psi_{S_i}(\bm\delta)dP(S_i).
\end{align*}
By Example~1 of \cite{hines_demystifying_2022} and the chain rule, we have,
\begin{align*}
    \psi(\cD_i; \bm\delta, \{ \pi_s, m_{s}, \tilde m_{s},  p_s, \boldf^{[s]} \}_{s = 1}^{S_i}, \Psi) = \psi_{S_i}(\cD_i; \bm\delta, \{ \pi_s, m_{s}, \tilde m_{s}, p_s, \boldf^{[s]}\}_{s = 1}^{S_i}, \Psi_{S_i} ) + \Psi_{S_i}(\bm\delta) - \Psi(\bm\delta).
\end{align*}
Therefore,
\begin{align*}
    &\psi(\cD_i; \bm\delta, \{ \pi_s, m_{s}, \tilde m_{s}, p_s, \boldf^{[s]} \}_{s = 1}^{S_i}, \Psi)\\
  = \   & \sum_{s = 1}^{S_i}\Biggl(
\prod_{s^\prime=1}^{s-1}
\omega_{s^\prime}(\overline{\bm{H}}_{is^\prime}, W_{is^\prime};\delta_{s^\prime}, p_{s^\prime}, \pi_{s^\prime} )
\Biggr)\\
&\ \  \times \left[
 \frac{\begin{array}{@{}l@{}} \delta_s  p_s(\overline{\bm W}_{i,s-1}) m_s(\overline{\bW}_{i,s-1}, 1, \overline{\bm F}_{is}; \bm\delta_{\underline{s}}) \left\{1 - \frac{W_{is}}{\pi_s(\overline{\bm H}_{is})}\right\} \\[2pt]
    \quad + \{1 - p_s(\overline{\bm W}_{i,s-1})\} m_s(\overline{\bW}_{i,s-1}, 0, \overline{\bm F}_{is}; \bm\delta_{\underline{s}}) \left\{1 - \frac{1 - W_{is}}{1 - \pi_s(\overline{\bm H}_{is})}\right\} \end{array}
    }{\delta_s  p_s(\overline{\bm W}_{i,s-1}) + 1 - p_s(\overline{\bm
      W}_{i,s-1})} \right] \\
    &\ \ + \sum_{s=1}^{S_i} \frac{\delta_s (W_{is} - p_s(\overline{\bm W}_{i,s-1})) \{ \tilde m_s(\overline{\bm W}_{i,s-1},  1; \bm\delta_{-s}) - \tilde m_s(\overline{\bm W}_{i,s-1},  0; \bm\delta_{-s}) \} }{ [\delta_s  p_s(\overline{\bm W}_{i,s-1}) + \{1 - p_s(\overline{\bm W}_{i,s-1})\}]^2 }\\
    & \ \ + \prod_{s= 1}^{S_i} \omega_s(\overline{\bm{H}}_{is}, W_{is};\delta_{s}, p_{s}, \pi_{s} ) Y_i -  \Psi(\bm\delta).
\end{align*}
\qed

\subsection{Proof of Theorem \ref{asymp_normal} }\label{proof_asymp_normal}

Let the oracle empirical process indexed by $\bm\delta$ be
\begin{align*}
    \mathbb{G}_n(\bm\delta) := \frac{1}{\sqrt{N}}\sum_{i = 1}^N f_{\bm\delta}(\cD_i) \qquad \text{where} \quad f_{\bm\delta}(\cD_i) =  \frac{ \psi(\cD_i; \bm\delta, \{ \pi_s, m_{s}, \tilde m_{s}, p_s, \boldf^{[s]}
      \}_{s= 1}^{S_i}, \Psi) }{\sigma(\bm\delta) }. 
\end{align*}
We will prove that
\begin{align}
& \mathbb{G}_n(\bm\delta) \xrightarrow{d} \mathbb{G}(\bm\delta) \text{ in } \ell^\infty(\Delta^{s_{\max}}) \label{statement2}\\
&\sup_{\bm\delta \in \Delta^{s_{\max}} }\biggl| \sqrt{N}\frac{\hat \Psi_n(\bm\delta) - \Psi(\bm\delta)}{{\sigma}(\bm\delta)} - \mathbb{G}_n(\bm\delta) \biggr| = o_p(1) \label{statement1}
\end{align}
where $\ell^\infty(\Delta^{s_{\max}})$ is the space of bounded functions in $\Delta^{s_{\max}}$. Together, this proves the desired statement.

\subsubsection{Proof of Equation~\eqref{statement2}}

Proving Equation~\eqref{statement2} is equivalent to proving the
function class
$\mathcal{F} = \{f_{\bm\delta}(\cD_i): \bm\delta \in \Delta^{s_{\max}}
\}$ is Donsker. It is sufficient to show that the bracketing integral
is finite (Theorem~2.5.6 of \citealt{van1996weak}); i.e.,
\begin{align*}
    \int_{0}^1 \sqrt{1 + \log N_{[]}(\epsilon \norm{F}_{2}, \mathcal{F}, L_2(\P))} d \epsilon < \infty
\end{align*}
where $N_{[]}(\epsilon \norm{F}_{2}, \mathcal{F}, L_2(\P))$ is the
bracketing number, which is the minimum number of
$\epsilon$-bracketing to cover the function class $\mathcal{F}$, with
an envelope function $F$ (i.e., any function such that
$|f_{\bm\delta}(\cD_i)| \leq F(\cD_i)$ for every $\cD_i$ and
$\bm\delta \in \Delta^{s_{\max}}$).

Because $\inf_{\bm\delta\in\Delta^{s_{\max}}}\sigma(\bm\delta)>0$, it
is sufficient to show that each component of the influence function is
uniformly Lipschitz in $\bm\delta$ with an $L_2(P)$ envelope.  We
first consider the derivative of the weight,
\begin{align*}
&\sup_{\delta_s \in \Delta} \biggl|\frac{d}{d\delta_s}
                 \omega_s(\overline{\bm{H}}_{is}, W_{is};\delta_s, p_{s}, \pi_{s}
                 )\biggr| \\
 = &  \sup_{\delta_s \in \Delta} \left|\frac{d}{d\delta_s}\left[\frac{ \delta_s W_{is} \frac{p_s(\overline{\bm W}_{i,s-1})}{\pi_s(\overline{\bm H}_{is})} + (1 - W_{is}) \frac{1 - p_s(\overline{\bm W}_{i,s-1})}{1 - \pi_s(\overline{\bm H}_{is})} }{ \delta_s p_s(\overline{\bm W}_{i,s-1}) + \{1 - p_s(\overline{\bm W}_{i,s-1})\} } \right] \right|\\
= & \sup_{\delta_s \in \Delta}\left|\frac{p_s(\overline{\bm W}_{i,s-1}) \{1 - p_s(\overline{\bm W}_{i,s-1})\} }{ [\delta_s p_s(\overline{\bm W}_{i,s-1}) + \{1 - p_s(\overline{\bm W}_{i,s-1})\}]^2 } \left(\frac{W_{is}}{\pi_s(\overline{\bm H}_{is})} - \frac{1 - W_{is} }{1 - \pi_s(\overline{\bm H}_{is})}\right) \right|\\
= & \sup_{\delta_s \in \Delta}\left|\frac{1 }{ [\delta_s p_s(\overline{\bm W}_{i,s-1}) + \{1 - p_s(\overline{\bm W}_{i,s-1})\}]^2 } \right| \\
&\qquad \times  \left|p_s(\overline{\bm W}_{i,s-1}) \{1 - p_s(\overline{\bm W}_{i,s-1})\}\left( \frac{W_{is}}{\pi_s(\overline{\bm H}_{is})} - \frac{1 - W_{is} }{1 - \pi_s(\overline{\bm H}_{is})} \right)\right|\\
\leq &  \frac{1}{ \min\{1, \delta_\ell\}^2 } \left|
\frac{W_{is}p_s(\overline{\bm W}_{i,s-1}) \{1 - p_s(\overline{\bm W}_{i,s-1})\}}{\pi_s(\overline{\bm H}_{is})} - \frac{(1 - W_{is})p_s(\overline{\bm W}_{i,s-1}) \{1 - p_s(\overline{\bm W}_{i,s-1})\} }{1 - \pi_s(\overline{\bm H}_{is})}
\right|\\
\leq & \frac{1}{ \min\{1, \delta_\ell\}^2 } \left\{\left|W_{is}
\frac{p_s(\overline{\bm W}_{i,s-1})}{\pi_s(\overline{\bm H}_{is})}\right| +\left|(1 - W_{is})\frac{ 1 - p_s(\overline{\bm W}_{i,s-1}) }{1 - \pi_s(\overline{\bm H}_{is})}
\right| \right\}\\
\leq & \frac{2}{ c\min\{1, \delta_\ell\}^2 } 
\end{align*}
where the first inequality is due to
$\delta_s p_s(\overline{\bm W}_{i,s-1}) + \{1 - p_s(\overline{\bm
  W}_{i,s-1})\} \geq \min\{1, \delta_\ell\}$ and the last inequality
follows from Assumption~\ref{overlap}. Also, we have,
\begin{align*}
    &\sup_{\delta_s \in \Delta}|\omega_s(\overline{\bm{H}}_{is}, W_{is};\delta_s, p_{s}, \pi_{s} )| = \sup_{\delta_s \in \Delta}\left|\frac{ \delta_s W_{is} \frac{p_s(\overline{\bm W}_{i,s-1})}{\pi_s(\overline{\bm H}_{is})} + (1 - W_{is}) \frac{1 - p_s(\overline{\bm W}_{i,s-1})}{1 - \pi_s(\overline{\bm H}_{is})} }{ \delta_s p_s(\overline{\bm W}_{i,s-1}) + (1 - p_s(\overline{\bm W}_{i,s-1})) } \right| \leq \frac{\delta_u/c + 1/c }{\min\{1, \delta_\ell\}}.
\end{align*}
The derivative of the product term is written as,
\begin{align*}
    &\sup_{\bm\delta \in \Delta^{s_{\max}}  } \biggl|\frac{d}{d\bm\delta} \prod_{s^\prime=1}^{s}
\omega_{s^\prime}(\overline{\bm{H}}_{is^\prime}, W_{is^\prime};\delta_{s^\prime}, p_{s^\prime}, \pi_{s^\prime})\biggr| \\
&= \sup_{\bm\delta \in \Delta^{s_{\max}}} \left| \sum_{s^\prime= 1}^s \biggl( \frac{d}{d\bm \delta} \omega_{s^\prime}(\overline{\bm{H}}_{is^\prime}, W_{is^\prime};\delta_{s^\prime}, p_{s^\prime}, \pi_{s^\prime} ) \biggr) \prod_{r \neq s'}
\omega_{r}(\overline{\bm{H}}_{ir}, W_{ir};\delta_r ) \right|. 
\end{align*}
Since each term is bounded, the product is uniformly Lipschitz.
Applying the same argument, we can show that
$\prod_{s= 1}^{S_i} \omega_s(\overline{\bm{H}}_{is}, W_{is};\delta_s )
Y_i$ is uniformly Lipschitz in $\bm\delta$ because $Y_i$ is bounded by
Assumption~\ref{reg_cond}.

Similarly,
\begin{align*}
     &\sup_{\delta_s \in \Delta} \left| \frac{\partial}{\partial \delta_s} \frac{\begin{array}{@{}l@{}} \delta_s  p_s(\overline{\bm W}_{i,s-1}) m_s(\overline{\bW}_{i,s-1}, 1, \overline{\bm F}_{is}; \bm\delta_{\underline{s}}) \left\{1 - \frac{W_{is}}{\pi_s(\overline{\bm H}_{is})}\right\} \\[2pt]
    \quad + \{1 - p_s(\overline{\bm W}_{i,s-1})\} m_s(\overline{\bW}_{i,s-1}, 0, \overline{\bm F}_{is}; \bm\delta_{\underline{s}}) \left\{1 - \frac{1 - W_{is}}{1 - \pi_s(\overline{\bm H}_{is})}\right\} \end{array}
    }{\delta_s  p_s(\overline{\bm W}_{i,s-1}) + \{1 - p_s(\overline{\bm
      W}_{i,s-1})\}} \right|\\
    = & \sup_{\delta_s \in \Delta} \left| \frac{\begin{array}{@{}l@{}}  p_s(\overline{\bm W}_{i,s-1})(1 - p_s(\overline{\bm W}_{i,s-1})) \biggl( m_s(\overline{\bW}_{i,s-1}, 1, \overline{\bm F}_{is}; \bm\delta_{\underline{s}}) \left\{1 - \frac{W_{is}}{\pi_s(\overline{\bm H}_{is})}\right\} \\[2pt]
    \quad -
    m_s(\overline{\bW}_{i,s-1}, 0, \overline{\bm F}_{is}; \bm\delta_{\underline{s}}) \left\{1 - \frac{1 - W_{is}}{1 - \pi_s(\overline{\bm H}_{is})} \right\} \biggr) \end{array}
    }{ [\delta_s  p_s(\overline{\bm W}_{i,s-1}) + \{1 - p_s(\overline{\bm
      W}_{i,s-1}) \}]^2 } \right|\\
\leq & \frac{1}{\min\{1, \delta_\ell\}^2 }\left|
             p_s(\overline{\bm W}_{i,s-1})(1 - p_s(\overline{\bm
             W}_{i,s-1})) \left( m_s(\overline{\bW}_{i,s-1}, 1, \overline{\bm F}_{is}; \bm\delta_{\underline{s}}) \left\{1 -
             \frac{W_{is}}{\pi_s(\overline{\bm H}_{is})}\right\}
             \right. \right. \\
             & \hspace{1in}\left. - 
    m_s(\overline{\bW}_{i,s-1}, 0, \overline{\bm F}_{is}; \bm\delta_{\underline{s}}) \left\{1 - \frac{1 - W_{is}}{1 - \pi_s(\overline{\bm H}_{is})} \right\} \biggr)
     \right|\\
    \leq & \frac{C}{\min\{1, \delta_\ell\}^2 }p_s(\overline{\bm W}_{i,s-1})\{1 - p_s(\overline{\bm W}_{i,s-1})\} \left( \left| 1-  \frac{1 - W_{is}}{1 - \pi_s(\overline{\bm H}_{is})} \biggr| + \biggl|1 -  \frac{W_{is}}{\pi_s(\overline{\bm H}_{is})}  \right|\right)\\
    \leq & \frac{C}{\min\{1, \delta_\ell\}^2 } \left( 2 + \left|  \frac{1 - p_s(\overline{\bm W}_{i,s-1}) }{1 - \pi_s(\overline{\bm H}_{is})}  \right| + \left|  \frac{p_s(\overline{\bm W}_{i,s-1})  }{\pi_s(\overline{\bm H}_{is})} \right| \right)
            \\
  \leq & \frac{2C(1 + 1/c) }{\min\{1, \delta_\ell\}^2},
\end{align*}
where the second inequality is due to the second condition in
Assumption~\ref{reg_cond}.

Finally, 
{\small \begin{align*}
    &\sup_{\delta_s \in \Delta} \left| \frac{\partial}{\partial
      \delta_s}\frac{\delta_s (W_{is} - p_s(\overline{\bm
      W}_{i,s-1})) \{ \tilde m_s(\overline{\bm W}_{i,s-1},
      1; \bm\delta_{-s}) - \tilde m_s(\overline{\bm W}_{i,s-1},
      0; \bm\delta_{-s}) \} }{ [\delta_s  p_s(\overline{\bm
      W}_{i,s-1}) + \{1 - p_s(\overline{\bm W}_{i,s-1})\}]^2 } \right|\\
    = &  \sup_{\delta_s \in \Delta} 
    \left| 
    \frac{ \{W_{is} - p_s(\overline{\bm W}_{i,s-1})\} \{ \tilde m_s(\overline{\bm W}_{i,s-1},
      1; \bm\delta_{-s}) - \tilde m_s(\overline{\bm W}_{i,s-1},
      0; \bm\delta_{-s}) \} \{ 1 - (1 + \delta_s) p_s(\overline{\bm W}_{i,s-1}) \} }{ 
    \{\delta_s  p_s(\overline{\bm W}_{i,s-1}) + (1 - p_s(\overline{\bm W}_{i,s-1}))\}^3
    }
    \right|\\
    \leq &  \sup_{\delta_s \in \Delta} 
    \left| 
    \frac{ \{ \tilde m_s(\overline{\bm W}_{i,s-1},
      1; \bm\delta_{-s}) - \tilde m_s(\overline{\bm W}_{i,s-1},
      0; \bm\delta_{-s}) \} \{ 1 - (1 + \delta_s) p_s(\overline{\bm W}_{i,s-1}) \} }{ 
    \{\delta_s  p_s(\overline{\bm W}_{i,s-1}) + (1 - p_s(\overline{\bm W}_{i,s-1}))\}^3
    }
    \right|\\
    \leq & \sup_{\delta_s \in \Delta} 
    \left| 
    \frac{ \{ \tilde m_s(\overline{\bm W}_{i,s-1},
      1; \bm\delta_{-s}) - \tilde m_s(\overline{\bm W}_{i,s-1},
      0; \bm\delta_{-s}) \} \max\{1, \delta_u\} }{ 
    \{\delta_s  p_s(\overline{\bm W}_{i,s-1}) + (1 - p_s(\overline{\bm W}_{i,s-1}))\}^3
    }
    \right|\\
    \leq & \sup_{\delta_s \in \Delta} 
    \frac{ \{ |\tilde m_s(\overline{\bm W}_{i,s-1},
      1; \bm\delta_{-s})| + |\tilde m_s(\overline{\bm W}_{i,s-1},
      0; \bm\delta_{-s})| \} \max\{1, \delta_u\} }{
    \{\delta_s  p_s(\overline{\bm W}_{i,s-1}) + (1 - p_s(\overline{\bm W}_{i,s-1}))\}^3
    }\\
    \leq & \frac{1}{\min\{1, \delta_\ell\}^3}
    \{ |\tilde m_s(\overline{\bm W}_{i,s-1},
      1; \bm\delta_{-s})| + |\tilde m_s(\overline{\bm W}_{i,s-1},
      0; \bm\delta_{-s})| \} \max\{1, \delta_u\} \\
    \leq & 
    \frac{ 2C \max\{1, \delta_u\} }{ 
           \min\{1, \delta_\ell\}^3},
        \end{align*}}where the first inequality is due to the fact that $|W_{is} -
      p_s(\overline{\bm W}_{i,s-1})| \leq 1$, the second inequality holds because 
\begin{align*}
    |1 - (1 + \delta_s) p_s(\overline{\bm W}_{i,s-1})| \leq |1 -  p_s(\overline{\bm W}_{i,s-1})| + |\delta_s  p_s(\overline{\bm W}_{i,s-1})|  \leq \max\{1, \delta_u\},
\end{align*}
and the last inequality is due to the second condition in
Assumption~\ref{reg_cond}.  Since each component of the influence
function is uniformly Lipschitz in $\bm\delta$, we have proven
Equation~\eqref{statement2}.

\subsubsection{Proof of Equation~\eqref{statement1}}
Recall that from the proof of Theorem~\ref{theorem_if}, $\tilde{\psi}^*$ denotes the influence function of the marginal structural model derived in the literature (e.g., \citealt{rotnitzky2017multiply}), and let $\zeta$ be the augmentation term defined as
\begin{equation}
\begin{aligned}
    &\zeta(\cD_i; \bm\delta, \{ \pi_s, m_s, \tilde m_s, p_s, \boldf^{[s]}\}_{s = 1}^{S_i})  \\
    &:= \sum_{s=1}^{S_i}
\Biggl\{
\sum_{w\in\{0,1\}}
\phi(\overline{\bm W}_{i,s-1}, W_{is}; w, \delta_s) \\
&\hspace{1in} \times
\E\biggl[\Biggl(
\prod_{s^\prime=1}^{s-1}
\omega_{s^\prime}(\overline{\bm{H}}_{is^\prime}, W_{is^\prime};\delta_{s^\prime}, p_{s^\prime}, \pi_{s^\prime})
\Biggr) m_s(\overline{\bW}_{i,s-1}, w, \overline{\bm F}_{is}; \bm\delta_{\underline{s}}) \mid \overline{\bm W}_{i,s-1}, S_i \geq s\biggr]
\Biggr\}. \label{augmentation_term}
\end{aligned}
\end{equation}
Note that the conditional expectation in
Equation~\eqref{augmentation_term} is exactly the marginalized outcome
regression $\tilde m_s(\overline{\bm W}_{i,s-1}, w; \bm\delta_{-s})$
defined in Theorem~\ref{theorem_if}, so $\zeta$ depends on $\tilde
m_s$; when $\zeta$ is evaluated at the estimated nuisance functions,
this term is replaced by the saturated regression estimator
$\widehat{\tilde m}_s^{(-k)}$.
Also, we write the estimand as $\Psi(\bm\delta) = \Psi(\bm\delta, Q)$ to explicitly denote which stochastic intervention $Q$ we use to calculate $\Psi$. By Lemma~\ref{lemma1}, we have
\begin{align*}
       &\sup_{\bm\delta \in \Delta^{s_{\max}} }\biggl| \sqrt{N}\frac{\hat \Psi_n(\bm\delta) - \Psi(\bm\delta)}{{\sigma}(\bm\delta)} - \mathbb{G}_n(\bm\delta) \biggr| \leq \sup_{\bm\delta \in \Delta^{s_{\max}} }\biggl|  \frac{\sqrt{N}}{\sigma(\bm\delta)}\frac{1}{K}\sum_{k = 1}^K \biggl( R_1^{(-k)}(\bm\delta) + R_2^{(-k)}(\bm\delta) \biggr) \biggr| + o_p(1),
\end{align*}
where
\begin{align*}
    R_1^{(-k)}(\bm\delta) & := \int \biggl[ \tilde{\psi}^*_{S_i}(\cD_i;
        \bm\delta, \{ \hat\pi_s^{(-k)}, \hat m_s^{(-k)}, \hat
        p_s^{(-k)}, \hat\boldf^{[s],(-k)}\}_{s = 1}^{S_i}, \hat\Psi) +
        \hat\Psi(\bm\delta; \hat Q^{(-k)}) - \Psi(\bm\delta, \hat{Q}^{(-k)})\biggr] dF(\cD_i),\\
    R_2^{(-k)}(\bm\delta) &:=  \int \biggl[\zeta(\cD_i; \bm\delta, \{ \hat \pi_s^{(-k)}, \hat m_s^{(-k)}, \widehat{\tilde m}_s^{(-k)}, \hat p_s^{(-k)}, \hat\boldf^{[s],(-k)}\}_{s = 1}^{S_i}) + \Psi(\bm\delta; \hat Q^{(-k)}) - \Psi(\bm\delta; Q) \biggr] dF(\cD_i).
\end{align*}
As $K$ is fixed, it is sufficient
to show that
$\sup_{\bm\delta \in \Delta^{s_{\max}} }|R_1^{(-k)}(\bm\delta)| =
o_p(N^{-\frac{1}{2}})$ and
$\sup_{\bm\delta \in \Delta^{s_{\max}} }| R_2^{(-k)}(\bm\delta)| =
o_p(N^{-\frac{1}{2}})$ for all $k \in \{1, \cdots, K\}$.

We begin with  $\sup_{\bm\delta \in \Delta^{s_{\max}} }|R_1^{(-k)}(\bm\delta)| =
o_p(N^{-\frac{1}{2}})$. Let $m_s^*$ be the outcome model under the estimated
stochastic intervention, which is recursively defined by the terminal
condition
$m_{S_i}^*(\overline{\bW}_{i,S_i-1}, W_{iS_i}, \overline{\bm F}_{iS_i}) :=
m_{S_i}(\overline{\bW}_{i,S_i-1}, W_{iS_i}, \overline{\bm F}_{iS_i})$
(where the intervention-parameter index is omitted because the
terminal outcome model does not depend on $\bm\delta$) and, for $s < S_i$,
\begin{equation}
\begin{aligned}
    &m_s^*(\overline{\bW}_{i,s-1}, W_{is}, \overline{\bm F}_{is}; \bm\delta_{\underline{s}})\\
    &:=
    \E\biggl[
        \frac{
            \delta_{s+1} \hat p_{s+1}^{(-k)}(\overline{\bm W}_{is})
            m_{s+1}^*(\overline{\bW}_{is}, 1, \overline{\bm F}_{i,s+1}; \bm\delta_{\underline{s+1}})
            +
            \{1-\hat p_{s+1}^{(-k)}(\overline{\bm W}_{is})\}
            m_{s+1}^*(\overline{\bW}_{is}, 0, \overline{\bm F}_{i,s+1}; \bm\delta_{\underline{s+1}})
        }{
            \delta_{s+1} \hat p_{s+1}^{(-k)}(\overline{\bm W}_{is})
            + 1 - \hat p_{s+1}^{(-k)}(\overline{\bm W}_{is})
        }
        \,\bigg|\, {\overline{\bm{H}}}_{is}, W_{is}
    \biggr], \label{outcome_true}
\end{aligned}
\end{equation}
so that the estimation errors in
$\hat p_{s+2}^{(-k)}, \ldots, \hat p_{S_i}^{(-k)}$ accumulated at the
later stages are propagated through the recursion.
By Lemma~\ref{lemma2}, we have
\begin{align*}
&\sup_{\bm\delta \in \Delta^{s_{\max}}}\biggl|\int \biggl( \tilde{\psi}^*\bigl(\cD_i; \bm\delta, \{ \hat\pi_s^{(-k)}, \hat m_s^{(-k)}, \hat p_s^{(-k)}, \hat\boldf^{[s],(-k)}\}_{s = 1}^{S_i}, \hat\Psi\bigr) + \hat\Psi(\bm\delta; \hat Q^{(-k)}) - \Psi(\bm\delta; \hat Q^{(-k)})\biggr)dF(\cD_i)
    \biggr| \\
    &\leq \  C \biggl\{ \sum_{s = 1}^{s_{\max} } \sum_{s^\prime = 1}^{s}
    \sup_{\bm\delta \in \Delta^{s_{\max}}}\E\biggl[
        \bigl|
        m_s^*(\overline{\bW}_{i,s-1}, W_{is}, \overline{\bm F}_{is}; \bm\delta_{\underline{s}})
        -
        \hat m_s^{(-k)}(\overline{\bW}_{i,s-1}, W_{is}, \widehat{\overline{\bm F}}^{(-k)}_{is}; \bm\delta_{\underline{s}})
        \bigr|^2
    \biggr]^{\frac{1}{2}} \\
    &\hspace{2.5in} \times
    \E\left[
        \bigl|
        \hat\pi_{s^\prime}^{(-k)}(\widehat{\overline{\bm{H}}}^{(-k)}_{is^\prime})
        -
        \pi_{s^\prime}(\overline{\bm{{H}}}_{is^\prime})
        \bigr|^2
    \right]^{\frac{1}{2}}\biggr\}
\end{align*}
Now, the first term in the product term converges at the rate of $N^{-1/4}$, uniformly over $\bm\delta \in \Delta^{s_{\max}}$, by Lemma~\ref{lemma3}. Also,
\begin{align*}
    \E\left[
        \bigl|
        \hat \pi_s(\widehat{\overline{\bm{H}}}^{(-k)}_{is})
        -
        \pi_s(\overline{\bm H}_{is})
        \bigr|^2
    \right]^{\frac{1}{2}} & \leq
    \E\left[
        \bigl|
        \hat \pi_s(\widehat{\overline{\bm{H}}}^{(-k)}_{is})
        -
        \pi_s(\widehat{\overline{\bm{H}}}^{(-k)}_{is})
        \bigr|^2
    \right]^{\frac{1}{2}}
    + 
    \E\left[
        \bigl|
        \pi_s(\widehat{\overline{\bm{H}}}^{(-k)}_{is})
        -
        \pi_s(\overline{\bm H}_{is})
        \bigr|^2
    \right]^{\frac{1}{2}} \\
    &\leq
    L \sum_{t = 1}^{s} \sum_{r = 1}^{t} \E\left[
        \bigl|
        \hat{\boldf}^{[t],(-k)}(\bR_{ir},r)-\boldf^{[t]}(\bR_{ir},r)
        \bigr|^2
    \right]^{\frac{1}{2}} + o_p(N^{-1/4}) \\
    & =
    o_p(N^{-1/4}),
\end{align*}
where $L$ is a Lipschitz constant. In the above transformation, we use
triangular inequality for the first line, and the second line is by
the convergence of $\hat\pi$ and the Lipschitz continuity of $\pi_s$
(Assumption~\ref{reg_cond}); since the history
$\widehat{\overline{\bm{H}}}^{(-k)}_{is}$ contains the estimated
deconfounders at all levels $t \leq s$ and segments $r \leq t$, the
representation errors accumulate over the entire history, and the
resulting summation contains at most $s_{\max}^2$ terms. Therefore, since $\mathcal{S}$ is finite and each nuisance function converges at the rate of $N^{-1/4}$, for all $k$,
\begin{align*}
    &\sup_{\bm\delta \in \Delta^{s_{\max}} }|R_1^{(-k)}(\bm\delta)| = o_p(N^{-1/2}).
\end{align*}

Lastly, we show that 
$\sup_{\bm\delta \in \Delta^{s_{\max}} }|R_2^{(-k)}(\bm\delta)| =
o_p(N^{-1/2})$. Let
\begin{align}
\check{m}_s^{(-k)}&(\overline{\bm{W}}_{i,s-1},  W_{is}; \bm\delta_{-s}) := \E\biggl[ 
    \biggl(
\prod_{s^\prime=1}^{s-1}
\omega_{s'}(\overline{\bm{H}}_{is'}, W_{is'}; \delta_{s'}, \hat p_{s'}^{(-k)},  \pi_{s'})\biggr) \nonumber \\
    &\hspace{1.2in} \hat m_s^{(-k)} (\overline{\bW}_{i,s-1},  W_{is}, \widehat{\overline{\bm F}}^{(-k)}_{is}; \bm\delta_{\underline{s}}) \mid \overline{\bm{W}}_{i,s-1}, S_i \geq s\biggr]. \label{check_m}
\end{align}
By Lemma~\ref{lemma4} and triangle inequality, 
\begin{align*}
\sup_{\bm\delta \in \Delta^{s_{\max}} }|R_2^{(-k)}(\bm\delta)| &= \sup_{\bm\delta \in \Delta^{s_{\max}} } \biggl| R_{2,1}^{(-k)}(\bm\delta) + R_{2,2}^{(-k)}(\bm\delta) + R_{2,3}^{(-k)}(\bm\delta)\biggr|\\
&\leq \sup_{\bm\delta \in \Delta^{s_{\max}} } \bigl| R_{2,1}^{(-k)}(\bm\delta)\bigr| + 
\sup_{\bm\delta \in \Delta^{s_{\max}} } \bigl| R_{2,2}^{(-k)}(\bm\delta)\bigr|
+
\sup_{\bm\delta \in \Delta^{s_{\max}} } \bigl| R_{2,3}^{(-k)}(\bm\delta)\bigr|,
\end{align*}
where
\begin{align}
  & R_{2,1}^{(-k)}(\bm\delta) \nonumber \\
  := \ &  \sum_{s = 1}^{s_{\max}} \int
                            \frac{\delta_s }{ [\delta_s  \hat
                            p_s^{(-k)}(\overline{\bm W}_{i,s-1}) + \{1
                            - \hat p_s^{(-k)}(\overline{\bm
                            W}_{i,s-1})\}]^2 } \biggl(
                            p_s(\overline{\bm W}_{i,s-1}) - \hat
                            p_s^{(-k)}(\overline{\bm
                            W}_{i,s-1})\biggr) \nonumber \\
&\Biggl( \sum_{w \in \{0,1\}} (2w-1) \biggl\{
\widehat{\tilde m}_s^{(-k)}(\overline{\bm{W}}_{i,s-1},  w; \bm\delta_{-s})  - 
\check{m}^{(-k)}_s(\overline{\bm{W}}_{i,s-1},  w; \bm\delta_{-s})
       \biggr\}  \Biggr) dF(\overline{\bm{W}}_{i,s-1} \mid S_i \geq  s)\Pr(S_i \geq s), \label{reminder_r21}\\
  & R_{2,2}^{(-k)}(\bm\delta) \nonumber\\
  := \ & \sum_{s = 1}^{s_{\max} } \int \frac{\begin{array}{@{}l@{}} \delta_s \{ \hat p_s^{(-k)}(\overline{\bm W}_{i,s-1}) - p_s(\overline{\bm W}_{i,s-1})\} \biggl( \sum_{w} \{2w-1\} \{ m_{s}(\overline{\bW}_{i,s-1}, w, \overline{\bm F}_{is}; \bm\delta_{\underline{s}}) \\[2pt]
                            \hspace{1.7in} -  \hat m_{s}^{(-k)}(\overline{\bW}_{i,s-1}, w, \widehat{\overline{\bm F}}^{(-k)}_{is}; \bm\delta_{\underline{s}})\}  \biggr) \end{array} }{ [\delta_s  \hat
                            p_s^{(-k)}(\overline{\bm W}_{i,s-1}) + \{1
                            - \hat p_s^{(-k)}(\overline{\bm
                            W}_{i,s-1})\}]^2 }\nonumber \\
      &\hspace{0.2in} \times  dF(\bR_{is} \mid \overline{\bm H}_{i,s-1}, W_{i,s-1}, S_i)
\prod_{s^\prime=1}^{s-1}
d\hat{Q}^{(-k)}_{s^\prime}(W_{is^\prime}; \overline{\bm W}_{i,s^\prime-1}, \delta_{s^\prime})
dF(\bR_{is^\prime} \mid \overline{\bm H}_{i,s^\prime-1}, W_{i,s^\prime-1}, S_i)\, \nonumber \\
&\hspace{4in} 
dP(S_i \mid S_i \geq  s)\Pr(S_i \geq s) \label{reminder_r22}\\
  & R_{2,3}^{(-k)}(\bm\delta) \nonumber \\
  := \ &  \sum_{s = 1}^{s_{\max} } \int \frac{\delta_s (\delta_s - 1)  \{ \hat p_s^{(-k)}(\overline{\bm W}_{i,s-1}) - p_s(\overline{\bm W}_{i,s-1})\}^2 \{ m_{s}(\overline{\bW}_{i,s-1}, 1, \overline{\bm F}_{is}; \bm\delta_{\underline{s}}) - m_{s}(\overline{\bW}_{i,s-1}, 0, \overline{\bm F}_{is}; \bm\delta_{\underline{s}})\} }{
        [\delta_s  \hat p_s^{(-k)}(\overline{\bm W}_{i,s-1}) + \{1 - \hat p_s^{(-k)}(\overline{\bm W}_{i,s-1})\}]^2
        [\delta_s   p_s(\overline{\bm W}_{i,s-1}) + \{1 -  p_s(\overline{\bm W}_{i,s-1})\}]
      } \nonumber\\
      &\hspace{0.2in} \times  dF(\bR_{is} \mid \overline{\bm H}_{i,s-1}, W_{i,s-1}, S_i) \prod_{s^\prime=1}^{s-1}
d\hat{Q}^{(-k)}_{s^\prime}(W_{is^\prime}; \overline{\bm W}_{i,s^\prime-1}, \delta_{s^\prime})
dF(\bR_{is^\prime} \mid \overline{\bm H}_{i,s^\prime-1}, W_{i,s^\prime-1}, S_i)\, \nonumber\\
&\hspace{4.0in} dP(S_i \mid S_i \geq s) \Pr(S_i \geq s) .
\label{reminder_r23}
\end{align}

Thus, we only need to bound each term.  For
$R_{2,1}^{(-k)}(\bm\delta)$, notice that
\begin{align*}
    \frac{\delta_s }{ [\delta_s  \hat p_s^{(-k)}(\overline{\bm W}_{i,s-1}) + \{1 - \hat p_s^{(-k)}(\overline{\bm W}_{i,s-1})\}]^2 } \leq \frac{\delta_u}{\min\{1, \delta_\ell\}^2}
\end{align*}
and thus applying the Cauchy-Schwartz inequality yields
\begin{align*}
    &\sup_{\bm\delta \in \Delta^{s_{\max}} } \bigl| R_{2,1}^{(-k)}(\bm\delta)\bigr|\\
    &\lesssim \sum_{s = 1}^{s_{\max}}\sum_{w \in \{0,1\} } \E\biggl[ \bigl|p_s(\overline{\bm W}_{i,s-1}) - \hat p_s^{(-k)}(\overline{\bm W}_{i,s-1})\bigr|^2 \biggr]^\frac{1}{2} \\
    &\hspace{1.0in} \sup_{\bm\delta \in \Delta^{s_{\max}} }\times \E\biggl[ \bigl|\widehat{\tilde m}_s^{(-k)}(\overline{\bm{W}}_{i,s-1},  w; \bm\delta_{-s})  - 
\check{m}^{(-k)}_s(\overline{\bm{W}}_{i,s-1},  w; \bm\delta_{-s})\bigr|^2\biggr]^\frac{1}{2} = o_p(N^{-\frac{1}{2}})
\end{align*}
where the last equality is by Assumption~\ref{reg_cond}. With the almost identical derivation, we can also show that
\begin{align*}
    &\sup_{\bm\delta \in \Delta^{s_{\max}} } \bigl| R_{2,2}^{(-k)}(\bm\delta)\bigr| \lesssim \sum_{s = 1}^{s_{\max}}\sum_{w \in \{0,1\} } \E\biggl[ \bigl|p_s(\overline{\bm W}_{i,s-1}) - \hat p_s^{(-k)}(\overline{\bm W}_{i,s-1})\bigr|^2 \biggr]^\frac{1}{2} \\
    &\hspace{1.0in} \times \sup_{\bm\delta \in \Delta^{s_{\max}} }\E\biggl[ \bigl|
    m_{s}(\overline{\bW}_{i,s-1}, w, \overline{\bm F}_{is}; \bm\delta_{\underline{s}}) -  \hat m_{s}^{(-k)}(\overline{\bW}_{i,s-1}, w, \widehat{\overline{\bm F}}^{(-k)}_{is}; \bm\delta_{\underline{s}})
      \bigr|^2\biggr]^\frac{1}{2} = o_p(N^{-\frac{1}{2}}).
\end{align*}
where the last equality is by Assumption~\ref{reg_cond}.
For the last term, notice that
\begin{align*}
    &\biggl|\frac{\delta_s (\delta_s - 1) \{ m_{s}(\overline{\bW}_{i,s-1}, 1, \overline{\bm F}_{is}; \bm\delta_{\underline{s}}) - m_{s}(\overline{\bW}_{i,s-1}, 0, \overline{\bm F}_{is}; \bm\delta_{\underline{s}})\} }{
        [\delta_s  \hat p_s^{(-k)}(\overline{\bm W}_{i,s-1}) + \{1 - \hat p_s^{(-k)}(\overline{\bm W}_{i,s-1})\}]^2
        [\delta_s   p_s(\overline{\bm W}_{i,s-1}) + \{1 -  p_s(\overline{\bm W}_{i,s-1})\}]
      } \biggr|\\
      &\leq \frac{ \delta_u \max\{|\delta_\ell - 1|, |\delta_u - 1|\}  }{ (\min\{1, \delta_\ell\})^3 } |\{ m_{s}(\overline{\bW}_{i,s-1}, 1, \overline{\bm F}_{is}; \bm\delta_{\underline{s}}) - m_{s}(\overline{\bW}_{i,s-1}, 0, \overline{\bm F}_{is}; \bm\delta_{\underline{s}})\}|\\
      &\leq \frac{2C \delta_u \max\{|\delta_\ell - 1|, |\delta_u - 1|\} }{ (\min\{1, \delta_\ell\})^3 }
\end{align*}
where the first inequality is by $\Delta \in [\delta_\ell, \delta_u]$ and the second inequality is boundedness of $m_s$. Therefore,
\begin{align*}
    \sup_{\bm\delta \in \Delta^{s_{\max}} } \bigl| R_{2,3}^{(-k)}(\bm\delta)\bigr| \lesssim \sum_{s = 1}^{s_{\max}} \E\biggl[ \bigl|p_s(\overline{\bm W}_{i,s-1}) - \hat p_s^{(-k)}(\overline{\bm W}_{i,s-1})\bigr|^2 \biggr] = o_p(N^{-\frac{1}{2}}).
\end{align*}
Hence,
\begin{align*}
\sup_{\bm\delta \in \Delta^{s_{\max}} }|R_2^{(-k)}(\bm\delta)| &\leq \sup_{\bm\delta \in \Delta^{s_{\max}} } \bigl| R_{2,1}^{(-k)}(\bm\delta)\bigr| + 
\sup_{\bm\delta \in \Delta^{s_{\max}} } \bigl| R_{2,2}^{(-k)}(\bm\delta)\bigr|
+
\sup_{\bm\delta \in \Delta^{s_{\max}} } \bigl| R_{2,3}^{(-k)}(\bm\delta)\bigr| = o_p(N^{-\frac{1}{2}}),
\end{align*}
which concludes the proof.
\qed

\subsection{Lemmas for Theorem~\ref{asymp_normal}}
\subsubsection{Lemma 1}
\begin{lemma}\label{lemma1}
\begin{equation*}
\begin{aligned}
    &\sup_{\bm\delta \in \Delta^{s_{\max}} }\biggl| \sqrt{N}\frac{\hat \Psi_n(\bm\delta) - \Psi(\bm\delta)}{{\sigma}(\bm\delta)} - \mathbb{G}_n(\bm\delta) \biggr| \leq \sup_{\bm\delta \in \Delta^{s_{\max}} }\biggl|  \frac{\sqrt{N}}{\sigma(\bm\delta)}\frac{1}{K}\sum_{k = 1}^K \biggl( R_1^{(-k)}(\bm\delta) + R_2^{(-k)}(\bm\delta) \biggr) \biggr| + o_p(1),
\end{aligned}
\end{equation*}
where $\tilde{\psi}^*$ and $\zeta$ are defined in
Equation~\ref{kennedy_lemma3}~and~\eqref{augmentation_term}, respectively, and
\begin{align*}
    R_1^{(-k)}(\bm\delta) & := \int \biggl[ \tilde{\psi}^*(\cD_i;
        \bm\delta, \{ \hat\pi_s^{(-k)}, \hat m_s^{(-k)}, \hat
        p_s^{(-k)}, \hat\boldf^{[s],(-k)}\}_{s = 1}^{S_i}, \hat\Psi) +
        \hat\Psi(\bm\delta; \hat Q^{(-k)}) - \Psi(\bm\delta, \hat{Q}^{(-k)})\biggr] dF(\cD_i),\\
    R_2^{(-k)}(\bm\delta) &:=  \int \biggl[\zeta(\cD_i; \bm\delta, \{ \hat \pi_s^{(-k)}, \hat m_s^{(-k)}, \widehat{\tilde m}_s^{(-k)}, \hat p_s^{(-k)}, \hat\boldf^{[s],(-k)}\}_{s = 1}^{S_i}) + \Psi(\bm\delta; \hat Q^{(-k)}) - \Psi(\bm\delta; Q) \biggr] dF(\cD_i).
\end{align*}
\end{lemma}

\begin{proof}
  Let
  $\varphi(\cD_i; \bm\delta, \{ \pi_s, m_{s}, \tilde m_{s}, p_s, \boldf^{[s]} \}_{s = 1}^{S_i})$ be an uncentered influence function, which is defined as
\begin{align}
    \varphi(\cD_i; \bm\delta, \{ \pi_s, m_{s}, \tilde m_{s}, p_s, \boldf^{[s]}
            \}_{s = 1}^{S_i}) = \psi(\cD_i; \bm\delta, \{ \pi_s, m_{s}, \tilde m_{s}, p_s, \boldf^{[s]}
      \}_{s = 1}^{S_i}, \Psi) + \Psi(\bm\delta). \label{uncentered_if}
\end{align}
Since we solve Equation~\eqref{est_eq} to obtain the estimator $\hat\Psi(\bm\delta)$, we have
\begin{align*}
    \hat\Psi(\bm\delta) = \frac{1}{nK}\sum_{k = 1}^K \sum_{I(i) = k} \varphi(\cD_i; \bm\delta, \{ \hat\pi_s^{(-k)}, \hat m_{s}^{(-k)}, \widehat{\tilde m}_{s}^{(-k)}, \hat p_s^{(-k)}, \hat\boldf^{[s],(-k)}
      \}_{s = 1}^{S_i}).
\end{align*}
Thus, using a strategy similar to that of
\cite{kennedy_nonparametric_2019}, we have
\begin{align*}
&\frac{\sqrt{N}\{\hat \Psi_n(\bm\delta) -
                 \Psi(\bm\delta)\}}{{\sigma}(\bm\delta)} -
                 \mathbb{G}_n(\bm\delta) \\
= \ & \frac{\sqrt{N}\{\hat \Psi_n(\bm\delta) -
                 \Psi(\bm\delta)\}}{{\sigma}(\bm\delta)} -
                 \frac{1}{\sqrt{N}} \sum_{i = 1}^N \frac{ \psi(\cD_i; \bm\delta, \{ \pi_s, m_{s}, \tilde m_{s}, p_s, \boldf^{[s]}
      \}_{s= 1}^{S_i}, \Psi) }{\sigma(\bm\delta) } \\
= \ & \frac{\sqrt{N}}{\sigma(\bm\delta)}\left[ \{\hat \Psi_n(\bm\delta) -
                 \Psi(\bm\delta)\} - \frac{1}{N} \sum_{i = 1}^N \psi(\cD_i; \bm\delta, \{ \pi_s, m_{s}, \tilde m_{s}, p_s, \boldf^{[s]}
      \}_{s= 1}^{S_i}, \Psi) \right] \\
= \ & \frac{\sqrt{N}}{\sigma(\bm\delta)}\left[ \{\hat \Psi_n(\bm\delta) -
                 \Psi(\bm\delta)\} - \frac{1}{N} \sum_{i = 1}^N \{\varphi(\cD_i; \bm\delta, \{ \pi_s, m_{s}, p_s, \boldf^{[s]}
            \}_{s = 1}^{S_i}) - \Psi(\bm\delta) \} \right] \\
= \ &  \frac{\sqrt{N}}{\sigma(\bm\delta)}\frac{1}{K}\sum_{k = 1}^K \left[ \frac{1}{n} \sum_{I(i) = k} \left\{\varphi(\cD_i; \bm\delta, \{ \hat\pi_s^{(-k)}, \hat m_{s}^{(-k)}, \hat p_s^{(-k)}, \hat\boldf^{[s],(-k)}
      \}_{s = 1}^{S_i}) - \varphi(\cD_i; \bm\delta, \{ \pi_s, m_{s}, \tilde m_{s}, p_s, \boldf^{[s]}
            \}_{s = 1}^{S_i}) \right\} \right] \\
= \ &  \frac{\sqrt{N}}{\sigma(\bm\delta)}\frac{1}{K}\sum_{k = 1}^K \left[ \frac{1}{n} \sum_{I(i) = k} \left\{\varphi(\cD_i; \bm\delta, \{ \hat\pi_s^{(-k)}, \hat m_{s}^{(-k)}, \hat p_s^{(-k)}, \hat\boldf^{[s],(-k)}
      \}_{s = 1}^{S_i}) - \varphi(\cD_i; \bm\delta, \{ \pi_s, m_{s}, \tilde m_{s}, p_s, \boldf^{[s]}
            \}_{s = 1}^{S_i}) \right\} \right. \\
      &\hspace{1in} - \E\left[\varphi(\cD_i; \bm\delta, \{ \hat\pi_s^{(-k)}, \hat m_{s}^{(-k)}, \widehat{\tilde m}_{s}^{(-k)}, \hat p_s^{(-k)}, \hat\boldf^{[s],(-k)}
      \}_{s = 1}^{S_i}) - \varphi(\cD_i; \bm\delta, \{ \pi_s, m_{s}, \tilde m_{s}, p_s, \boldf^{[s]}
            \}_{s = 1}^{S_i})\right] \Biggr]\\
      & + \frac{\sqrt{N}}{\sigma(\bm\delta)}\frac{1}{K}\sum_{k = 1}^K \E\left[\varphi(\cD_i; \bm\delta, \{ \hat\pi_s^{(-k)}, \hat m_{s}^{(-k)}, \widehat{\tilde m}_{s}^{(-k)}, \hat p_s^{(-k)}, \hat\boldf^{[s],(-k)}
      \}_{s = 1}^{S_i}) - \varphi(\cD_i; \bm\delta, \{ \pi_s, m_{s}, \tilde m_{s}, p_s, \boldf^{[s]}
            \}_{s = 1}^{S_i})\right],
\end{align*}
where the first equality is by the definition of
$\mathbb{G}_n(\bm\delta)$, the third equality is due to the definition
of uncentered influence function given in
Equation~\eqref{uncentered_if}, the fifth equality is by adding and
subtracting the same term.

We use the same argument as the one used in the proof of Theorem 3 of
\cite{kennedy_nonparametric_2019}.  Since we have already shown that
the bracketing integral is finite, we have,
\begin{align*}
      &\sup_{\bm\delta \in
  \Delta^{s_{\max}} } \left| \frac{\sqrt{N}}{\sigma(\bm\delta)}\frac{1}{K}\sum_{k = 1}^K \left[ \frac{1}{n} \sum_{I(i) = k} \left\{\varphi(\cD_i; \bm\delta, \{ \hat\pi_s^{(-k)}, \hat m_{s}^{(-k)}, \hat p_s^{(-k)}, \hat\boldf^{[s],(-k)}
      \}_{s = 1}^{S_i}) - \varphi(\cD_i; \bm\delta, \{ \pi_s, m_{s}, \tilde m_{s}, p_s, \boldf^{[s]}
            \}_{s = 1}^{S_i}) \right\} \right.\right. \\
      &\qquad - \E[\varphi(\cD_i; \bm\delta, \{ \hat\pi_s^{(-k)}, \hat m_{s}^{(-k)}, \widehat{\tilde m}_{s}^{(-k)}, \hat p_s^{(-k)}, \hat\boldf^{[s],(-k)}
      \}_{s = 1}^{S_i}) - \varphi(\cD_i; \bm\delta, \{ \pi_s, m_{s}, \tilde m_{s}, p_s, \boldf^{[s]}
        \}_{s = 1}^{S_i})] \biggr] \biggr| \\
  = & \ o_p(1)
\end{align*}
Therefore, by triangular inequality,
\begin{align*}
&\sup_{\bm\delta \in \Delta^{s_{\max}} }\biggl| \sqrt{N}\frac{\hat \Psi_n(\bm\delta) - \Psi(\bm\delta)}{{\sigma}(\bm\delta)} - \mathbb{G}_n(\bm\delta) \biggr|\\
    &\leq \sup_{\bm\delta \in \Delta^{s_{\max}} }\biggl|  \frac{\sqrt{N}}{\sigma(\bm\delta)}\frac{1}{K}\sum_{k = 1}^K \E[\varphi(\cD_i; \bm\delta, \{ \hat\pi_s^{(-k)}, \hat m_{s}^{(-k)}, \widehat{\tilde m}_{s}^{(-k)}, \hat p_s^{(-k)}, \hat\boldf^{[s],(-k)}
      \}_{s = 1}^{S_i}) \\
    &\hspace{2in} - \varphi(\cD_i; \bm\delta, \{ \pi_s, m_{s}, \tilde m_{s}, p_s, \boldf^{[s]}
            \}_{s = 1}^{S_i})] \biggr| + o_p(1).
\end{align*}
Finally, we have,
\begin{align*}
       & \frac{\sqrt{N}}{\sigma(\bm\delta)}\frac{1}{K}\sum_{k = 1}^K \E[\varphi(\cD_i; \bm\delta, \{ \hat\pi_s^{(-k)}, \hat m_{s}^{(-k)}, \widehat{\tilde m}_{s}^{(-k)}, \hat p_s^{(-k)}, \hat\boldf^{[s],(-k)}
      \}_{s = 1}^{S_i}) - \varphi(\cD_i; \bm\delta, \{ \pi_s, m_{s}, \tilde m_{s}, p_s, \boldf^{[s]}
            \}_{s = 1}^{S_i})] \\
    = & \frac{\sqrt{N}}{\sigma(\bm\delta)}\frac{1}{K}\sum_{k = 1}^K\int \left[
    \varphi(\cD_i; \bm\delta, \{ \hat\pi_s^{(-k)}, \hat m_{s}^{(-k)}, \widehat{\tilde m}_{s}^{(-k)}, \hat p_s^{(-k)}, \hat\boldf^{[s],(-k)}
      \}_{s = 1}^{S_i})
      - \varphi(\cD_i; \bm\delta, \{ \pi_s, m_{s}, \tilde m_{s}, p_s, \boldf^{[s]}
      \}_{s = 1}^{S_i})
    \right] dF(\cD_i)\\
    = & \frac{\sqrt{N}}{\sigma(\bm\delta)}\frac{1}{K}\sum_{k = 1}^K\int \left[
    \tilde{\psi}^*(\cD_i; \bm\delta, \{ \hat\pi_s^{(-k)},
      \hat m_s^{(-k)}, \hat p_s^{(-k)}, \hat\boldf^{[s],(-k)}\}_{s = 1}^{S_i}, \hat\Psi) + \hat\Psi(\bm\delta, \hat Q^{(-k)}) \right.\\
    &\left.  \hspace{1in}  - \psi(\cD_i; \bm\delta, \{ \pi_s, m_{s}, \tilde m_{s}, p_s, \boldf^{[s]}
      \}_{s= 1}^{S_i}, \Psi)- \Psi(\bm\delta, Q)
    \right] dF(\cD_i)\\
    = & \frac{\sqrt{N}}{\sigma(\bm\delta)}\frac{1}{K}\sum_{k = 1}^K\int \left[
    \tilde{\psi}^*(\cD_i; \bm\delta, \{ \hat\pi_s^{(-k)},
      \hat m_s^{(-k)}, \hat p_s^{(-k)}, \hat\boldf^{[s],(-k)}\}_{s = 1}^{S_i}, \hat\Psi) + \hat\Psi(\bm\delta, \hat Q^{(-k)})  - \Psi(\bm\delta, Q)
    \right] dF(\cD_i)\\
    = & \frac{\sqrt{N}}{\sigma(\bm\delta)}\frac{1}{K}\sum_{k = 1}^K\int
      \biggl[ \tilde{\psi}^*(\cD_i; \bm\delta, \{ \hat\pi_s^{(-k)},
      \hat m_s^{(-k)}, \hat p_s^{(-k)}, \hat\boldf^{[s],(-k)}\}_{s = 1}^{S_i}, \hat\Psi) \\
   & \hspace{1in} + \zeta(\cD_i;  \bm\delta,  \{ \hat \pi_s^{(-k)}, \hat m_s^{(-k)}, \widehat{\tilde m}_s^{(-k)}, \hat p_s^{(-k)}, \hat\boldf^{[s],(-k)}\}_{s = 1}^{S_i} )+ \hat\Psi(\bm\delta; \hat Q^{(-k)}) - \Psi(\bm\delta; Q) \biggr] dF(\cD_i) \\
    = &  \frac{\sqrt{N}}{\sigma(\bm\delta)}\frac{1}{K}\sum_{k = 1}^K
        \biggl\{\int \biggl[ \tilde{\psi}^*(\cD_i;
        \bm\delta, \{ \hat\pi_s^{(-k)}, \hat m_s^{(-k)}, \hat
        p_s^{(-k)}, \hat\boldf^{[s],(-k)}\}_{s = 1}^{S_i}, \hat\Psi) +
        \hat\Psi(\bm\delta; \hat Q^{(-k)}) - \Psi(\bm\delta, \hat{Q}^{(-k)})\biggr] dF(\cD_i)\\
    &\hspace{1in} + \int \biggl[\zeta(\cD_i; \bm\delta, \{ \hat \pi_s^{(-k)}, \hat m_s^{(-k)}, \widehat{\tilde m}_s^{(-k)}, \hat p_s^{(-k)}, \hat\boldf^{[s],(-k)}\}_{s = 1}^{S_i}) + \Psi(\bm\delta; \hat Q^{(-k)}) - \Psi(\bm\delta; Q) \biggr] dF(\cD_i) \biggr\}
\end{align*}
where the second equality is by the definition of an uncentered
influence function
$\varphi(\cD_i; \bm\delta, \{ \pi_s, m_{s}, \tilde m_{s}, p_s, \boldf^{[s]} \}_{s = 1}^{S_i})$, the third equality holds because the expectation of the
centered influence function is zero, the fourth equality is due the
definition of the augmentation term, and the last equality results by
adding and subtracting $\Psi(\bm\delta; \hat Q^{(-k)})$.
\end{proof}

\subsubsection{Lemma 2}
\begin{lemma}\label{lemma2} There exists a constant $C < \infty$,
which depends on $\bm\delta$ only through $\delta_\ell$ and
$\delta_u$, that satisfies
\begin{align*}
&\sup_{\bm\delta \in \Delta^{s_{\max}}}\left|\int \left[ \tilde{\psi}^*\bigl(\cD_i; \bm\delta, \{ \hat\pi_s, \hat m_s, \hat p_s, \hat\boldf^{[s]}\}_{s = 1}^{S_i}, \hat\Psi\bigr) + \hat\Psi(\bm\delta; \hat Q) - \Psi(\bm\delta; \hat Q)\right]dF(\cD_i)
    \right| \\
    &\leq \  C \sum_{s = 1}^{s_{\max}} \sum_{s^\prime = 1}^{s}
    \sup_{\bm\delta \in \Delta^{s_{\max}}}\E\biggl[
        \bigl|
        m_s^*(\overline{\bW}_{i,s-1}, W_{is}, \overline{\bm F}_{is}; \bm\delta_{\underline{s}})
        -
        \hat m_s(\overline{\bW}_{i,s-1}, W_{is}, \widehat{\overline{\bm F}}_{is}; \bm\delta_{\underline{s}})
        \bigr|^2
    \biggr]^{\frac{1}{2}} \\
    &\hspace{2.5in} \times
    \E\left[
        \bigl|
        \hat\pi_{s^\prime}(\widehat{\overline{\bm{H}}}_{is^\prime})
        -
        \pi_{s^\prime}(\overline{\bm{{H}}}_{is^\prime})
        \bigr|^2
    \right]^{\frac{1}{2}},
\end{align*}
where $m_s^*$ is the outcome model under the estimated
stochastic intervention.
\end{lemma}

\begin{proof}
By Lemma~5 of
\cite{kennedy_nonparametric_2019}, for each
fixed $\tilde s \in \mathcal{S}$, we have,
\begin{align*}
    &\int \left[ \tilde{\psi}^*_{\tilde s}\bigl(\cD_i; \bm\delta, \{ \hat\pi_s, \hat m_s, \hat p_s, \hat\boldf^{[s]}\}_{s = 1}^{\tilde s}, \hat\Psi_{\tilde s}\bigr)\right]dF(\cD_i \mid S_i = \tilde s) + \hat\Psi_{\tilde s}(\bm\delta; \hat Q) - \Psi_{\tilde s}(\bm\delta; \hat Q) \\
    = &  \sum_{s = 1}^{\tilde s} \sum_{s^\prime = 1}^{s} \int  
    \left[
        m_s^*(\overline{\bW}_{i,s-1}, W_{is}, \overline{\bm F}_{is}; \bm\delta_{\underline{s}})
        -
        \hat m_s(\overline{\bW}_{i,s-1}, W_{is}, \widehat{\overline{\bm F}}_{is}; \bm\delta_{\underline{s}})
        \right]
        \frac{ \{\hat  \pi_{s^\prime}(\widehat{\overline{\bm H}}_{is^\prime}) -  \pi_{s^\prime}(\overline{\bm H}_{is^\prime})\} (2W_{is^\prime} - 1) }{ W_{is^\prime}\hat  \pi_{s^\prime}(\widehat{\overline{\bm H}}_{is^\prime}) + (1 - W_{is^\prime})(1 - \hat  \pi_{s^\prime}(\widehat{\overline{\bm H}}_{is^\prime})) } \\
    & \hspace{.5in} \times
    \left[
        \prod_{r=1}^{s^\prime-1}
         \frac{W_{ir} \pi_r({\overline{\bm{H}}}_{ir}) + (1 - W_{ir}) \{ 1 - \pi_r({\overline{\bm{H}}}_{ir})\}
        }{W_{ir}\hat \pi_r(\widehat{\overline{\bm{H}}}_{ir}) + (1 - W_{ir}) \{ 1 - \hat \pi_r(\widehat{\overline{\bm{H}}}_{ir})\}}
        \right]\\
    & \hspace{.25in} \times 
    \left[
        \prod_{r=1}^{s}
        \frac{\delta_r W_{ir} \hat p_r(\overline{\bm W}_{i,r-1}) + (1 - W_{ir}) \{ 1 -  \hat p_r(\overline{\bm W}_{i,r-1})\} }{ \delta_r \hat p_r(\overline{\bm W}_{i,r-1}) + 1 -  \hat p_r(\overline{\bm W}_{i,r-1})} \,
        dF\bigl(\overline{\bm R}_{ir} \mid \overline{\bm R}_{i,r-1},S_i=\tilde s\bigr)
    \right]\\
     = & \sum_{s = 1}^{\tilde s} \sum_{s^\prime = 1}^{s} \int  
    \left[
        m_s^*(\overline{\bW}_{i,s-1}, W_{is}, \overline{\bm F}_{is}; \bm\delta_{\underline{s}})
        -
        \hat m_s(\overline{\bW}_{i,s-1}, W_{is}, \widehat{\overline{\bm F}}_{is}; \bm\delta_{\underline{s}})
        \right]
        \times
        \left[ \hat  \pi_{s^\prime}(\widehat{\overline{\bm H}}_{is^\prime}) -  \pi_{s^\prime}(\overline{\bm H}_{is^\prime}) \right]\\
    & \hspace{.5in} \times 
        \left[ \frac{ (2W_{is^\prime} - 1) }{ W_{is^\prime}\hat  \pi_{s^\prime}(\widehat{\overline{\bm H}}_{is^\prime}) + (1 - W_{is^\prime})(1 - \hat  \pi_{s^\prime}(\widehat{\overline{\bm H}}_{is^\prime})) } 
    \times 
    \frac{\delta_{s^\prime} W_{is^\prime} \hat p_{s^\prime}(\overline{\bm{W}}_{i,s^\prime-1}) + (1 - W_{is^\prime}) \{ 1 - \hat p_{s^\prime}(\overline{\bm{W}}_{i,s^\prime-1})\}
        }{\delta_{s^\prime} \hat p_{s^\prime}(\overline{\bm{W}}_{i,s^\prime-1}) +  1 - \hat p_{s^\prime}(\overline{\bm{W}}_{i,s^\prime-1}) }\right]\\
    & \hspace{.5in} \times  \prod_{r=1}^{s^\prime-1} \left[
        \frac{W_{ir} \pi_r({\overline{\bm{H}}}_{ir}) + (1 - W_{ir}) \{ 1 - \pi_r({\overline{\bm{H}}}_{ir})\}
        }{W_{ir}\hat \pi_r(\widehat{\overline{\bm{H}}}_{ir}) + (1 - W_{ir}) \{ 1 - \hat \pi_r(\widehat{\overline{\bm{H}}}_{ir})\}}\right.\\
    &\left. \hspace{2in} \times \frac{\delta_r W_{ir} \hat p_r(\overline{\bm W}_{i,r-1}) + (1 - W_{ir}) \{ 1 -  \hat p_r(\overline{\bm W}_{i,r-1})\} }{ \delta_r \hat p_r(\overline{\bm W}_{i,r-1}) + 1 -  \hat p_r(\overline{\bm W}_{i,r-1})}  \right]\\
    & \hspace{.25in} \times 
    \left[
        \prod_{r=s^\prime + 1}^{s}
        \frac{\delta_r W_{ir} \hat p_r(\overline{\bm W}_{i,r-1}) + (1 - W_{ir}) \{ 1 -  \hat p_r(\overline{\bm W}_{i,r-1})\} }{ \delta_r \hat p_r(\overline{\bm W}_{i,r-1}) + 1 -  \hat p_r(\overline{\bm W}_{i,r-1})} \right] \left[  \prod_{r=1}^{s}
        dF\bigl(\overline{\bm R}_{ir} \mid \overline{\bm R}_{i,r-1},    S_i=\tilde s\bigr)
    \right].
\end{align*}
Then, notice that
\begin{align*}
    &\frac{W_{ir} \pi_r({\overline{\bm{H}}}_{ir}) + (1 - W_{ir}) \{ 1 - \pi_r({\overline{\bm{H}}}_{ir})\}
        }{W_{ir}\hat \pi_r(\widehat{\overline{\bm{H}}}_{ir}) + (1 - W_{ir}) \{ 1 - \hat \pi_r(\widehat{\overline{\bm{H}}}_{ir})\}}\times\frac{\delta_r W_{ir} \hat p_r(\overline{\bm W}_{i,r-1}) + (1 - W_{ir}) \{ 1 -  \hat p_r(\overline{\bm W}_{i,r-1})\} }{ \delta_r \hat p_r(\overline{\bm W}_{i,r-1}) + 1 -  \hat p_r(\overline{\bm W}_{i,r-1})}\\
    = \ &  W_{ir}\frac{\pi_r({\overline{\bm{H}}}_{ir}) \{\delta_r \hat p_r(\overline{\bm W}_{i,r-1})\}   }{ \hat\pi_r(\widehat{\overline{\bm{H}}}_{ir}) \{\delta_r \hat p_r(\overline{\bm W}_{i,r-1}) + 1 -  \hat p_r(\overline{\bm W}_{i,r-1})\} } + (1 - W_{ir}) \frac{ 
    \{ 1 - \pi_r({\overline{\bm{H}}}_{ir})\} \{ 1 -  \hat p_r(\overline{\bm W}_{i,r-1})\}
    }{\{\delta_r \hat p_r(\overline{\bm W}_{i,r-1}) + 1 -  \hat p_r(\overline{\bm W}_{i,r-1})\}\{ 1 - \hat\pi_r(\widehat{\overline{\bm{H}}}_{ir}) \} }\\
    \leq \ & W_{ir}\frac{ \delta_r \hat p_r(\overline{\bm W}_{i,r-1})   }{ \hat\pi_r(\widehat{\overline{\bm{H}}}_{ir}) \{\delta_r \hat p_r(\overline{\bm W}_{i,r-1}) + 1 -  \hat p_r(\overline{\bm W}_{i,r-1})\} } + (1 - W_{ir}) \frac{ 
    1 -  \hat p_r(\overline{\bm W}_{i,r-1})
    }{\{\delta_r \hat p_r(\overline{\bm W}_{i,r-1}) + 1 -  \hat p_r(\overline{\bm W}_{i,r-1})\}\{ 1 - \hat\pi_r(\widehat{\overline{\bm{H}}}_{ir}) \} }\\
    \leq \ & W_{ir}\frac{ \delta_r \hat p_r(\overline{\bm W}_{i,r-1})   }{ \hat\pi_r(\widehat{\overline{\bm{H}}}_{ir}) \min\{1, \delta_\ell\} } + (1 - W_{ir}) \frac{ 
    1 -  \hat p_r(\overline{\bm W}_{i,r-1})
    }{ 
    \min\{1, \delta_\ell\}
    \{ 1 - \hat\pi_r(\widehat{\overline{\bm{H}}}_{ir}) \} }\\
    \leq \ & W_{ir}\frac{\delta_r}{c\min\{1, \delta_\ell\}} + (1 - W_{ir}) \frac{ 1
    }{
    c\min\{1, \delta_\ell\}
     } \\
    \leq \ & \frac{\delta_u}{c\min\{1, \delta_\ell\}} + \frac{ 1
    }{
    c\min\{1, \delta_\ell\}
     },
\end{align*}
where %the first inequality is because $\pi_r$ and $1 - \pi_r$ is
      %upper bounded by 1 by definition, the second inequality is
      %because $\delta_r \in [\delta_u, \delta_\ell]$,
the third inequality is by
Assumption~\ref{reg_cond}. %and the last line is because $W_{ir}$ is binary.
Similarly,
\begin{align*}
    &\frac{ (2W_{is^\prime} - 1) }{ W_{is^\prime}\hat  \pi_{s^\prime}(\widehat{\overline{\bm H}}_{is^\prime}) + (1 - W_{is^\prime})(1 - \hat  \pi_{s^\prime}(\widehat{\overline{\bm H}}_{is^\prime})) } 
    \times 
    \frac{\delta_{s^\prime} W_{is^\prime} \hat p_{s^\prime}(\overline{\bm{W}}_{i,s^\prime-1}) + (1 - W_{is^\prime}) \{ 1 - \hat p_{s^\prime}(\overline{\bm{W}}_{i,s^\prime-1})\}
        }{\delta_{s^\prime} \hat p_{s^\prime}(\overline{\bm{W}}_{i,s^\prime-1}) +  1 - \hat p_{s^\prime}(\overline{\bm{W}}_{i,s^\prime-1})}\\
    \leq \ & \frac{ 1 }{ W_{is^\prime}\hat  \pi_{s^\prime}(\widehat{\overline{\bm H}}_{is^\prime}) + (1 - W_{is^\prime})(1 - \hat  \pi_{s^\prime}(\widehat{\overline{\bm H}}_{is^\prime})) } 
    \times 
    \frac{\delta_{s^\prime} W_{is^\prime} \hat p_{s^\prime}(\overline{\bm{W}}_{i,s^\prime-1}) + (1 - W_{is^\prime}) \{ 1 - \hat p_{s^\prime}(\overline{\bm{W}}_{i,s^\prime-1})\}
        }{\delta_{s^\prime} \hat p_{s^\prime}(\overline{\bm{W}}_{i,s^\prime-1}) +  1 - \hat p_{s^\prime}(\overline{\bm{W}}_{i,s^\prime-1})}\\
    = \ &  W_{is^\prime}\frac{ \delta_{s^\prime} \hat p_{s^\prime}(\overline{\bm{W}}_{i,s^\prime-1}) }{ \hat  \pi_{s^\prime}(\widehat{\overline{\bm H}}_{is^\prime}) \{\delta_{s^\prime} \hat p_{s^\prime}(\overline{\bm{W}}_{i,s^\prime-1}) +  1 - \hat p_{s^\prime}(\overline{\bm{W}}_{i,s^\prime-1})\}  }\\
    &\qquad + (1 - W_{is^\prime}) \frac{1 - \hat p_{s^\prime}(\overline{\bm{W}}_{i,s^\prime-1})}{ 
    \{ 1- \hat  \pi_{s^\prime}(\widehat{\overline{\bm H}}_{is^\prime})\} \{\delta_{s^\prime} \hat p_{s^\prime}(\overline{\bm{W}}_{i,s^\prime-1}) +  1 - \hat p_{s^\prime}(\overline{\bm{W}}_{i,s^\prime-1})\}
    }\\
    \leq \ &  W_{is^\prime} \frac{ \delta_{s^\prime} \hat p_{s^\prime}(\overline{\bm{W}}_{i,s^\prime-1}) }{ \hat  \pi_{s^\prime}(\widehat{\overline{\bm H}}_{is^\prime}) \min\{1, \delta_\ell\}  } + (1 - W_{is^\prime})  \frac{  1 - \hat p_{s^\prime}(\overline{\bm{W}}_{i,s^\prime-1}) }{ \{1 - \hat  \pi_{s^\prime}(\widehat{\overline{\bm H}}_{is^\prime})\} \min\{1, \delta_\ell\}  }\\
   \leq \ & W_{is^\prime} \frac{\delta_{s^\prime}}{c\min\{1, \delta_\ell\} } + \{1 - W_{is^\prime} \} \frac{1}{c\min\{1, \delta_\ell\} }\\
   \leq \ &  \frac{\delta_u}{c\min\{1, \delta_\ell\}} + \frac{ 1
    }{
    c\min\{1, \delta_\ell\}
     }.
\end{align*}
Finally,
\begin{align*}
    \frac{\delta_r W_{ir} \hat p_r(\overline{\bm W}_{i,r-1}) + (1 - W_{ir}) \{ 1 -  \hat p_r(\overline{\bm W}_{i,r-1})\} }{ \delta_r \hat p_r(\overline{\bm W}_{i,r-1}) + 1 -  \hat p_r(\overline{\bm W}_{i,r-1})} &\leq \frac{\delta_r W_{ir} \hat p_r(\overline{\bm W}_{i,r-1}) + (1 - W_{ir}) \{ 1 -  \hat p_r(\overline{\bm W}_{i,r-1})\} }{ \min\{1, \delta_\ell\} }\\
    &\le \frac{\delta_r W_{ir}  + (1 - W_{ir}) }{ \min\{1, \delta_\ell\} }\\
    &\leq \frac{\max\{1, \delta_u\} }{
    \min\{1, \delta_\ell\} 
    }.
\end{align*}
%where the first and second inequalities are by the boundedness of $\hat p_s$ and $\delta_{s'} \in [\delta_\ell, \delta_u]$ in Assumption~\ref{reg_cond} and the last inequality is because $W_{ir}$ is binary.
Thus, as each weight term is upper bounded,
\begin{align*}
    &\int \left[ \tilde{\psi}^*_{\tilde s}\bigl(\cD_i; \bm\delta, \{ \hat\pi_s, \hat m_s, \hat p_s, \hat\boldf^{[s]}\}_{s = 1}^{\tilde s}, \hat\Psi_{\tilde s}\bigr)\right]dF(\cD_i \mid S_i = \tilde s) + \hat\Psi_{\tilde s}(\bm\delta; \hat Q) - \Psi_{\tilde s}(\bm\delta; \hat Q)\\
    \lesssim \ & \sum_{s = 1}^{\tilde s} \sum_{s^\prime = 1}^{s} \E\biggl[  
    \left[
        m_s^*(\overline{\bW}_{i,s-1}, W_{is}, \overline{\bm F}_{is}; \bm\delta_{\underline{s}})
        -
        \hat m_s(\overline{\bW}_{i,s-1}, W_{is}, \widehat{\overline{\bm F}}_{is}; \bm\delta_{\underline{s}})
        \right]
        \times
        \left[ \hat  \pi_{s^\prime}(\widehat{\overline{\bm H}}_{is^\prime}) -  \pi_{s^\prime}(\overline{\bm H}_{is^\prime}) \right]\biggr]\\
   \lesssim \  &
    \sum_{s = 1}^{\tilde s} \sum_{s^\prime = 1}^{s}
    \E\biggl[
        \bigl|
        m_s^*(\overline{\bW}_{i,s-1}, W_{is}, \overline{\bm F}_{is}; \bm\delta_{\underline{s}})
        -
        \hat m_s(\overline{\bW}_{i,s-1}, W_{is}, \widehat{\overline{\bm F}}_{is}; \bm\delta_{\underline{s}})
        \bigr|^2
    \biggr]^{\frac{1}{2}} \\
    &\hspace{2.5in} \times
    \E\left[
        \bigl|
        \hat\pi_{s^\prime}(\widehat{\overline{\bm{H}}}_{is^\prime})
        -
        \pi_{s^\prime}(\overline{\bm{{H}}}_{is^\prime})
        \bigr|^2
    \right]^{\frac{1}{2}}
\end{align*}
where the last line is by Cauchy-Schwartz inequality. Since $S_i \leq s_{\max} < \infty$,
\begin{align*}
    &\biggl|\int \biggl( \tilde{\psi}^*\bigl(\cD_i; \bm\delta, \{ \hat\pi_s, \hat m_s, \hat p_s, \hat\boldf^{[s]}\}_{s = 1}^{S_i}, \hat\Psi\bigr) + \hat\Psi(\bm\delta; \hat Q) - \Psi(\bm\delta; \hat Q)\biggr)dF(\cD_i)
    \biggr|\\
    = \ &  \biggl| \sum_{\tilde s = 1}^{s_{\max}} \biggl\{ \int \left[ \tilde{\psi}^*_{\tilde s}\bigl(\cD_i; \bm\delta, \{ \hat\pi_s, \hat m_s, \hat p_s, \hat\boldf^{[s]}\}_{s = 1}^{\tilde s}, \hat\Psi_{\tilde s}\bigr)\right]dF(\cD_i \mid S_i = \tilde s) + \hat\Psi_{\tilde s}(\bm\delta; \hat Q) - \Psi_{\tilde s}(\bm\delta; \hat Q) \biggr\} \Pr(S_i = \tilde s) 
    \biggr|\\
    \leq \ &  \sum_{\tilde s = 1}^{s_{\max}} \biggl|  \biggl\{ \int \left[ \tilde{\psi}^*_{\tilde s}\bigl(\cD_i; \bm\delta, \{ \hat\pi_s, \hat m_s, \hat p_s, \hat\boldf^{[s]}\}_{s = 1}^{\tilde s}, \hat\Psi_{\tilde s}\bigr)\right]dF(\cD_i \mid S_i = \tilde s) + \hat\Psi_{\tilde s}(\bm\delta; \hat Q) - \Psi_{\tilde s}(\bm\delta; \hat Q) \biggr\} 
    \biggr| \Pr(S_i = \tilde s) \\
    \lesssim \ & \sum_{\tilde s = 1}^{s_{\max}}
    \sum_{s = 1}^{\tilde s} \sum_{s^\prime = 1}^{s}
    \E\biggl[
        \bigl|
        m_s^*(\overline{\bW}_{i,s-1}, W_{is}, \overline{\bm F}_{is}; \bm\delta_{\underline{s}})
        -
        \hat m_s(\overline{\bW}_{i,s-1}, W_{is}, \widehat{\overline{\bm F}}_{is}; \bm\delta_{\underline{s}})
        \bigr|^2
    \biggr]^{\frac{1}{2}} \\
    &\hspace{2.5in} \times
    \E\left[
        \bigl|
        \hat\pi_{s^\prime}(\widehat{\overline{\bm{H}}}_{is^\prime})
        -
        \pi_{s^\prime}(\overline{\bm{{H}}}_{is^\prime})
        \bigr|^2
    \right]^{\frac{1}{2}} \\
    = \ & \sum_{s = 1}^{s_{\max}} \sum_{s^\prime = 1}^{s} (s_{\max} - s + 1)
    \E\biggl[
        \bigl|
        m_s^*(\overline{\bW}_{i,s-1}, W_{is}, \overline{\bm F}_{is}; \bm\delta_{\underline{s}})
        -
        \hat m_s(\overline{\bW}_{i,s-1}, W_{is}, \widehat{\overline{\bm F}}_{is}; \bm\delta_{\underline{s}})
        \bigr|^2
    \biggr]^{\frac{1}{2}} \\
    &\hspace{2.5in} \times
    \E\left[
        \bigl|
        \hat\pi_{s^\prime}(\widehat{\overline{\bm{H}}}_{is^\prime})
        -
        \pi_{s^\prime}(\overline{\bm{{H}}}_{is^\prime})
        \bigr|^2
    \right]^{\frac{1}{2}}\\
     \lesssim \ & 
    \sum_{s = 1}^{s_{\max}} \sum_{s^\prime = 1}^{s}
    \E\biggl[
        \bigl|
        m_s^*(\overline{\bW}_{i,s-1}, W_{is}, \overline{\bm F}_{is}; \bm\delta_{\underline{s}})
        -
        \hat m_s(\overline{\bW}_{i,s-1}, W_{is}, \widehat{\overline{\bm F}}_{is}; \bm\delta_{\underline{s}})
        \bigr|^2
    \biggr]^{\frac{1}{2}} \\
    &\hspace{2.5in} \times
    \E\left[
        \bigl|
        \hat\pi_{s^\prime}(\widehat{\overline{\bm{H}}}_{is^\prime})
        -
        \pi_{s^\prime}(\overline{\bm{{H}}}_{is^\prime})
        \bigr|^2
    \right]^{\frac{1}{2}}
\end{align*}
where the first inequality is by triangle inequality, the second inequality
is by applying the derived upper bound and using $\Pr(S_i = \tilde s)
\leq 1$. Finally, since every constant in the above derivation depends
on $\bm\delta$ only through $\delta_\ell$ and $\delta_u$, taking the
supremum over $\bm\delta \in \Delta^{s_{\max}}$ on both sides of the
resulting inequality yields the stated uniform bound.
\end{proof}

\subsubsection{Lemma 3}

\begin{lemma}\label{lemma3}
Let $r_n$ be a sequence of positive constants approaching zero as the sample size $N$ increases. For sufficiently large $N$ and each $s \in \mathcal{S}$, we have
\begin{align*}
     &\sup_{\bm\delta \in \Delta^{s_{\max}}} \E\left[ \bigl|m_s^*(\overline{\bW}_{i,s-1}, W_{is}, \overline{\bm F}_{is}; \bm\delta_{\underline{s}}) - \hat m_s(\overline{\bW}_{i,s-1}, W_{is}, \widehat{\overline{\bm F}}_{is}; \bm\delta_{\underline{s}}) \bigr|^2 \right]^\frac{1}{2} \\
     &\hspace{3.5in} \lesssim N^{-\frac{1}{4}} r_n
\end{align*}
where $m^*$ is the outcome model under the estimated stochastic intervention defined in Equation~\eqref{outcome_true}.
\end{lemma}

\begin{proof}
Now, notice that
\begin{align*}
    &\E\left[ \bigl|m_s^*(\overline{\bW}_{i,s-1}, W_{is}, \overline{\bm F}_{is}; \bm\delta_{\underline{s}}) - \hat m_s(\overline{\bW}_{i,s-1}, W_{is}, \widehat{\overline{\bm F}}_{is}; \bm\delta_{\underline{s}}) \bigr|^2 \right]^\frac{1}{2}\\
     = \ & \E\left[ \biggl| \biggl(m_s^*(\overline{\bW}_{i,s-1}, W_{is}, \overline{\bm F}_{is}; \bm\delta_{\underline{s}}) - m_s(\overline{\bW}_{i,s-1}, W_{is}, \overline{\bm F}_{is}; \bm\delta_{\underline{s}})\biggr) \right. \\
      &\quad \left. + \biggl(m_s(\overline{\bW}_{i,s-1}, W_{is}, \overline{\bm F}_{is}; \bm\delta_{\underline{s}}) - m_s(\overline{\bW}_{i,s-1}, W_{is}, \widehat{\overline{\bm F}}_{is}; \bm\delta_{\underline{s}})\biggr) \right. \\
      &\quad \left.  + \biggl(m_s(\overline{\bW}_{i,s-1}, W_{is}, \widehat{\overline{\bm F}}_{is}; \bm\delta_{\underline{s}}) - \hat m_s(\overline{\bW}_{i,s-1}, W_{is}, \widehat{\overline{\bm F}}_{is}; \bm\delta_{\underline{s}})\biggr)
      \biggr|^2 \right]^\frac{1}{2}\\
     \leq \ &  \E\left[ \bigl|m_s^*(\overline{\bW}_{i,s-1}, W_{is}, \overline{\bm F}_{is}; \bm\delta_{\underline{s}}) - m_s(\overline{\bW}_{i,s-1}, W_{is}, \overline{\bm F}_{is}; \bm\delta_{\underline{s}})\bigr|^2\right]^{\frac{1}{2}}\\
      & + \E\left[\bigl|m_s(\overline{\bW}_{i,s-1}, W_{is}, \overline{\bm F}_{is}; \bm\delta_{\underline{s}}) - m_s(\overline{\bW}_{i,s-1}, W_{is}, \widehat{\overline{\bm F}}_{is}; \bm\delta_{\underline{s}})\bigr|^2\right]^\frac{1}{2}\\
      & + \E\left[\bigl|m_s(\overline{\bW}_{i,s-1}, W_{is}, \widehat{\overline{\bm F}}_{is}; \bm\delta_{\underline{s}}) - \hat m_s(\overline{\bW}_{i,s-1}, W_{is}, \widehat{\overline{\bm F}}_{is}; \bm\delta_{\underline{s}})\bigr|^2 \right]^{\frac{1}{2}}\\
      \leq \ &  \E\left[ \bigl|m_s^*(\overline{\bW}_{i,s-1}, W_{is}, \overline{\bm F}_{is}; \bm\delta_{\underline{s}}) - m_s(\overline{\bW}_{i,s-1}, W_{is}, \overline{\bm F}_{is}; \bm\delta_{\underline{s}})\bigr|^2\right]^{\frac{1}{2}}\\
      & + \E\left[\bigl|m_s(\overline{\bW}_{i,s-1}, W_{is}, \overline{\bm F}_{is}; \bm\delta_{\underline{s}}) - m_s(\overline{\bW}_{i,s-1}, W_{is}, \widehat{\overline{\bm F}}_{is}; \bm\delta_{\underline{s}})\bigr|^2\right]^\frac{1}{2}
   + r_n N^{-1/4}\\
    \leq \ &  \E\left[ \bigl|m_s^*(\overline{\bW}_{i,s-1}, W_{is}, \overline{\bm F}_{is}; \bm\delta_{\underline{s}}) - m_s(\overline{\bW}_{i,s-1}, W_{is}, \overline{\bm F}_{is}; \bm\delta_{\underline{s}})\bigr|^2\right]^{\frac{1}{2}}\\
      & +  L \sum_{r = 1}^{s} \E\left[
        \bigl|
        \hat{\boldf}^{[s]}(\bR_{ir},r)-\boldf^{[s]}(\bR_{ir},r)
        \bigr|^2
    \right]^{\frac{1}{2}} + r_n N^{-1/4}\\
   \leq\ &  \E\left[ \bigl|m_s^*(\overline{\bW}_{i,s-1}, W_{is}, \overline{\bm F}_{is}; \bm\delta_{\underline{s}}) - m_s(\overline{\bW}_{i,s-1}, W_{is}, \overline{\bm F}_{is}; \bm\delta_{\underline{s}})\bigr|^2\right]^{\frac{1}{2}}\\
   & + (L s_{\max}+1) r_n N^{-1/4}
\end{align*}
where the first inequality is by triangular inequality, the second and
fourth inequalities are by Assumption~\ref{reg_cond}, and the third
inequality is by Lipschitz continuity of $m_s$, whose argument
contains the deconfounders at all segments $r \leq s$, so that the
representation errors accumulate over the entire history (the
resulting summation contains at most $s_{\max}$ terms).

Regarding the first term, we proceed by backward induction on $s$,
starting from the terminal stage. At the terminal stage $s = S_i$, we
have $m_{S_i}^* = m_{S_i}$ by the terminal condition of the recursion
in Equation~\eqref{outcome_true}, and hence the first term is exactly
zero. Next, fix $s < S_i$ and assume, as the induction hypothesis,
that
\begin{align*}
    \E\left[ \bigl|m_{s+1}^*(\overline{\bW}_{is}, w, \overline{\bm F}_{i,s+1}; \bm\delta_{\underline{s+1}}) - m_{s+1}(\overline{\bW}_{is}, w, \overline{\bm F}_{i,s+1}; \bm\delta_{\underline{s+1}})\bigr|^2\right]^{\frac{1}{2}} \ \lesssim \ r_n N^{-\frac{1}{4}}
\end{align*}
for $w \in \{0, 1\}$. Since $\hat p_{s+1}(\overline{\bm W}_{is})$ and
$p_{s+1}(\overline{\bm W}_{is})$ are functions of the conditioning
variables, a direct calculation yields
\begin{align*}
    &m_s^*(\overline{\bW}_{i,s-1}, W_{is}, \overline{\bm F}_{is}; \bm\delta_{\underline{s}}) - m_s(\overline{\bW}_{i,s-1}, W_{is}, \overline{\bm F}_{is}; \bm\delta_{\underline{s}})\\
  = \  & \E\biggl[
        \sum_{w \in \{0,1\}}
        d\hat Q_{s+1}(w; \overline{\bm W}_{is}, \delta_{s+1})
        \bigl\{m_{s+1}^*(\overline{\bW}_{is}, w, \overline{\bm F}_{i,s+1}; \bm\delta_{\underline{s+1}}) - m_{s+1}(\overline{\bW}_{is}, w, \overline{\bm F}_{i,s+1}; \bm\delta_{\underline{s+1}})\bigr\}
        \,\bigg|\, {\overline{\bm{H}}}_{is}, W_{is}
    \biggr]\\
    & + \E\biggl[
        \frac{
            \delta_{s+1} \hat p_{s+1}(\overline{\bm W}_{is})
            m_{s+1}(\overline{\bW}_{is}, 1, \overline{\bm F}_{i,s+1}; \bm\delta_{\underline{s+1}})
            +
            \{1-\hat p_{s+1}(\overline{\bm W}_{is})\}
            m_{s+1}(\overline{\bW}_{is}, 0, \overline{\bm F}_{i,s+1}; \bm\delta_{\underline{s+1}})
        }{
            \delta_{s+1} \hat p_{s+1}(\overline{\bm W}_{is})
            + 1 - \hat p_{s+1}(\overline{\bm W}_{is})
        }
        \,\bigg|\, {\overline{\bm{H}}}_{is}, W_{is}
    \biggr]\\
    & - \E\biggl[
        \frac{
            \delta_{s+1} p_{s+1}(\overline{\bm W}_{is})
            m_{s+1}(\overline{\bW}_{is}, 1, \overline{\bm F}_{i,s+1}; \bm\delta_{\underline{s+1}})
            +
            \{1- p_{s+1}(\overline{\bm W}_{is})\}
            m_{s+1}(\overline{\bW}_{is}, 0, \overline{\bm F}_{i,s+1}; \bm\delta_{\underline{s+1}})
        }{
            \delta_{s+1} p_{s+1}(\overline{\bm W}_{is})
            + 1 - p_{s+1}(\overline{\bm W}_{is})
        }
        \,\bigg|\, {\overline{\bm{H}}}_{is}, W_{is}
    \biggr]\\
    = \ & \E\biggl[
        \sum_{w \in \{0,1\}}
        d\hat Q_{s+1}(w; \overline{\bm W}_{is}, \delta_{s+1})
        \bigl\{m_{s+1}^*(\overline{\bW}_{is}, w, \overline{\bm F}_{i,s+1}; \bm\delta_{\underline{s+1}}) - m_{s+1}(\overline{\bW}_{is}, w, \overline{\bm F}_{i,s+1}; \bm\delta_{\underline{s+1}})\bigr\}
        \,\bigg|\, {\overline{\bm{H}}}_{is}, W_{is}
    \biggr]\\
    & + \frac{\delta_{s+1}
    \left\{ \hat p_{s+1}(\overline{\bm W}_{is}) - p_{s+1}(\overline{\bm W}_{is})\right\}
    }{
    [\delta_{s+1} \hat p_{s+1}(\overline{\bm W}_{is})
            + 1 - \hat p_{s+1}(\overline{\bm W}_{is})]
    [\delta_{s+1} p_{s+1}(\overline{\bm W}_{is})
            + 1 - p_{s+1}(\overline{\bm W}_{is})]
    }\\
    & \qquad \times \E\biggl[
    m_{s+1}(\overline{\bW}_{is}, 1, \overline{\bm F}_{i,s+1}; \bm\delta_{\underline{s+1}})
    -
    m_{s+1}(\overline{\bW}_{is}, 0, \overline{\bm F}_{i,s+1}; \bm\delta_{\underline{s+1}})
    \,\bigg|\, {\overline{\bm{H}}}_{is}, W_{is}
    \biggr],
\end{align*}
where the first equality results by adding and subtracting the
conditional expectation of the weighted average of $m_{s+1}$ under the
estimated intervention $\hat Q_{s+1}$ (the first term propagates the
estimation errors accumulated at the later stages), and the last
equality follows by combining the two fractions of the remaining two
terms over a common denominator and simplifying.

For the first term, the weights
$d\hat Q_{s+1}(w; \overline{\bm W}_{is}, \delta_{s+1})$ are
nonnegative and sum to one over $w \in \{0,1\}$. Therefore,
\begin{align*}
    &\E\left[\biggl|\E\biggl[
        \sum_{w \in \{0,1\}}
        d\hat Q_{s+1}(w; \overline{\bm W}_{is}, \delta_{s+1})
        \bigl\{m_{s+1}^*(\overline{\bW}_{is}, w, \overline{\bm F}_{i,s+1}; \bm\delta_{\underline{s+1}}) \right.\\
    &\left. \hspace{2in}- m_{s+1}(\overline{\bW}_{is}, w, \overline{\bm F}_{i,s+1}; \bm\delta_{\underline{s+1}})\bigr\}
        \,\bigg|\, {\overline{\bm{H}}}_{is}, W_{is}
    \biggr]\biggr|^2\right]^{\frac{1}{2}}\\
    \leq \ & \sum_{w \in \{0,1\}}\E\left[ \bigl|m_{s+1}^*(\overline{\bW}_{is}, w, \overline{\bm F}_{i,s+1}; \bm\delta_{\underline{s+1}}) - m_{s+1}(\overline{\bW}_{is}, w, \overline{\bm F}_{i,s+1}; \bm\delta_{\underline{s+1}})\bigr|^2\right]^{\frac{1}{2}}\\
    \lesssim \ & r_n N^{-\frac{1}{4}},
\end{align*}
where the first inequality is by the triangle inequality together with
Jensen's inequality applied to the conditional expectation, and the
second inequality is by the induction hypothesis.

For the second term, we have
\begin{align*}
    \ &\left|\frac{\delta_{s+1}
    \left\{ \hat p_{s+1}(\overline{\bm W}_{is}) - p_{s+1}(\overline{\bm W}_{is})\right\}
    \E\biggl[
    m_{s+1}(\overline{\bW}_{is}, 1, \overline{\bm F}_{i,s+1}; \bm\delta_{\underline{s+1}})
    -
    m_{s+1}(\overline{\bW}_{is}, 0, \overline{\bm F}_{i,s+1}; \bm\delta_{\underline{s+1}})
    \,\bigg|\, {\overline{\bm{H}}}_{is}, W_{is}
    \biggr]
    }{
    [\delta_{s+1} \hat p_{s+1}(\overline{\bm W}_{is})
            + 1 - \hat p_{s+1}(\overline{\bm W}_{is})]
    [\delta_{s+1} p_{s+1}(\overline{\bm W}_{is})
            + 1 - p_{s+1}(\overline{\bm W}_{is})]
    }\right|\\
    &\leq \ \frac{2C\delta_u }{ \min\{1, \delta_\ell\}^2 }
    \bigl| \hat p_{s+1}(\overline{\bm W}_{is}) - p_{s+1}(\overline{\bm W}_{is})\bigr|,
\end{align*}
where the inequality holds due to the boundedness of the outcome
model (Assumption~\ref{reg_cond}), $\delta_{s+1} \leq \delta_u$, and
$\delta_{s+1} p_{s+1} + 1 - p_{s+1} \geq \min\{1, \delta_\ell\}$ (and
similarly for $\hat p_{s+1}$). Thus, combining the two bounds by the
triangle inequality and applying Assumption~\ref{reg_cond} to the
$L_2$ norm of $\hat p_{s+1} - p_{s+1}$, we obtain
\begin{align*}
    \E\left[ \bigl|m_s^*(\overline{\bW}_{i,s-1}, W_{is}, \overline{\bm F}_{is}; \bm\delta_{\underline{s}}) - m_s(\overline{\bW}_{i,s-1}, W_{is}, \overline{\bm F}_{is}; \bm\delta_{\underline{s}})\bigr|^2\right]^{\frac{1}{2}} &\lesssim \frac{2C\delta_u }{ \min\{1, \delta_\ell\}^2 }
       r_n N^{-\frac{1}{4}} + r_n N^{-\frac{1}{4}} \ \lesssim \ r_n N^{-\frac{1}{4}},
\end{align*}
Since the above derivation holds verbatim with $W_{is}$ replaced by a
fixed argument $w \in \{0,1\}$, the same bound holds for both fixed
arguments, which yields the induction hypothesis at stage $s$ and
closes the induction. Note that the difference is exactly zero on the
event $\{S_i = s\}$ by the terminal condition, so the unconditional
$L_2$ norms above are well defined. Since $s_{\max}$ is finite, the
backward recursion terminates after at most $s_{\max}$ steps, and the
implied constants remain bounded. Substituting this bound into the first
display completes the proof of the pointwise bound. Finally, because
every constant above depends on $\bm\delta$ only through $\delta_\ell$
and $\delta_u$, and the nuisance rate conditions in
Assumption~\ref{reg_cond} hold uniformly over
$\bm\delta \in \Delta^{s_{\max}}$, the bound holds uniformly over
$\bm\delta \in \Delta^{s_{\max}}$.
\end{proof}

\subsubsection{Lemma 4}
\begin{lemma}\label{lemma4}
\begin{align*}
    &\int
\left[
    \zeta\bigl(\cD_i; \bm\delta, \{ \hat\pi_s, \hat m_s, \widehat{\tilde m}_s, \hat p_s, \hat\boldf^{[s]}\}_{s=1}^{S_i}\bigr)
    +
    \Psi(\bm\delta;\hat Q)
    -
    \Psi(\bm\delta;Q)
\right]
dF(\cD_i) =  R_{2,1}(\bm\delta)
+
R_{2,2}(\bm\delta)
+
R_{2,3}(\bm\delta),
\end{align*}
where $R_{2,1}(\bm\delta)$, $R_{2,2}(\bm\delta)$, and $R_{2,3}(\bm\delta)$ are defined in Equations~\eqref{reminder_r21}--\eqref{reminder_r23}.
\end{lemma}

\begin{proof}
Similar to Lemma~6 of \cite{kennedy_nonparametric_2019}, for any $\tilde s \in \mathcal{S}$ we have
\begin{align*}
    &\Psi_{\tilde s}(\bm\delta; \hat Q) - \Psi_{\tilde s}(\bm\delta; Q)\\ 
   = \ &  \int
                m_{\tilde s}(\overline{\bW}_{i,\tilde s-1}, W_{i\tilde s}, \overline{\bm F}_{i\tilde s}; \bm\delta_{\underline{\tilde s}})
    \prod_{s = 1}^{\tilde s} d \hat Q_s(W_{is}; \overline{\bm{W}}_{i,s-1}, \delta_s)\, dF( \bR_{is} \mid \overline{\bm H}_{i,s-1}, W_{i,s-1}, S_i = \tilde s)\\
    & - \int
                m_{\tilde s}(\overline{\bW}_{i,\tilde s-1}, W_{i\tilde s}, \overline{\bm F}_{i\tilde s}; \bm\delta_{\underline{\tilde s}})
    \prod_{s = 1}^{\tilde s} d Q_s(W_{is}; \overline{\bm{W}}_{i,s-1}, \delta_s)\,  dF( \bR_{is} \mid \overline{\bm H}_{i,s-1}, W_{i,s-1}, S_i = \tilde s)\\
  = \ &  \int m_{\tilde s}(\overline{\bW}_{i,\tilde s-1}, W_{i\tilde s}, \overline{\bm F}_{i\tilde s}; \bm\delta_{\underline{\tilde s}})
    \prod_{s = 1}^{\tilde s} d \hat Q_s(W_{is}; \overline{\bm{W}}_{i,s-1}, \delta_s)\, dF( \bR_{is} \mid \overline{\bm H}_{i,s-1}, W_{i,s-1}, S_i = \tilde s)\\
    & - \int m_{\tilde s}(\overline{\bW}_{i,\tilde s-1}, W_{i\tilde s}, \overline{\bm F}_{i\tilde s}; \bm\delta_{\underline{\tilde s}}) d  Q_{\tilde s}(W_{i\tilde s}; \overline{\bm{W}}_{i,\tilde s-1}, \delta_{\tilde s}) dF( \bR_{i \tilde s} \mid \overline{\bm H}_{i,\tilde s-1}, W_{i,\tilde s-1}, S_i = \tilde s)\\
      &\quad\quad \times \prod_{r = 1}^{\tilde s -1}  d  \hat Q_{r}(W_{ir}; \overline{\bm{W}}_{i,r-1}, \delta_{r}) dF( \bR_{ir} \mid \overline{\bm H}_{i,r-1}, W_{i,r-1}, S_i = \tilde s)\\
    & - \int
                m_{\tilde s}(\overline{\bW}_{i,\tilde s-1}, W_{i\tilde s}, \overline{\bm F}_{i\tilde s}; \bm\delta_{\underline{\tilde s}})
    \prod_{s = 1}^{\tilde s} d Q_s(W_{is}; \overline{\bm{W}}_{i,s-1}, \delta_s)\, dF( \bR_{is} \mid \overline{\bm H}_{i,s-1}, W_{i,s-1}, S_i = \tilde s)\\
      & + \int m_{\tilde s}(\overline{\bW}_{i,\tilde s-1}, W_{i\tilde s}, \overline{\bm F}_{i\tilde s}; \bm\delta_{\underline{\tilde s}}) d  Q_{\tilde s}(W_{i\tilde s}; \overline{\bm{W}}_{i,\tilde s-1}, \delta_{\tilde s}) dF( \bR_{i \tilde s} \mid \overline{\bm H}_{i,\tilde s-1}, W_{i,\tilde s-1}, S_i = \tilde s)\\
      &\qquad \times \prod_{r = 1}^{\tilde s -1}  d  \hat Q_{r}(W_{ir}; \overline{\bm{W}}_{i,r-1}, \delta_{r}) dF( \bR_{i r} \mid \overline{\bm H}_{i,r-1}, W_{i,r-1}, S_i = \tilde s)\\
    = \ &   \int
                m_{\tilde s}(\overline{\bW}_{i,\tilde s-1}, W_{i\tilde s}, \overline{\bm F}_{i\tilde s}; \bm\delta_{\underline{\tilde s}}) 
      \left\{  d \hat  Q_{\tilde s}(W_{i\tilde s}; \overline{\bm{W}}_{i,\tilde s-1}, \delta_{\tilde s}) -  d  Q_{\tilde s}(W_{i\tilde s}; \overline{\bm{W}}_{i,\tilde s-1}, \delta_{\tilde s})\right\} dF( \bR_{i \tilde s} \mid \overline{\bm H}_{i,\tilde s-1}, W_{i,\tilde s-1}, S_i = \tilde s)
      \\
      &\quad \times \prod_{s = 1}^{\tilde s-1} d \hat Q_s(W_{is}; \overline{\bm{W}}_{i,s-1}, \delta_s) dF( \bR_{is} \mid \overline{\bm H}_{i,s-1}, W_{i, s-1}, S_i = \tilde s)\\
      & + \int
      m_{\tilde s}(\overline{\bW}_{i,\tilde s-1}, W_{i\tilde s}, \overline{\bm F}_{i\tilde s}; \bm\delta_{\underline{\tilde s}}) 
    d  Q_{\tilde s}(W_{i\tilde s}; \overline{\bm{W}}_{i,\tilde s-1}, \delta_{\tilde s}) dF( \bR_{i \tilde s} \mid \overline{\bm H}_{i,\tilde s-1}, W_{i,\tilde s-1}, S_i = \tilde s) \\
    &\quad\quad \times 
    \biggl\{ \prod_{s = 1}^{\tilde s - 1} d \hat Q_s(W_{is};
      \overline{\bm{W}}_{i,s-1}, \delta_s)dF( \bR_{is} \mid \overline{\bm H}_{i,s-1}, W_{i, s-1}, S_i = \tilde s) \\
  & \hspace{1in} - 
    \prod_{s = 1}^{\tilde s - 1} d Q_s(W_{is}; \overline{\bm{W}}_{i,s-1}, \delta_s)dF( \bR_{is} \mid \overline{\bm H}_{i,s-1}, W_{i, s-1}, S_i = \tilde s)\biggr\}
    \\
   = \ &  \int
                m_{\tilde s}(\overline{\bW}_{i,\tilde s-1}, W_{i\tilde s}, \overline{\bm F}_{i\tilde s}; \bm\delta_{\underline{\tilde s}}) 
      \left\{  d \hat  Q_{\tilde s}(W_{i\tilde s}; \overline{\bm{W}}_{i,\tilde s-1}, \delta_{\tilde s}) -  d  Q_{\tilde s}(W_{i\tilde s}; \overline{\bm{W}}_{i,\tilde s-1}, \delta_{\tilde s})\right\}
      dF( \bR_{i \tilde s} \mid \overline{\bm H}_{i,\tilde s-1}, W_{i,\tilde s-1}, S_i = \tilde s)
      \\
      &\quad \times \prod_{s = 1}^{\tilde s-1} d \hat Q_s(W_{is}; \overline{\bm{W}}_{i,s-1}, \delta_s) dF( \bR_{is} \mid \overline{\bm H}_{i,s-1}, W_{i, s-1}, S_i = \tilde s)\\
      & + \int
      m_{\tilde s-1}(\overline{\bW}_{i,\tilde s-2}, W_{i,\tilde s-1}, \overline{\bm F}_{i,\tilde s-1}; \bm\delta_{\underline{\tilde s-1}}) \times \biggl\{ \prod_{s = 1}^{\tilde s - 1} d \hat Q_s(W_{is}; \overline{\bm{W}}_{i,s-1}, \delta_s)dF( \bR_{is} \mid \overline{\bm H}_{i,s-1}, W_{i, s-1}, S_i = \tilde s) \\
    &\hspace{2.0in} - 
    \prod_{s = 1}^{\tilde s - 1} d  Q_s(W_{is}; \overline{\bm{W}}_{i,s-1}, \delta_s)dF( \bR_{is} \mid \overline{\bm H}_{i,s-1}, W_{i, s-1}, S_i = \tilde s)\biggr\}
    \\
  = \ &  \sum_{s = 1}^{\tilde s}  \int
                m_{s}(\overline{\bW}_{i,s-1}, W_{i s}, \overline{\bm F}_{is}; \bm\delta_{\underline{s}}) 
      \left\{  d \hat  Q_{s}(W_{i s}; \overline{\bm{W}}_{i,s-1}, \delta_{s}) -  d  Q_{s}(W_{i s}; \overline{\bm{W}}_{i,s-1}, \delta_{ s})\right\}\\
      &\hspace{.5in} \times 
      dF( \bR_{is} \mid \overline{\bm H}_{i,s-1}, W_{i, s-1}, S_i = \tilde s)
      \prod_{r = 1}^{ s-1} d \hat Q_r(W_{ir}; \overline{\bm{W}}_{i,r-1}, \delta_r) dF( \bR_{ir} \mid \overline{\bm H}_{i,r-1}, W_{i, r-1}, S_i = \tilde s)\\
    = \  & \sum_{s = 1}^{\tilde s}  \int
      \biggl\{  
      \frac{\delta_s \{ \hat p_s(\overline{\bm W}_{i,s-1}) - p_s(\overline{\bm W}_{i,s-1})\}
      \{m_{s}(\overline{\bW}_{i,s-1}, 1, \overline{\bm F}_{is}; \bm\delta_{\underline{s}}) - m_{s}(\overline{\bW}_{i,s-1}, 0, \overline{\bm F}_{is}; \bm\delta_{\underline{s}})\}
      }{
        [\delta_s  \hat p_s(\overline{\bm W}_{i,s-1}) + \{1 - \hat p_s(\overline{\bm W}_{i,s-1})\}] [\delta_s   p_s(\overline{\bm W}_{i,s-1}) + \{1 -  p_s(\overline{\bm W}_{i,s-1})\}]
      }
      \biggr\}\\
      &\hspace{.5in} \times dF( \bR_{is} \mid \overline{\bm H}_{i,s-1}, W_{i, s-1}, S_i = \tilde s) \prod_{r = 1}^{ s-1} d \hat Q_r(W_{ir}; \overline{\bm{W}}_{i,r-1}, \delta_r) dF( \bR_{ir} \mid \overline{\bm H}_{i,r-1}, W_{i, r-1}, S_i = \tilde s),
\end{align*}
where the first equality is by the g-formula representation of
$\Psi_{\tilde s}$, in which the joint law under a stochastic
intervention factorizes into the product, over all stages
$s = 1, \ldots, \tilde s$, of the intervention distribution
$dQ_s(W_{is}; \overline{\bm W}_{i,s-1}, \delta_s)$ and the covariate
transition
$dF(\bR_{is} \mid \overline{\bm H}_{i,s-1}, W_{i,s-1}, S_i = \tilde
s)$; the second equality results by adding and subtracting the same
quantity, the third equality is just rearranging each term, the fourth
equality is from the definition of outcome model, the fifth equality
is due to the recursive application of the same decomposition in the
fourth equality, and the last equality is by explicitly writing the
measure based on stochastic intervention. Therefore,
\begin{align*}
  &\Psi(\bm\delta; \hat Q) - \Psi(\bm\delta; Q) \\
  = \ &  \sum_{\tilde{s} = 1}^{s_{\max} } \left\{\Psi_{\tilde s}(\bm\delta; \hat Q) - \Psi_{\tilde s}(\bm\delta; Q)\right\} \Pr(S_i = \tilde s)\\
        = \ & \sum_{\tilde{s} = 1}^{s_{\max} } \Biggl(\sum_{s = 1}^{\tilde s}  \int 
      \frac{\delta_s \{ \hat p_s(\overline{\bm W}_{i,s-1}) - p_s(\overline{\bm W}_{i,s-1})\}
      \{m_{s}(\overline{\bW}_{i,s-1}, 1, \overline{\bm F}_{is}; \bm\delta_{\underline{s}}) - m_{s}(\overline{\bW}_{i,s-1}, 0, \overline{\bm F}_{is}; \bm\delta_{\underline{s}})\}
      }{
        [\delta_s  \hat p_s(\overline{\bm W}_{i,s-1}) + \{1 - \hat p_s(\overline{\bm W}_{i,s-1})\}] [\delta_s   p_s(\overline{\bm W}_{i,s-1}) + \{1 -  p_s(\overline{\bm W}_{i,s-1})\}]
      }\\
      &\hspace{0.1in} \times dF(\bR_{is} \mid \overline{\bm H}_{i,s-1}, W_{i,s-1}, S_i = \tilde s) \prod_{r = 1}^{ s-1} d \hat Q_r(W_{ir}; \overline{\bm{W}}_{i,r-1}, \delta_r) dF(\bR_{ir} \mid \overline{\bm H}_{i,r-1}, W_{i,r-1}, S_i = \tilde s)\Biggr)\Pr(S_i = \tilde s)\\
      = \ & \sum_{s = 1}^{s_{\max} } \Biggl(\int
      \frac{\delta_s \{ \hat p_s(\overline{\bm W}_{i,s-1}) - p_s(\overline{\bm W}_{i,s-1})\}
      \{  m_{s}(\overline{\bW}_{i,s-1}, 1, \overline{\bm F}_{is}; \bm\delta_{\underline{s}}) - m_{s}(\overline{\bW}_{i,s-1}, 0, \overline{\bm F}_{is}; \bm\delta_{\underline{s}})\}
      }{
        [\delta_s  \hat p_s(\overline{\bm W}_{i,s-1}) + \{1 - \hat p_s(\overline{\bm W}_{i,s-1})\}] [\delta_s   p_s(\overline{\bm W}_{i,s-1}) + \{1 -  p_s(\overline{\bm W}_{i,s-1})\}]
      } \\
      &\hspace{0.7in} \times dF(\bR_{is} \mid \overline{\bm H}_{i,s-1}, W_{i,s-1}, S_i) \prod_{r = 1}^{ s-1} d \hat Q_r(W_{ir}; \overline{\bm{W}}_{i,r-1}, \delta_r) dF(\bR_{ir} \mid \overline{\bm H}_{i,r-1}, W_{i,r-1}, S_i)\, dP(S_i \mid S_i \geq  s)\Biggr)\Pr(S_i \geq s)
\end{align*}
where the first line is by law of total probability, the second line
is by plugging in the derived formula, and the last line is by
re-indexing the double summation over $(\tilde s, s)$. In the last
line, the exact value of $S_i$ remains an integration variable inside
the mixture because $m_s$ and $\overline{\bm H}_{is}$ depend on $S_i$;
specifically, we use the identity
$\sum_{\tilde s \geq s} (\cdot) \Pr(S_i = \tilde s) = \Pr(S_i \geq s)
\int (\cdot)\, dP(S_i \mid S_i \geq s)$.

Now, the correction term is
\begin{align*}
&\int
    \zeta\bigl(\cD_i; \bm\delta, \{ \hat\pi_s, \hat m_s, \widehat{\tilde m}_s, \hat p_s, \hat\boldf^{[s]}\}_{s=1}^{S_i}\bigr)
dF(\cD_i) \\
= \ & \sum_{\tilde s = 1}^{s_{\max}} \sum_{s = 1}^{\tilde s} \int   \frac{\delta_s \{W_{is} - \hat p_s(\overline{\bm W}_{i,s-1})\} \{ \widehat{\tilde m}_s(\overline{\bm{W}}_{i,s-1},  1; \bm\delta_{-s})  - \widehat{\tilde m}_s(\overline{\bm{W}}_{i,s-1},  0; \bm\delta_{-s}) \} }{ [\delta_s  \hat p_s(\overline{\bm W}_{i,s-1}) + \{1 - \hat p_s(\overline{\bm W}_{i,s-1})\}]^2 } dF(\overline{\bm{W}}_{is} \mid S_i = \tilde s)\Pr(S_i = \tilde s) \\
%= \ & \sum_{s = 1}^{s_{\max}} \sum_{\tilde s = s}^{\tilde s} \int   \frac{\delta_s (W_{is} - \hat p_s(\overline{\bm W}_{i,s-1})) \{ \widehat{\tilde m}_s(\overline{\bm{W}}_{i,s-1},  1; \bm\delta_{-s})  - \widehat{\tilde m}_s(\overline{\bm{W}}_{i,s-1},  0; \bm\delta_{-s}) \} }{ [\delta_s  \hat p_s(\overline{\bm W}_{i,s-1}) + \{1 - \hat p_s(\overline{\bm W}_{i,s-1})\}]^2 } dF(\overline{\bm{W}}_{is} \mid S_i = \tilde s)\Pr(S_i = \tilde s) \\
= \ & \sum_{s = 1}^{s_{\max}} \int   \frac{\delta_s \{W_{is} - \hat p_s(\overline{\bm W}_{i,s-1})\} \{ \widehat{\tilde m}_s(\overline{\bm{W}}_{i,s-1},  1; \bm\delta_{-s})  - \widehat{\tilde m}_s(\overline{\bm{W}}_{i,s-1},  0; \bm\delta_{-s}) \} }{ [\delta_s  \hat p_s(\overline{\bm W}_{i,s-1}) + \{1 - \hat p_s(\overline{\bm W}_{i,s-1})\}]^2 } dF(\overline{\bm{W}}_{is} \mid S_i \geq s)\Pr(S_i \geq s) \\
=\ & \sum_{s = 1}^{s_{\max}} \int \  \frac{\delta_s \{ p_s(\overline{\bm W}_{i,s-1}) - \hat p_s(\overline{\bm W}_{i,s-1})\} \{ \widehat{\tilde m}_s(\overline{\bm{W}}_{i,s-1},  1; \bm\delta_{-s})  - \widehat{\tilde m}_s(\overline{\bm{W}}_{i,s-1},  0; \bm\delta_{-s}) \} }{ [\delta_s  \hat p_s(\overline{\bm W}_{i,s-1}) + \{1 - \hat p_s(\overline{\bm W}_{i,s-1})\}]^2 }\\
&\hspace{3in} dF(\overline{\bm{W}}_{i,s-1} \mid S_i \geq s)\Pr(S_i \geq s) \\
= \ &  \sum_{s = 1}^{s_{\max} } \int  \frac{\delta_s \{ p_s(\overline{\bm W}_{i,s-1}) - \hat p_s(\overline{\bm W}_{i,s-1})\} }{ [\delta_s  \hat p_s(\overline{\bm W}_{i,s-1}) + \{1 - \hat p_s(\overline{\bm W}_{i,s-1})\}]^2 }
\left[ \left\{\check{m}_s(\overline{\bm{W}}_{i,s-1},  1; \bm\delta_{-s}) - \check{m}_s(\overline{\bm{W}}_{i,s-1},  0; \bm\delta_{-s})\right\}\right.
\\
& \qquad - \left.\left\{\check{m}_s(\overline{\bm{W}}_{i,s-1},  1; \bm\delta_{-s}) - \check{m}_s(\overline{\bm{W}}_{i,s-1},  0; \bm\delta_{-s})\right\} + \left\{\widehat{\tilde m}_s(\overline{\bm{W}}_{i,s-1},  1; \bm\delta_{-s})  - \widehat{\tilde m}_s(\overline{\bm{W}}_{i,s-1},  0; \bm\delta_{-s})\right\} \right] \\
&\qquad \qquad dF(\overline{\bm{W}}_{i,s-1} \mid  S_i  \geq s)\Pr(S_i \geq s)\\
= \ & Z + R_{2,1}(\bm\delta)
\end{align*}
where $R_{2,1}$ is defined in Equation~\eqref{reminder_r21}, and
\begin{align*}
Z &:= \sum_{s = 1}^{s_{\max}}  \int  \frac{\delta_s \{ p_s(\overline{\bm W}_{i,s-1}) - \hat p_s(\overline{\bm W}_{i,s-1})\} }{ [\delta_s  \hat p_s(\overline{\bm W}_{i,s-1}) + \{1 - \hat p_s(\overline{\bm W}_{i,s-1})\}]^2 }\\
&\hspace{1in}\left\{\check{ m}_s(\overline{\bm{W}}_{i,s-1},  1; \bm\delta_{-s}) - \check{ m}_s(\overline{\bm{W}}_{i,s-1},  0; \bm\delta_{-s})\right\}dF(\overline{\bm{W}}_{i,s-1} \mid S_i \geq s)\Pr(S_i \geq s)
%R_{2,1}(\bm\delta) &:= \sum_{s = 1}^{s_{\max}} \int   \frac{\delta_s
%                     \{ p_s(\overline{\bm W}_{i,s-1}) - \hat p_s(\overline{\bm W}_{i,s-1})\} }{ [\delta_s  \hat p_s(\overline{\bm W}_{i,s-1}) + \{1 - \hat p_s(\overline{\bm W}_{i,s-1})\}]^2 } \left[
%\left\{\widehat{\tilde m}_s(\overline{\bm{W}}_{i,s-1},  1; \bm\delta_{-s})  - \widehat{\tilde m}_s(\overline{\bm{W}}_{i,s%-1},  0;\tilde{\bm{\delta}}_{s+1})\right\}\right.\\
%&\hspace{1in} -\left. \left\{
%\check{m}_s(\overline{\bm{W}}_{i,s-1},  1; \bm\delta_{-s}) - \check{ m}_s(\overline{\bm{W}}_{i,s-1},  0; \bm\delta_{-s})\right\} \right]dF(\overline{\bm{W}}_{i,s-1}, S_i \geq s).
\end{align*}
In the above calculation, the first equality follows from the
definition of the augmentation term. The second equality is obtained
by reindexing the summation, which is valid here because the
integrand depends on $S_i$ only through the event $S_i \geq s$. The
third equality follows from integrating $W_{is}$ under its
conditional distribution given
$(\overline{\bm W}_{i,s-1}, S_i \geq s)$, while the fourth is
obtained by adding and subtracting the same quantity.

Now, we can  further transform $Z$ as
\begin{align*}
  & Z \\
  = \ &  \sum_{s = 1}^{s_{\max}}  \int  \frac{\delta_s \{p_s(\overline{\bm W}_{i,s-1}) - \hat p_s(\overline{\bm W}_{i,s-1})\} }{ [\delta_s  \hat p_s(\overline{\bm W}_{i,s-1}) + \{1 - \hat p_s(\overline{\bm W}_{i,s-1})\}]^2 }  \left[
\prod_{s^\prime=1}^{s-1}
\frac{\delta_{s^\prime} W_{is^\prime}  \frac{ \hat p_{s^\prime}(\overline{\bm W}_{i,s^\prime-1})}{\pi_{s^\prime}(\overline{\bm H}_{is^\prime})} + (1 - W_{is^\prime})
    \frac{1 - \hat p_{s^\prime}(\overline{\bm W}_{i,s^\prime-1})}{1 - \pi_{s^\prime}(\overline{\bm
    H}_{is^\prime})}}{ \delta_{s^\prime} \hat p_{s^\prime}(\overline{\bm W}_{i,s^\prime-1}) + 1 -
    \hat p_{s^\prime}(\overline{\bm W}_{i,s^\prime-1}) }
\right] \\
&\quad\left\{ \hat{ m}_s(\overline{\bW}_{i,s-1}, 1, \widehat{\overline{\bm F}}_{is}; \bm\delta_{\underline{s}}) - \hat{ m}_s(\overline{\bW}_{i,s-1}, 0, \widehat{\overline{\bm F}}_{is}; \bm\delta_{\underline{s}})
\right\} \\
&\qquad dF(\overline{\bm R}_{is} \mid \overline{\bm W}_{i,s-1}, S_i)\, dF(\overline{\bm{W}}_{i,s-1} \mid S_i)\, dP(S_i \mid S_i \geq s)\Pr(S_i \geq s) \\
%= \ & - \sum_{s = 1}^{s_{\max}}  \int \frac{\delta_s ( \hat p_s(\overline{\bm W}_{i,s-1}) - p_s(\overline{\bm W}_{i,s-1})) }{ [\delta_s  \hat p_s(\overline{\bm W}_{i,s-1}) + \{1 - \hat p_s(\overline{\bm W}_{i,s-1})\}]^2 }  \left[
%\prod_{s^\prime=1}^{s-1}
%\frac{\delta_{s^\prime} W_{is^\prime}  \frac{ \hat p_{s^\prime}(\overline{\bm W}_{i,s^\prime-1})}{\pi_{s^\prime}(\overline{\bm H}_{is^\prime})} + (1 - W_{is^\prime})
%    \frac{1 - \hat p_{s^\prime}(\overline{\bm W}_{i,s^\prime-1})}{1 - \pi_{s^\prime}(\overline{\bm
%    H}_{is^\prime})}}{ \delta_{s^\prime} \hat p_{s^\prime}(\overline{\bm W}_{i,s^\prime-1}) + 1 -
%    \hat p_{s^\prime}(\overline{\bm W}_{i,s^\prime-1}) }
%\right] \\
%&\qquad \left\{\hat{ m}_s(\overline{\bW}_{i,s-1}, 1, \widehat{\overline{\bm F}}_{is}; \bm\delta_{\underline{s}}) - \hat{ m}_s(\overline{\bW}_{i,s-1}, 0, \widehat{\overline{\bm F}}_{is}; \bm\delta_{\underline{s}}) 
%\right\} dF(\overline{\bm R}_{is} \mid \overline{\bm W}_{i,s-1}, S_i \geq s) dF(\overline{\bm{W}}_{i,s-1} \mid S_i \geq s)\Pr(S_i \geq s)\\
= \ &  - \sum_{s = 1}^{s_{\max}}  \int  \frac{\delta_s \{\hat p_s(\overline{\bm W}_{i,s-1}) - p_s(\overline{\bm W}_{i,s-1})\} }{ [\delta_s  \hat p_s(\overline{\bm W}_{i,s-1}) + \{1 - \hat p_s(\overline{\bm W}_{i,s-1})\}]^2 }\\
&\hspace{1in} \times
\left\{ \hat{ m}_s(\overline{\bW}_{i,s-1}, 1, \widehat{\overline{\bm F}}_{is}; \bm\delta_{\underline{s}}) - \hat{ m}_s(\overline{\bW}_{i,s-1}, 0, \widehat{\overline{\bm F}}_{is}; \bm\delta_{\underline{s}})
\right\}\\
&\hspace{-0.1in} dF(\bR_{is} \mid \overline{\bm H}_{i,s-1}, W_{i,s-1}, S_i)
\prod_{s^\prime=1}^{s-1}
d\hat{Q}_{s^\prime}(W_{is^\prime}; \overline{\bm W}_{i,s^\prime-1}, \delta_{s^\prime})
dF(\bR_{is^\prime} \mid \overline{\bm H}_{i,s^\prime-1}, W_{i,s^\prime-1}, S_i)\, dP(S_i \mid S_i \geq s)\Pr(S_i \geq s),
\end{align*}
where the first equality follows by explicitly writing the weight in
$\check{m}_s$ and the second equality holds by rearranging the
measure.
\begin{comment}
notice that we have an identity
\begin{align*}
&\prod_{s' = 1}^{s-1} \frac{W_{is'}\pi_{s'}(\overline{\bm H}_{is'}) + (1 - W_{is'})\{1 - \pi_{s'}(\overline{\bm H}_{is'})\} }{
W_{is'}\hat\pi_{s'}(\widehat{\overline{\bm H}}_{is'}) + (1 - W_{is'})\{1 - \hat\pi_{s'}(\widehat{\overline{\bm H}}_{is'})\}
}\\
&= 1 + \sum_{r = 1}^{s - 1}\Biggl[ \frac{W_{ir}\pi_{r}(\overline{\bm H}_{ir}) + (1 - W_{ir})\{1 - \pi_{r}(\overline{\bm H}_{ir})\} }{
W_{ir}\hat\pi_{r}(\widehat{\overline{\bm H}}_{ir}) + (1 - W_{ir})\{1 - \hat\pi_{r}(\widehat{\overline{\bm H}}_{ir})\}
} - 1\Biggr] \times \prod_{s'= 1}^{r-1} \frac{W_{is'}\pi_{s'}(\overline{\bm H}_{is'}) + (1 - W_{is'})\{1 - \pi_{s'}(\overline{\bm H}_{is'})\} }{
W_{is'}\hat\pi_{s'}(\widehat{\overline{\bm H}}_{is'}) + (1 - W_{is'})\{1 - \hat\pi_{s'}(\widehat{\overline{\bm H}}_{is'})\}
}\\
&= 1 + \sum_{r = 1}^{s - 1}\Biggl[ \frac{ (2W_{ir} - 1)\{\pi_{r}(\overline{\bm H}_{ir}) - \hat\pi_{r}(\widehat{\overline{\bm H}}_{ir})\} }{
W_{ir}\hat\pi_{r}(\widehat{\overline{\bm H}}_{ir}) + (1 - W_{ir})\{1 - \hat\pi_{r}(\widehat{\overline{\bm H}}_{ir})\}
} \Biggr] \times \prod_{s'= 1}^{r-1} \frac{W_{is'}\pi_{s'}(\overline{\bm H}_{is'}) + (1 - W_{is'})\{1 - \pi_{s'}(\overline{\bm H}_{is'})\} }{
W_{is'}\hat\pi_{s'}(\widehat{\overline{\bm H}}_{is'}) + (1 - W_{is'})\{1 - \hat\pi_{s'}(\widehat{\overline{\bm H}}_{is'})\}
}.
\end{align*}
\end{comment}
Therefore,
\begin{align*}
    &Z + \Psi(\bm\delta;\hat Q)
    -
    \Psi(\bm\delta;Q)\\
    = \ &  \sum_{s = 1}^{s_{\max}} \int \left[
      \frac{\delta_s \{ \hat p_s(\overline{\bm W}_{i,s-1}) - p_s(\overline{\bm W}_{i,s-1})\}
      \{m_{s}(\overline{\bW}_{i,s-1}, 1, \overline{\bm F}_{is}; \bm\delta_{\underline{s}}) - m_{s}(\overline{\bW}_{i,s-1}, 0, \overline{\bm F}_{is}; \bm\delta_{\underline{s}})\}
      }{
        [\delta_s  \hat p_s(\overline{\bm W}_{i,s-1}) + \{1 - \hat p_s(\overline{\bm W}_{i,s-1})\}] [\delta_s   p_s(\overline{\bm W}_{i,s-1}) + \{1 -  p_s(\overline{\bm W}_{i,s-1})\}]
      }
      \right.\\
      &\qquad\qquad \left. -  \frac{\delta_s \{\hat p_s(\overline{\bm W}_{i,s-1}) - p_s(\overline{\bm W}_{i,s-1})\} }{ [\delta_s  \hat p_s(\overline{\bm W}_{i,s-1}) + \{1 - \hat p_s(\overline{\bm W}_{i,s-1})\}]^2 }
\left\{ \hat{ m}_s(\overline{\bW}_{i,s-1}, 1, \widehat{\overline{\bm F}}_{is}; \bm\delta_{\underline{s}}) - \hat{ m}_s(\overline{\bW}_{i,s-1}, 0, \widehat{\overline{\bm F}}_{is}; \bm\delta_{\underline{s}}) 
\right\} \right]\\
&\hspace{-0.2in} \times  dF(\bR_{is} \mid \overline{\bm H}_{i,s-1}, W_{i,s-1}, S_i)
\prod_{s^\prime=1}^{s-1}
d\hat{Q}_{s^\prime}(W_{is^\prime}; \overline{\bm W}_{i,s^\prime-1}, \delta_{s^\prime})
dF(\bR_{is^\prime} \mid \overline{\bm H}_{i,s^\prime-1}, W_{i,s^\prime-1}, S_i)\, dP(S_i \mid S_i \geq s) \Pr(S_i \geq s)\\
= \ &  \sum_{s = 1}^{s_{\max}} \int \frac{\delta_s \{\hat p_s(\overline{\bm W}_{i,s-1}) - p_s(\overline{\bm W}_{i,s-1})\} }{ [\delta_s  \hat p_s(\overline{\bm W}_{i,s-1}) + \{1 - \hat p_s(\overline{\bm W}_{i,s-1})\}]^2 }\\
&\hspace{0.5in} \times \biggl[\frac{\delta_s  \hat p_s(\overline{\bm W}_{i,s-1}) + \{1 - \hat p_s(\overline{\bm W}_{i,s-1})\}}{\delta_s  p_s(\overline{\bm W}_{i,s-1}) + \{1 - p_s(\overline{\bm W}_{i,s-1})\} } \left\{m_{s}(\overline{\bW}_{i,s-1}, 1, \overline{\bm F}_{is}; \bm\delta_{\underline{s}}) \right. \\
      &\hspace{2in} \left. - m_{s}(\overline{\bW}_{i,s-1}, 0, \overline{\bm F}_{is}; \bm\delta_{\underline{s}})\right\} \\
      &\hspace{.75in} - \left\{\hat m_{s}(\overline{\bW}_{i,s-1}, 1, \widehat{\overline{\bm F}}_{is}; \bm\delta_{\underline{s}}) - \hat m_{s}(\overline{\bW}_{i,s-1}, 0, \widehat{\overline{\bm F}}_{is}; \bm\delta_{\underline{s}}) \right\}\biggr]\\
      &\hspace{-0.25in}  \times  dF(\bR_{is} \mid \overline{\bm H}_{i,s-1}, W_{i,s-1}, S_i)
\prod_{s^\prime=1}^{s-1}
d\hat{Q}_{s^\prime}(W_{is^\prime}; \overline{\bm W}_{i,s^\prime-1}, \delta_{s^\prime})
dF(\bR_{is^\prime} \mid \overline{\bm H}_{i,s^\prime-1}, W_{i,s^\prime-1}, S_i)\, dP(S_i \mid S_i \geq s)\Pr(S_i \geq s)\\
= \ & R_{2,2}(\bm\delta) + R_{2,3}(\bm\delta),
\end{align*}
where $R_{2,2}$ and $R_{2,3}$ are defined in
Equations~\eqref{reminder_r22}~and~\eqref{reminder_r23}, respectively.
%\begin{align*}
%    R_{2,2}(\bm\delta) &:=\sum_{s = 1}^{s_{\max}} \int \frac{\delta_s \{ \hat p_s(\overline{\bm W}_{i,s-1}) - p_s(\overline{\bm W}_{i,s-1})\} \biggl( \sum_{w} \{2w-1\} \{ m_{s}({\overline{\bm{H}}}_{is}, w;
%      \tilde{\bm \delta}_{ s+1}) -  \hat m_{s}(\widehat{\overline{\bm{H}}}_{is}, w;
%      \tilde{\bm \delta}_{ s+1})\}  \biggr) }{ [\delta_s  \hat p_s(\overline{\bm W}_{i,s-1}) + \{1 - \hat p_s(\overline{\bm W}_{i,s-1})\}]^2 }\\
 %     &\hspace{0.5in} \times  dF(\bR_{is} \mid \overline{\bR}_{i,s-1}, S_i \geq s)  
%\prod_{s^\prime=1}^{s-1}
%d\hat{Q}_{s^\prime}(W_{is^\prime}; \overline{\bm W}_{i,s^\prime-1}, \delta_{s^\prime})
%dF(\overline{\bm R}_{is^\prime} \mid \overline{\bm R}_{i,s^\prime-1}, S_i \geq  s)\Pr(S_i \geq s)\\
%R_{2,3}(\bm\delta) &:= \sum_{s = 1}^{s_{\max}} \int \frac{\delta_s (\delta_s - 1)  \{ \hat p_s(\overline{\bm W}_{i,s-1}) - p_s(\overline{\bm W}_{i,s-1})\}^2 \{ m_{s}({\overline{\bm{H}}}_{is}, 1;
%      \tilde{\bm \delta}_{ s+1}) - m_{s}({\overline{\bm{H}}}_{is}, 0;
%      \tilde{\bm \delta}_{ s+1})\} }{
%        [\delta_s  \hat p_s(\overline{\bm W}_{i,s-1}) + \{1 - \hat p_s(\overline{\bm W}_{i,s-1})\}]^2
%        [\delta_s   p_s(\overline{\bm W}_{i,s-1}) + \{1 -  p_s(\overline{\bm W}_{i,s-1})\}]
%      }\\
%      &\hspace{0.5in} \times  dF(\bR_{is} \mid \overline{\bR}_{i,s-1}, S_i \geq s) \prod_{s^\prime=1}^{s-1}
%d\hat{Q}_{s^\prime}(W_{is^\prime}; \overline{\bm W}_{i,s^\prime-1}, \delta_{s^\prime})
%dF(\overline{\bm R}_{is^\prime} \mid \overline{\bm R}_{i,s^\prime-1}, S_i \geq s)\Pr(S_i \geq s),
%\end{align*}
This completes the proof.
\end{proof}

\newpage 
\section{Numerical Results of Simulation Studies} \label{app:emp_result}

\begin{table}[!htbp]\centering
\begin{tabular}{l c c c c c}
\toprule
& & & & \multicolumn{2}{c}{95\% Confidence Interval} \\
\cmidrule(lr){5-6}
$N$ & Incremental Parameter $\delta$ & Bias & RMSE & Coverage & Avg. Length \\
\midrule
\multirow{6}{*}{2000} & 0.50 & -0.16 & 0.36 & 0.89 & 1.09 \\
 & 0.75 & -0.09 & 0.23 & 0.92 & 0.74 \\
 & 1.00 & -0.01 & 0.18 & 0.93 & 0.66 \\
 & 1.25 & 0.03 & 0.18 & 0.95 & 0.67 \\
 & 1.50 & 0.05 & 0.20 & 0.92 & 0.73 \\
 & 2.00 & 0.09 & 0.25 & 0.94 & 0.93 \\

\midrule
\multirow{6}{*}{3000} & 0.50 & -0.09 & 0.26 & 0.89 & 0.81 \\
 & 0.75 & -0.05 & 0.17 & 0.90 & 0.57 \\
 & 1.00 & 0.00 & 0.15 & 0.93 & 0.51 \\
 & 1.25 & 0.02 & 0.14 & 0.94 & 0.52 \\
 & 1.50 & 0.03 & 0.16 & 0.95 & 0.56 \\
 & 2.00 & 0.05 & 0.19 & 0.95 & 0.69 \\

\midrule
\multirow{6}{*}{4000} & 0.50 & -0.07 & 0.20 & 0.92 & 0.66 \\
 & 0.75 & -0.05 & 0.14 & 0.94 & 0.48 \\
 & 1.00 & -0.01 & 0.12 & 0.93 & 0.44 \\
 & 1.25 & 0.00 & 0.12 & 0.93 & 0.44 \\
 & 1.50 & 0.01 & 0.14 & 0.90 & 0.47 \\
 & 2.00 & 0.03 & 0.16 & 0.93 & 0.57 \\

\midrule
\multirow{6}{*}{5000} & 0.50 & -0.05 & 0.15 & 0.95 & 0.57 \\
 & 0.75 & -0.03 & 0.11 & 0.95 & 0.42 \\
 & 1.00 & -0.01 & 0.09 & 0.97 & 0.38 \\
 & 1.25 & -0.00 & 0.10 & 0.97 & 0.39 \\
 & 1.50 & -0.00 & 0.11 & 0.95 & 0.41 \\
 & 2.00 & 0.00 & 0.13 & 0.95 & 0.50 \\

\bottomrule
\end{tabular}
\caption{Simulation results with 200 Monte Carlo trials}\label{tab:mc_summary_raw}
\end{table}

\end{document}